\documentclass[11pt,a4paper,twoside]{article}
\usepackage[ngerman,USenglish]{babel} % Different languages:
\usepackage{amsmath}                % AMS Fonts
\usepackage{amsfonts}               % AMS Fonts 
\usepackage{amssymb}                % Mathemat. symbols of AMS
\usepackage{amsopn}                 % Further math. operators like
\allowdisplaybreaks[3] % Page breaks within eqnarray and IEEEeqnarray
\usepackage{amsthm}                % Other choices for Theorem. 
\usepackage{bbm}                    % $\mathbbm{1}$ yields a special ``1''
\usepackage{mathrsfs}               % new mathfonts: \mathscr{ABC}
\usepackage{array}                  % Improved array and
\usepackage{IEEEtrantools}
\usepackage{calc}                   % Calculations within the tex-file
\usepackage{graphicx}        % Further graphic-commands (use
\usepackage{psfrag}                 % Include PS-graphics with LaTeX-
\usepackage{subfigure}              % Put figures into groups
\usepackage{wrapfig}                % Put pictures to the border of
\usepackage{color}           % Coloring of text
\usepackage{fancyhdr}               % Special page layouts
\usepackage{verbatim}               % citations directly included
\usepackage{moreverb}               % even better citations...
\usepackage{exscale}                % Improving on a weakness of
\usepackage{multirow}               % for tables
\usepackage[multiple]{footmisc}

\usepackage{tikz}
\usetikzlibrary{circuits.logic.US,circuits.logic.IEC}
\usetikzlibrary{arrows,positioning} 
\usetikzlibrary{plotmarks}
\usetikzlibrary{decorations.pathreplacing}
\usepackage{float}
\usepackage{wasysym}
\newtheorem{theorem}{Theorem}
\newtheorem{lemma}[theorem]{Lemma}
\newtheorem{corollary}[theorem]{Corollary}
\newtheorem{proposition}[theorem]{Proposition}

\newtheorem{remark}{Remark}
\newtheorem{example}{Example}

\newcommand{\card}[1]{| #1 |}			% cardinality of a set
\newcommand{\setM}{\mathcal{M}}

\newcommand{\setS}{\mathcal{S}}
\newcommand{\setT}{\mathcal{T}}
\newcommand{\setU}{\mathcal{U}}

\newcommand{\setW}{\mathcal{W}}
\newcommand{\setX}{\mathcal{X}}
\newcommand{\setY}{\mathcal{Y}}
\newcommand{\setZ}{\mathcal{Z}}

\newcommand{\rndM}{M}

\newcommand{\rndQ}{Q}

\newcommand{\rndX}{X}

\newcommand{\vecu}{\mathbf{u}}

\newcommand{\vecv}{\mathbf{v}}

\newcommand{\vecx}{\mathbf{x}}

\newcommand{\vecy}{\mathbf{y}}

\DeclareMathOperator{\unif}{Unif}		% uniform distribution
\DeclareMathOperator{\ber}{Ber}			% bernoulli distribution
\newcommand{\prob}[2]{p_{#1} ( #2 )} 	% probability function
\newcommand{\ind}[1]{\mathbbm{1}_{#1}}	% indicator function
\newcommand{\channel}[2]{W ( #1 | #2 )}	% channel transition law
\newcommand{\binent}[1]{h_{\text b} \! \left( #1 \right)} 			%binary entropy
\newcommand{\binentinv}[1]{h_{b}^{-1} \! \left( #1 \right)}		%binary entropy

\def\be{\begin{equation}}
\def\ee{\end{equation}}
\def\een{\nonumber \end{equation}}
\def\mat{\begin{bmatrix}}
\def\emat{\end{bmatrix}}
\def\btm{\begin{textbmatrix}}
\def\etm{\end{textbmatrix}}

\def\ba#1\ea{\begin{align}#1\end{align}}
\def\bas#1\eas{\begin{align*}#1\end{align*}}
\def\bs#1\es{\begin{split}#1\end{split}} 
\def\bg#1\eg{\begin{gather}#1\end{gather}}
\def\bml#1\eml{\begin{multline}#1\end{multline}}
\def\bi#1\ei{\begin{itemize}#1\end{itemize}} 
\def\bipi#1\eipi{\begin{inparaitem}#1\end{inparaitem}}

\newcommand{\ry}{\setY}			%Receiver Y
\newcommand{\dist}{\mathbb{P}}			%distribution
\newcommand{\distof}[1]{\dist \! \left[ #1 \right]}		%distribution with arguments
\newcommand{\Bigdistof}[1]{\dist \! \Bigl[ #1 \Bigr]}		%distribution with arguments
\newcommand{\opE}{\mathbb{E}}
\DeclareMathOperator{\Exop}{\opE}								% expectation operator
\newcommand{\smEx}[2]{\ensuremath{\Exop_{#1} [#2]}} 				% expectation
\newcommand{\bigEx}[2]{\ensuremath{\Exop_{#1} \bigl[#2\bigr]}} 	% expectation
\newcommand{\BigEx}[2]{\ensuremath{\Exop_{#1} \Bigl[#2\Bigr]}} 	% expectation
\newcommand{\reals}{\mathbb R}

\fussy
\large \normalsize
\begin{document}
%
%%%%%%%%%%%%%%%%%%%%%%%%%%%%%%%%%%%%%%%%%%%%%%%%%%%%%%%%%%%%%%%%%%
%% LANGUAGE OF DOCUMENT
%%%%%%%%%%%%%%%%%%%%%%%%%%%%%%%%%%%%%%%%%%%%%%%%%%%%%%%%%%%%%%%%%%
%
\selectlanguage{USenglish}
\pagenumbering{arabic}
%
%%%%%%%%%%%%%%%%%%%%%%%%%%%%%%%%%%%%%%%%%%%%%%%%%%%%%%%%%%%%%%%%%%
\title{Feedback and Partial Message Side-Information on the Semideterministic Broadcast Channel}
  
\author{Annina Bracher and Mich\`ele Wigger}

\maketitle
%%%%%%%%%%%%%%%%%%%%%%%%%%%%%%%%%%%%%%%%%%%%%%%%%%%%%%%%%%%%%%%%%%
%% Sections
%%%%%%%%%%%%%%%%%%%%%%%%%%%%%%%%%%%%%%%%%%%%%%%%%%%%%%%%%%%%%%%%%%

\huge
\begin{abstract}
  \normalsize
  \vspace{0.5cm}
  
  \let\thefootnote\relax\footnotetext{The results in this paper were presented in part at the IEEE International Symposium on Information Theory (ISIT), Hong Kong, China, Jun.\ 2015.}
  \let\thefootnote\relax\footnotetext{A.~Bracher is with Swiss Reinsurance Company Ltd, Mythenquai~50, 8022 Zurich, Switzerland (e-mail: annina\_bracher@swissre.com).}
  \let\thefootnote\relax\footnotetext{M.~Wigger is with LTCI, Telecom ParisTech, Universit\'e Paris-Saclay, 75013 Paris, France (e-mail: michele.wigger@telecom-paristech.fr).}

The capacity of the semideterministic discrete memoryless broadcast channel (SD-BC) with partial message side-information (P-MSI) at the receivers is established. In the setting without a common message, it is shown that P-MSI to the stochastic receiver alone can increase capacity, whereas P-MSI to the deterministic receiver can only increase capacity if also the stochastic receiver has P-MSI. The latter holds only for the setting without a common message: if the encoder also conveys a common message, then P-MSI to the deterministic receiver alone can increase capacity.

These capacity results are used to show that feedback from the stochastic receiver can increase the capacity of the SD-BC without P-MSI and the sum-rate capacity of the SD-BC with P-MSI at the deterministic receiver. The link between P-MSI and feedback is a feedback code, which---roughly speaking---turns feedback into P-MSI at the stochastic receiver and hence helps the stochastic receiver mitigate experienced interference. For the case where the stochastic receiver has full MSI (F-MSI) and can thus fully mitigate experienced interference also in the absence of feedback, it is shown that feedback cannot increase capacity.

\end{abstract}
\normalsize

\section{Introduction} \label{sec:int}

We derive the capacity region of the semideterministic discrete memoryless broadcast channel (SD-BC) with \emph{partial message side-information (P-MSI)} (Theorem~\ref{th:partMessSI}). In this setting each receiver knows part of the message intended for the other receiver already before the transmission begins. Our capacity result generalizes that of \cite{marton79,gelfandpinsker80} for the SD-BC without MSI. The capacity region of the general BC with \emph{full MSI (F-MSI)}, where each receiver knows the entire message intended for the other receiver, was established in~\cite{tuncel06,kramershamai07}. The work of Kramer and Shamai \cite{kramershamai07} also considers P-MSI and establishes the capacity region of the BC with P-MSI and degraded message sets. The three-receiver BC with P-MSI is studied in \cite{oechteringtimorwigger12}. Independently of our work, Asadi, Ong, and Johnson proposed a coding scheme for general two-receiver BCs with P-MSI \cite{asadiongjohnson15}. One can show that---for a judicious choice of the auxiliary random variables---their scheme achieves the capacity region of the SD-BC. Their work does not, however, provide a converse.\footnote{Somewhat similar to P-MSI is decoder cooperation on the BC, which allows the decoders to exchange information via finite-capacity links. The capacity region of the SD-BC with one-sided cooperation via a link from the deterministic to the stochastic receiver is established in \cite{goldfeldpermuterkramer14}. In fact, as it is shown in \cite{goldfeldpermuterkramer14}, this network is operationally equivalent to a class of relay-broadcast channels whose capacity region is established in \cite{liangkramer07}. The physically-degraded BC with parallel conferencing and the BC with conferencing and degraded message sets are studied in \cite{steinberg15}.}
 
Generally speaking, P-MSI reduces the effect of self-interference on the BC and hence enables more efficient communication (see, e.g., \cite{kramershamai07}). More specifically, in the current paper we show that on the SD-BC P-MSI affects the capacity region as follows:
\begin{itemize}
\item P-MSI at the deterministic receiver can increase capacity if, and only if, one of the following two holds: 1) also the stochastic receiver has P-MSI; or 2) the encoder conveys also a common message (Remark~\ref{re:partMessSIRecy}).
\item P-MSI at the stochastic receiver can increase capacity (Remark~\ref{re:MSIRecZ}); and this holds irrespective of whether or not the deterministic receiver has P-MSI or the encoder conveys a common message.
\end{itemize}
%When there also a common message has to be transmitted,
%\begin{itemize}
%\item P-MSI at the deterministic receiver can increase capacity even if the stochastic receiver has No MSI (Remark~\ref{re:MSIRecZ}).
%\end{itemize}
To establish these findings we use our capacity result for the SD-BC with P-MSI. Of particular interest to us is the latter finding, which we shall use to design a feedback code for the SD-BC with or without P-MSI that can improve over the channel's no-feedback capacity.\\

Feedback on the BC was first studied in \cite{gamal78}, where it is shown that even perfect feedback does not increase the capacity region of the physically-degraded BC \cite{gamal78}. It was later proved that feedback can, however, increase the capacity region of several BCs that are not physically degraded \cite{dueck80,ozarowleung84,kramer03,shayevitzwigger13,venkataramananpradhan13,wuwigger14}; and achievable rate regions for the BC with feedback were established in \cite{ozarowleung84,kramer03,shayevitzwigger13,venkataramananpradhan13,wuwigger14}. An intuition for the gain due to feedback is that feedback allows the transmitter to create a common message that is useful to both receivers \cite{dueck80, shayevitzwigger13}. Typically, transmitting one common message is more efficient than transmitting two private messages, because in the latter case the transmissions of the two private messages intefere with each other. In prominent previous examples where feedback increases the BC's capacity---e.g., in Dueck's example \cite{dueck80}---the common message is built up of past noise symbols. It is not clear how the idea of constructing a common message using past noise symbols should be adapted to the SD-BC, which is the focus of this paper; one receiver of the SD-BC is deterministic, and hence it is not clear why information that is constructed only from previous noise symbols should be useful to this receiver.

In the current paper we show that---not withstanding the above observations---feedback can increase the capacity region of the SD-BC (Theorem~\ref{th:fbCanIncCapSDBC}). To establish this result, we use the feedback to create an improved situation where the stochastic receiver has P-MSI. (As mentioned before, P-MSI at the stochastic receiver can increase the SD-BC's capacity.) More precisely, we use the idea that the encoder can create from the feedback a new message that is useful to the deterministic receiver and can be created (and hence is known) at the stochastic receiver. A similar idea was previously used by Wu and Wigger to construct a coding scheme for the general BC with rate-limited feedback \cite{wuwigger14}, though their work does not make the connection to the BC with P-MSI explicit. They use the scheme to show that feedback can increase the capacity of a large class of stochastically- (but not physically-) degraded BCs as well as the capacity of a class of BCs that consist of a binary symmetric channel and a binary erasure channel \cite{wuwigger14}. The argument is particularly intuitive for the class of stochastically-degraded BCs that satisfy that one receiver is stronger than the other: it is shown in \cite{wuwigger14} that for any BC in this class the encoder can create a new message that is useful to the stronger receiver and can be created at the weaker receiver, and that the encoder can send this message without reducing the rates at which the fresh message-information is transmitted.
 
Unlike the above class of stochastically-degraded BCs, on the SD-BC there is a tradeoff between the rates at which fresh message-information and the message that the encoder constructs from the feedback are sent. Hence---even with the results of \cite{wuwigger14} at hand---showing that feedback can increase the capacity region of the SD-BC is nontrivial. We show by means of an example that---with a judicious choice of the rates at which the fresh and the feedback information are sent---we can increase the overall rates at which the messages are sent to the receivers (Example~\ref{ex:addEras}). From this we conclude that feedback can increase the capacity region of the SD-BC.
 
As already mentioned, in \cite{wuwigger14} the connection between the coding idea and the BC with P-MSI is not made explicit. We make the connection explicit, and this allows us to readily extend our feedback coding scheme for the SD-BC to the case where the receivers have P-MSI.
Using this extension of the feedback code, we show that if the deterministic receiver has P-MSI, then feedback can increase the sum-rate capacity of the SD-BC (Theorem~\ref{th:fbCanIncCapSDBC}). For the case where the stochastic receiver has F-MSI, we show that feedback cannot increase capacity, irrespective of whether or not the deterministic receiver has P-MSI (Theorem~\ref{th:fullMessSIRecZ}).\\

The rest of this paper is structured as follows. We conclude this section by introducing some notation. Section~\ref{sec:chMod} describes the channel model. Section~\ref{sec:pmsi} contains the results for the SD-BC with P-MSI, and Section~\ref{sec:fb} studies the effect of feedback on the SD-BC with and without P-MSI. %The paper is concluded by a brief summary.

\subsection{Notation and Preliminaries}

We use calligraphic letters to denote finite sets and $\card \cdot$ for their cardinality, e.g., $\setX$ and $\card{ \setX }$. %By $\mathscr C$ we denote a capacity region and by $\partial \mathscr C$ its boundary.
Random variables are denoted by upper-case letters and their realizations by lower-case letters, e.g., $\rndX$ and $x$. By $X_i^j$ and $x_i^j$ we denote the tuples $(X_i,\ldots, X_j)$ and $(x_i,\ldots, x_j)$, where $j>i$; and we drop the subscript $i = 1$, e.g., we write $X^n$ instead of $X_1^n$. Sequences are in bold lower- or upper-case letters depending on whether they are deterministic or random, e.g., $\vecx$ denotes an $n$-length codeword.

%Sequences of variables that occur in the time-range $j$ to $i$ bear a subscript $j$ and a superscript $i$, where the subscript $j = 1$ may be dropped, e.g., $X_4^5$ denotes the fourth and fifth channel input, and $X^n$ denotes all the inputs through Time~$n$.

%We use bold font for sequences or random variables or realizations. % Most of these sequences are of length% variables or numbers are  % denotes the stochastic receiver's random channel output and $z \in \setZ$ a value it may take. %Sequences are in bold lower- or upper-case letters depending on whether they are deterministic or random, e.g., $\vecx ( m_\setY, m_\setZ )$ is the codeword corresponding to the message-pair $( m_\setY, m_\setZ )$. The positive integer $n \in \naturals$ stands for the blocklength, and, unless otherwise specified, sequences are of length $n$.

By $\rndQ \sim \unif [ 1 : n ]$ we indicate that the random variable $Q$ is uniformly drawn from the set $\{1,\ldots, n\}$, and by $S \sim \ber (p)$, where $p \in [0,1]$, we indicate that $S$ is a Bernoulli-$p$ random variable. We denote the binary entropy function by $h_\textnormal{b} (\cdot)$ and its inverse on $[0,1/2]$ by $h_\textnormal{b}^{-1} (\cdot)$. %, i.e., $h_\textnormal{b} (p)$ is the entropy of the binary random variable $S$, and for $p \in [0,1/2]$ we have $h_\textnormal{b}^{-1} \bigl( h_\textnormal{b} (p) \bigr) = p$.

% $2^{n R}$ is an  accurately, one should write $\lfloor 2^{n R} \rfloor$. In the regime of interest where $n \rightarrow \infty$, the difference between the two quantities is, however, negligible, and we therefore ignore this technicality. Similarly, we shall assume that, for every positive real rumbers $R$, $R_a$, and $R_b$ satisfying $R = R_a + R_b$, every set of size $2^{n R}$ can be written as the Cartesian product of any two sets of size $2^{n R_a}$ and $2^{n R_b}$, respectively.
%

A joint probability mass function (PMF), its marginal PMF, and its conditional PMF are all denoted by the same function $\prob {} \cdot$, with the exact meaning specified by the subscripts or arguments, e.g., $\prob {X, Y}{0,1}$ denotes the probability of the event $( X, Y ) = ( 0, 1 )$ and $\prob {}{ x | y }$ the probability that $X = x$ given $Y = y$.

We denote the set of $\epsilon$-typical length-$n$ sequences defined in  \cite[Chapter~2]{gamalkim11} by $\setT^{( n )}_{\epsilon}$. By $\delta (\epsilon)$ we denote any function of $\epsilon$ that converges to $0$ as $\epsilon$ approaches $0$; and $\{ \epsilon_n \}$ can stand for any sequence of numbers that converges to $0$ as $n$ tends to infinity.\\

We shall use the following lemma, which is proved, e.g., in \cite{willemsmeulen85}:

\begin{lemma}[Functional Representation lemma] \label{le:funcRep}
Given two random variables $X$ and $Y$ of finite support, there exist a chance variable $S$ of finite support $\setS$ that is independent of $X$ and a function $g \colon \setX \times \setS \rightarrow \setY$ such that $Y = g ( X, S )$.
\end{lemma}

\section{Channel Model} \label{sec:chMod}

We consider the SD-BC of transition law $$W (y,z|x) = \ind {\{ y = f (x) \}} \, W (z|x),$$ where we assume that the channel-input alphabet $\setX$ and the channel-output alphabets $\setY$ and $\setZ$ are finite. Transmitting an $n$-tuple $X^n$, the encoder wants to convey the message-pairs $(M,M_\setY)$ and $(M,M_\setZ)$ to the deterministic receiver~$\setY$ and the stochastic receiver~$\setZ$, respectively, where $M$ denotes the common message and $M_\setY$ and $M_\setZ$ the private messages. We assume that $M$, $M_\setY$, and $M_\setZ$ are independent, that $M$ is uniformly drawn from a size-$2^{n R}$ set, and that for each $\nu \in \{ \setY, \setZ \}$ message $M_\nu$ is uniformly drawn from a size-$2^{n R_\nu}$ set. We study the SD-BC with P-MSI, and we thus assume that each private message comprises two parts, i.e., 
$$M_\nu = \big( M_\nu^{(p)} \!, M_\nu^{(c)} \big), \quad \nu \in \{ \setY, \setZ \},$$ that Receiver~$\setY$ knows $M_\setZ^{(c)}$ and decodes $(M, M_\setY)$ from $( Y^n, M_\setZ^{(c )} )$, and that Receiver~$\setZ$ knows $M_\setY^{(c)}$ and decodes $(M, M_\setZ)$ from $( Z^n, M_\setY^{(c )} )$. For each $\nu \in \{ \setY, \setZ \}$ we assume that $M_\nu^{(p)}$ and $M_\nu^{(c)}$ are independent of each other and uniformly drawn from sets of size\footnote{For simplicity, when we write $2^{nR}$, for some $n,R\geq 0$, we implicitly assume that it is an integer value. It would be more precise to write $\lfloor 2^{nR}\rfloor$ instead. However, the ratio between the two expressions tends to 1 when $n\to \infty$, which is the regime of interest in this paper.}  $2^{n R_\nu^{(p)}}$ and $2^{n R_\nu^{(c)}}$, respectively, where $$R_\setY= R_\setY^{(p)} \! + R_\setY^{(c)}\qquad \text{and} \qquad R_\setZ= R_\setZ^{(p)} \! + R_\setZ^{(c)}.$$ Note the extreme cases:
\begin{itemize}
\item $( R_\setY^{(p)} \!, R_\setY^{(c)} ) = ( R_\setY, 0 ) \quad \Longleftrightarrow \quad \text{no MSI at Receiver~}\setZ$
\item $( R_\setY^{(p)} \!, R_\setY^{(c)} ) = ( 0,R_\setY ) \quad \Longleftrightarrow \quad \text{F-MSI at Receiver~}\setZ$
\item $( R_\setZ^{(p)} \!, R_\setZ^{(c)} ) = ( R_\setZ, 0 ) \quad \Longleftrightarrow \quad \text{no MSI at Receiver~}\setY$
\item $( R_\setZ^{(p)} \!, R_\setZ^{(c)} ) = ( 0,R_\setZ ) \quad \Longleftrightarrow \quad \text{F-MSI at Receiver~}\setY$.
\end{itemize}
A rate-tuple $( R, R_\setY^{(p)} \!, R_\setY^{(c)} \!, R_\setZ^{(p)} \!, R_\setZ^{(c)} )$ is achievable if there exists a sequence of encoders and decoders so that at each receiver the probability of a decoding error tends to zero as $n$ tends to infinity. The capacity region 
is the closure of the set of all achievable rate-tuples.

We study the SD-BC with P-MSI in the absence and in the presence of feedback. In the absence of feedback, the encoder selects the channel-input sequence as a function of the triple $(M, M_\setY, M_\setZ)$, i.e., $X^n = \phi (M, M_\setY, M_\setZ)$. This setting corresponds to that of Figure~\ref{fig:model} without the dashed links. We denote its capacity region by $\mathscr C_{\textnormal {P-MSI}}$, and in the special case without MSI by $\mathscr C$.

When there is feedback, it is assumed to be one-sided from the stochastic receiver~$\setZ$ only. (Feedback from the deterministic receiver~$\setY$ is useless, because the encoder can always compute $Y^n$ from $X^n$.) We consider perfect and rate-limited feedback. Perfect feedback allows the encoder to form the Time-$i$ input also as a function of $Z^{i-1}$, i.e., $$X_i = \phi_i (M, M_\setY, M_\setZ, Z^{i-1}), \quad i \in [1 : n].$$ Rate-limited feedback of rate $R_{\textnormal {FB}}$ allows Receiver~$\setZ$ to transmit after Transmission~$i$ a feedback signal $W_i = h ( Z^{i}, M_\setY^{(c)} ) \in \setW_i$ to the encoder, and in turn the encoder can form the Time-$i$ input also as a function of $W^{i-1}$, i.e., $$X_i = \phi_i (M, M_\setY, M_\setZ, W^{i-1}), \quad i \in [1 : n].$$ The rate-limitation implies that
\ba
\prod^n_{i=1} |\setW_i| \leq 2^{n R_{\textnormal {FB}}}. \label{eq:rateLimitation}
\ea
The SD-BC with P-MSI and perfect feedback (rate-limited feedback) corresponds to the setting of Figure~\ref{fig:model} when the dashed links transport the feedback signal $Z_i$ ($W_i$). Note that perfect feedback is more powerful than rate-limited feedback: any rate-tuple that is achievable with rate-limited feedback can also be achieved with perfect feedback. Rate-limited and perfect feedback are equally powerful when $R_{\textnormal{FB}} \geq \log |\setZ|$.

\begin{figure}[ht]
\vspace{-2mm}

\begin{center}
\def\pgfsysdriver{pgfsys-dvipdfm.def}
\begin{tikzpicture}[circuit logic US]
% Define a few styles and constants
	\tikzstyle{sensor}=[draw, minimum width=4em, text centered, minimum height=2em]
	\tikzstyle{stategen}=[draw, text width=2.2em, text centered, minimum height=2em]
	\tikzstyle{delay}=[draw, text width=0.5em, text centered, minimum height=1em]
	\tikzstyle{naveqs} = [sensor, minimum width=4em, minimum height=5.5em]
	\def\blockdist{2.4}
	\def\edgedist{2.5}
    \node (naveq) [naveqs] {$\channel {y,z}{x}$};
    \path (naveq.west)+(-0.7*\blockdist,0) node (enc) [sensor] {Encoder};
    \path (naveq.east)+(0.8*\blockdist,0.3*\blockdist) node (dec1) [sensor] {Rec.~$\setY$};
    \path (naveq.east)+(0.8*\blockdist,-0.3*\blockdist) node (dec2) [sensor] {Rec.~$\setZ$};
    \path (dec1.east)+(0.3*\blockdist,0) node (guess1) [text centered] {};
    \path (dec2.east)+(0.3*\blockdist,0) node (guess2) [text centered] {};
    \path (dec1)+(0,0.5*\blockdist) node (msi1) [text centered] {};
    \path (dec2)+(0,-0.5*\blockdist) node (msi2) [text centered] {};     
    \path (enc.west)+(-0.3*\blockdist,0) node (sour) [text centered] {};
    \path (dec2)+(-0.3*\blockdist,-0.45*\blockdist) node (fb1) [text centered] {};
    \path (fb1.north)+(-1.3*\blockdist,0) node (del) [delay] {\small D};
    		
    \path [draw, ->] (enc) -- node [above] {$X_{i}$} 
        (naveq.west |- enc);        
    \path [draw, ->] (sour) -- node [left] {$\begin{matrix}M_\setY\\ M \\ M_\setZ\end{matrix}$ \quad}
    		(enc.west |- sour);
    	\path [draw, <-] (dec1) -- node [above] {$Y_i$}
    		(naveq.east |- dec1);
    \path [draw, <-] (dec2) -- node [above] {$Z_i$}
    		(naveq.east |- dec2);
    	\path [draw, <-] (dec1) -- node [right] {$M^{(c)}_\setZ$} 
        (msi1);
    	\path [draw, <-] (dec2) -- node [right] {$M^{(c)}_\setY$} 
        (msi2);    		
    	\path [draw, ->] (dec1) -- node [right] { \; $\begin{matrix}\widehat{M} \\ \widehat{M_\setY}\end{matrix}$} 
        (guess1);
    \path [draw, ->] (dec2) -- node [right] { \; $\begin{matrix}\widehat{M}\\\widehat{M_\setZ}\end{matrix}$} 
        (guess2);
    \path [draw, dashed] (fb1) -- node [above] {} (fb1.north |- dec2.south);
    \path [draw, dashed, ->] (fb1.north) -- node [below] {$Z_i$ / $W_i$} (del);
%    \path [draw, dashed] (del.west) -- node [above] {$Z^{i-1}$ / $W^{i-1}$} (enc.south |- del.west);
	\path [draw, dashed] (del.west) -- (enc.south |- del.west);
    \path [draw, dashed, <-] (enc.south) -- node [left] {} (enc.south |- del.west);
    
\end{tikzpicture}

\caption[SD-BC with P-MSI and feedback]{SD-BC with P-MSI and feedback.}
\label{fig:model}
\end{center}
\vspace{-2mm}

\end{figure}
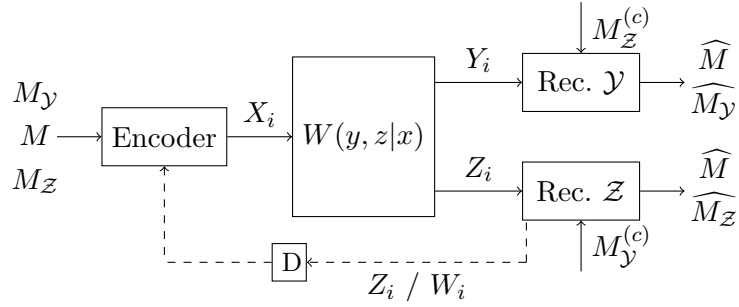

\section{The SD-BC with P-MSI} \label{sec:pmsi}

In this section, we assume that there is no feedback. 

\subsection{Capacity Region and Optimal Coding Scheme} \label{sec:pmsiCapResults}
%Our main result is a single-letter characterization of the capacity region of the SD-BC with P-MSI:

\begin{theorem}[Capacity with P-MSI]\label{th:partMessSI}
The capacity region $\mathscr C_{\textnormal {P-MSI}}$ of the SD-BC with P-MSI is the set of rate-tuples $(R,R_\setY^{(p)} \!, R_\setY^{(c)} \!, R_\setZ^{(p)} \!, R_\setZ^{(c)})$ satisfying
\begin{subequations} \label{bl:capRegPartMessSI}
\ba
R+R_\setY &\leq H (Y) \label{eq:ry} \\
R+R_\setZ &\leq I (U;Z) \label{eq:rz} \\
R+R_\setY + R_\setZ^{(p)} \! &\leq I (V;Y) + H (Y|U) + I (U;Z|V) \label{eq:ryRzp} \\
R+R_\setY^{(p)} \! + R_\setZ &\leq H (Y|U) + I (U;Z) \label{eq:rypRz} \\
2 R+R_\setY + R_\setZ &\leq I (V;Y) + H (Y|U) + I (U;Z)\label{eq:Rlast}
\ea
\end{subequations}
for some PMF of the form
\ba
p (v,u,x,y,z) = p (v,u) \, p (x|u) \, W (y,z|x). \label{eq:partMessSIPMF}
\ea
W.l.g., one can restrict $X$ to be a function of $(Y,U)$.
\end{theorem}

\begin{proof}
See Appendix~\ref{sec:thPartMessSI}.
\end{proof}

In the following we sketch and discuss the proof of the direct part. The capacity-achieving code is described rigorously in Appendix~\ref{sec:thPartMessSI}. %Here, we try to provide some intuition. A tentative code construction is as follows. 
Use Marton's code construction (see \cite[Section~8.4]{gamalkim11}) to encode the ``common message-tuple'' $\big(M, M_\setY^{(c)} \!, M_\setZ^{(c)}\big)$ into a cloud-center $V^n$ and the private messages $M_\setY^{(p)}$ and $M_\setZ^{(p)}$ into satellites $Y^n$ and $U^n$, respectively. Receiver~$\setY$ decodes $\big(M, M_\setY^{(c)} \!, M_\setZ^{(c)}\big)$ and $M_\setY^{(p)}$ jointly, while taking into account its knowledge of $M_\setZ^{(c)}$; and likewise Receiver~$\setZ$ decodes $\big(M, M_\setY^{(c)} \!, M_\setZ^{(c)}\big)$ and $M_\setZ^{(p)}$ jointly, while taking into account its knowledge of $M_\setY^{(c)}$.

The tentative code can achieve all rate-tuples $(R, R_\setY^{(p)} \!, R_\setY^{(c)} \!, R_\setZ^{(p)} \!, R_\setZ^{(c)})$ that for some PMF of the form \eqref{eq:partMessSIPMF} satisfy \eqref{bl:capRegPartMessSI} and
\begin{subequations}\label{bl:ccNotBinned}
\ba
R_\setY^{(p)} \! &\leq H (Y|V) \\
R_\setZ^{(p)} \! &\leq I (U;Z|V) \label{eq:ccNotBinnedRZp} \\
R_\setY^{(p)} \! + R_\setZ^{(p)} \! &\leq H (Y|U) + I (U;Z|V).
\ea
\end{subequations}
Note that this achievable region differs from the capacity region $\mathscr C_{\textnormal {P-MSI}}$ of the SD-BC with P-MSI in that the rates $R_\setY^{(p)}$ and $R_\setZ^{(p)}$ must also satisfy \eqref{bl:ccNotBinned}. As we show in Appendix~\ref{sec:binCCNeeded}, the region is---in general---strictly contained in $\mathscr C_{\textnormal {P-MSI}}$.\footnote{To show this, we shall use Corollary~\ref{co:fullMessSIRecZCommMess} ahead.} 

To get rid of the constraints \eqref{bl:ccNotBinned} and hence achieve the entire capacity region $\mathscr C_{\textnormal {P-MSI}}$, a fix is needed: the encoder must be able to convey more information about $M_\setY^{(p)}$ and $M_\setZ^{(p)}$ by allowing the cloud-center $V^n$ to depend not only on the triple $(M, M_\setY^{(c)} \!, M_\setZ^{(c)})$ but also on $M_\setY^{(p)}$ and $M_\setZ^{(p)}$. To this end we use the following code construction, which is depicted in Figure~\ref{fig:codeconstruction} for the setting without a common message. (In Figure~\ref{fig:codeconstruction} each dot represents an $n$-length codeword.)
\begin{figure}[ht!]
\includegraphics[width=1.1\textwidth]{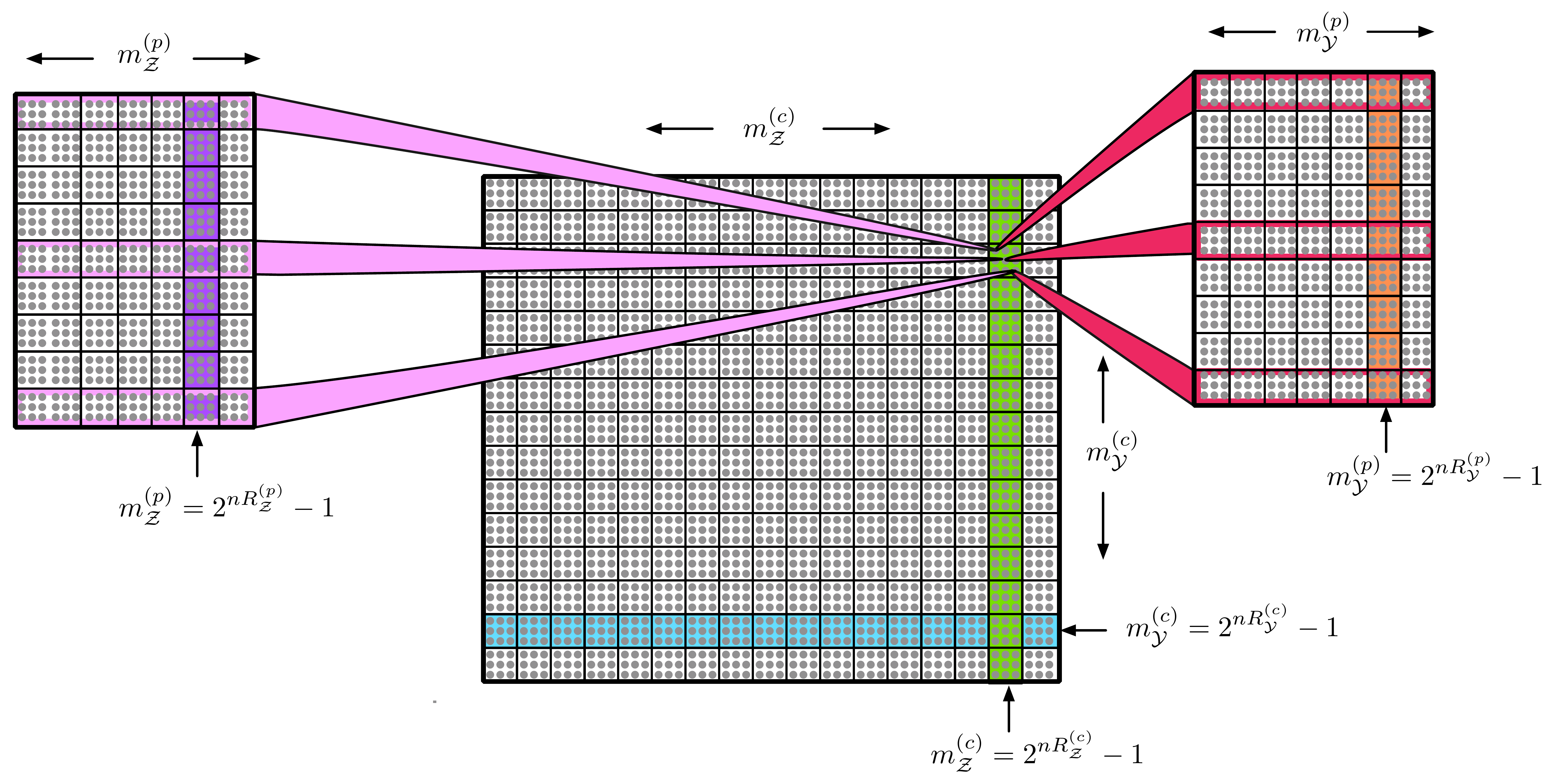}
\label{fig:codeconstruction}
\caption{Code construction without a common message.}
\end{figure}
Fix some PMF $p (v,u,x,y,z)$ of the form \eqref{eq:partMessSIPMF}. For each triple $(m, m_\setY^{(c)} \!, m_\setZ^{(c)})$ we generate a bin containing $2^{n \tilde{R}^{(c)}}$ $n$-tuples $\vecv$, which are drawn independently of each other and each from the PMF $\prod^{n}_{i = 1} p (v_i)$. (In Figure~\ref{fig:codeconstruction} the light-blue row represents all bins and codewords that are associated with $m_{\setY}^{(c)} \! = 2^{n R_{\setY}^{(c)}}-1$, and the light-green column represents all bins and codewords that are associated with $m_{\setZ}^{(c)} \! = 2^{n R_{\setZ}^{(c)}}-1$.) For each cloud-center-bin we generate two satellite codebooks: one to encode $M_{\setY}^{(p)}$ and one to encode $M_{\setZ}^{(p)}$. (Figure~\ref{fig:codeconstruction} depicts the satellite codebooks corresponding to the pair $( m_{\setY}^{(c)} \!, m_{\setZ}^{(c)} ) = (3, 2^{n R_{\setZ}^{(c)}}-1)$: that for $M_{\setY}^{(p)}$ on the right, and that for $M_{\setZ}^{(p)}$ on the left.) For each $m_{\setY}^{(p)}$ the first $2^{n(\tilde{R}_\setY- \tilde{R}^{(c)})}$ $\vecy$-codewords in the satellite codebook corresponding to any pair $( m_{\setY}^{(c)} \!, m_{\setZ}^{(c)} )$ are superpositioned on the first codeword in the corresponding cloud-center-bin; the following $2^{n(\tilde{R}_\setY- \tilde{R}^{(c)})}$ codewords in each satellite codebook are superpositioned on the second codeword in the corresponding cloud-center-bin; and so on. That is, the first $2^{n(\tilde{R}_\setY- \tilde{R}^{(c)})}$ $\vecy$-codewords are drawn according to the conditional PMF $\prod^{n}_{i = 1} p (y_i|v_i^{(1)})$, where $v_i^{(1)}$ denotes the $i$-th component of the first codeword in the corresponding clound-center-bin; the following $2^{n(\tilde{R}_\setY- \tilde{R}^{(c)})}$ $\vecy$-codewords are drawn according to the conditional PMF $\prod^{n}_{i = 1} p (y_i|v_i^{(2)})$, where $v_i^{(2)}$ denotes the $i$-th component of the second codeword in the corresponding clound-center-bin; and so on. The $\vecu$-codewords in the satellite codebooks for $M_{\setZ}^{(p)}$ are drawn similarly. (In Figure~\ref{fig:codeconstruction} the uppermost row on the right, which is framed in red, and the uppermost row on the left, which is framed in lila, correspond to the first codeword in the $(3, 2^{n R_{\setZ}^{(c)}}-1)$-cloud-center-bin.)
%
%all codewords in the satellite codebook for $M_{\setY}^{(p)}$ are drawn independently of each other according to a product distribution, where for each $n$-tuple $\vecy$ we use the conditional PMF $\prod^{n}_{i = 1} p (y_i|v_i)$ for $v_i$ the $i$-th component of the corresponding codeword $\vecv$ in the cloud-center.
The codewords in each satellite codebook are partitioned into as many different bins as there are possible realizations of the message $M_{\setY}^{(p)}$ or $M_{\setZ}^{(p)}$, respectively, and each such bin is associated with a different realization $m_{\setY}^{(p)}$ or $m_{\setZ}^{(p)}$, respectively. (In Figure~3 the orange column represents the bin that is associated with $m_{\setY}^{(p)}= 2^{n R_{\setY}^{(p)}}-1$, and the purple column represents the bin that is associated with $m_{\setZ}^{(p)}= 2^{n R_{\setZ}^{(p)}}-1$. Note that each bin comprises multiple subbins: one for each $(3, 2^{n R_{\setZ}^{(c)}}-1)$-cloud-center-codeword, where $( m_{\setY}^{(c)} \!, m_{\setZ}^{(c)} ) = (3, 2^{n R_{\setZ}^{(c)}}-1)$ is the ``common message-pair'' to which the depicted satellite codebooks correspond.)

%(In the figure these correspond to the first and second triplets of codeword-rows of the satellite codebooks. to Conditional on $\vecv$, we generate for each $m_\setY^{(p)} \! \in \setM_\setY^{(p)}$ a bin containing $n$-tuples $\vecy$, which are drawn independently of each other and each from the PMF $\prod^{n}_{i = 1} p (y_i|v_i)$, and likewise we generate for each $m_\setZ^{(p)} \! \in \setM_\setZ^{(p)}$ a bin containing $n$-tuples $\vecu$. }

To transmit the message-tuple $(m, m_\setY^{(p)} \!, m_\setY^{(c)} \!, m_\setZ^{(p)} \!, m_\setZ^{(c)})$, the encoder first looks for $n$-tuples $\vecv$, $\vecy$, and $\vecu$ in the bins corresponding to $(m, m_\setY^{(c)} \!, m_\setZ^{(c)})$, $m_\setY^{(p)}$, and $m_\setZ^{(p)}$, respectively, satisfying that $(\vecv,\vecy,\vecu)$ are jointly typical. It then generates the channel input $X^n$ from the product distribution $\prod^{n}_{i = 1} p (x_i|y_i,u_i)$. Receiver~$\setY$ decodes $(M, M_\setY^{(c)} \!, M_\setZ^{(c)})$ and $M_\setY^{(p)}$ jointly, while restricting attention to the column of the cloud-center that corresponds to the message $M_\setZ^{(c)}$, which Receiver~$\setY$ knows. Likewise, Receiver~$\setZ$ decodes $(M, M_\setY^{(c)} \!, M_\setZ^{(c)})$ and $M_\setZ^{(p)}$ jointly, while restricting attention to the row of the cloud-center that corresponds to the message $M_\setY^{(c)}$, which Receiver~$\setZ$ knows.\footnote{It is well-known that, without a cloud-center, this scheme achieves the capacity region of the SD-BC without a common message and without MSI (see, e.g., \cite[Sections~8.3.1--8.3.2]{gamalkim11}).}

%The encoder then looks for $n$-tuples $\vecv$, $\vecy$, and $\vecu$ such that $(\vecv,\vecy,\vecu)$ are jointly typical. 

As we explain in Appendix~\ref{sec:thPartMessSI}, the effect of binning the cloud-center is the same as that of rate-splitting (see Remark~\ref{re:binningCCEqRS} ahead). If we were to use rate-splitting instead of binning, then we would not bin the cloud-center, but instead we would divide the messages $M_\setY^{(p)}$ and $M_\setZ^{(p)}$ into two parts each, i.e., $$M_\nu^{(p)} \! = ( M_{\nu,s}^{(p)}, M_{\nu,c}^{(p)} ), \quad \nu \in \{ \setY, \setZ \}.$$ Of these parts we would associate only $M_{\setY,s}^{(p)}$ and $M_{\setZ,s}^{(p)}$ with the satellites $Y^n$ and $U^n$, respectively, whereas we would encode $M_{\setY,c}^{(p)}$ and $M_{\setZ,c}^{(p)}$ in the cloud-center. The benefit of binning the cloud-center is that it necessitates only one auxiliary rate: the rate at which the cloud-center is binned. In contrast, rate-splitting necessitates two auxiliary rates: the rates $R_{\setY,c}^{(p)}$ and $R_{\setZ,c}^{(p)}$ associated with $M_{\setY,c}^{(p)}$ and $M_{\setZ,c}^{(p)}$, respectively.\\

We next specialize Theorem~\ref{th:partMessSI} to cases where one or both of the receivers have no MSI or F-MSI, Table~\ref{tb:specialCasesPMSI} illustrates to which.\footnote{Corollary~\ref{co:noMessSI} ahead is only for the setting without a common message ($R = 0$).} These corollaries help understand when P-MSI increases capacity (see Subsection~\ref{sec:effectPMSI}).

\begin{table}[h]
\begin{center}
  \renewcommand{\arraystretch}{1.2}
  \begin{tabular}{ c c | c c c }
    \multicolumn{2}{c}{\multirow{2}{*}{}} & \multicolumn{3}{c}{Receiver~$\setY$} \\
  	&& no MSI & P-MSI & F-MSI \\
    \cline{2-5}    
    \multirow{3}{*}{\rotatebox[origin=c]{90}{Receiver~$\setZ$}} & no MSI & Corollary~\ref{co:noMessSI} & Corollary~\ref{co:noMessSI} & Corollary~\ref{co:noMessSI} \\
    & P-MSI & -- & -- & Corollary~\ref{co:fullMessSIRecYCommMess} \\
    & F-MSI & -- & Corollary~\ref{co:fullMessSIRecZCommMess} & Corollary~\ref{co:fullMessSIRecYZCommMess}
  \end{tabular}
  \caption{Special cases of Theorem~\ref{th:partMessSI}.}
  \label{tb:specialCasesPMSI}			
\end{center}
\end{table}

%\mw{
%Without a common message and without MSI \mw{at the stochastic receiver $\setZ$}, i.e., when $$R = 0 %, \quad R_\setY^{(p)} \!= R_\setY, 
%\quad \text{and} \quad R_\setY^{(p)} \!= R_\setY,$$ we recover: %Constraints~\eqref{eq:ryRzp} and \eqref{eq:Rlast} are redundant. Consequently, we recover:

\begin{corollary}[No MSI at $\setZ$] \label{co:noMessSI}%[From \cite{gelfandpinsker80, marton79}, see also Section~8.3.1 of \cite{gamalkim11}]
For each $R_\setZ^{(p)} \! \in [0, R_\setZ]$, the capacity region of the SD-BC without a common message and without MSI at Receiver~$\setZ$ ($R=0$ and $R_{\setY}^{(p)}=R_{\setY}$) is the set of rate-tuples $(0, R_\setY, 0, R_\setZ^{(p)}, R_\setZ^{(c)})$ satisfying
\begin{subequations}\label{bl:capRegNoMessSI}
\ba
R_\setY &\leq H (Y) \label{eq:ryCapRegNoMessSI} \\
R_\setZ &\leq I (U;Z) \label{eq:rzCapRegNoMessSI} \\
R_\setY + R_\setZ &\leq H (Y|U) + I (U;Z) \label{eq:srCapRegNoMessSI}
\ea
\end{subequations}
for some PMF of the form
\ba
p (u,x,y,z) = p (u,x) \, W (y,z|x). \label{eq:noMessSIPMF}
\ea
W.l.g., one can restrict $X$ to be a function of $(Y,U)$.
\end{corollary} 

\begin{proof}
Let $V$ be deterministic. In this case, and because $R=0$ and $R_\setY^{(p)} \!= R_\setY$, Constraints~\eqref{eq:ryRzp} and \eqref{eq:Rlast} are redundant in view of Constraint~\eqref{eq:rypRz}. % redundant \eqref{bl:capRegPartMessSI} and \eqref{bl:capRegNoMessSI} are equivalent. 
Hence, \eqref{bl:capRegNoMessSI} is an inner bound on the capacity region. That \eqref{bl:capRegNoMessSI} is also an outer bound follows from \eqref{eq:ry}, \eqref{eq:rz}, and \eqref{eq:rypRz}.
\end{proof}

%From With F-MSI at the deterministic receiver~$\setY$, i.e., when $R_\setZ^{(p)} = 0$, Theorem~\ref{th:partMessSI} gives:

\begin{corollary}[F-MSI at $\setY$]\label{co:fullMessSIRecYCommMess}
For each $R_\setY^{(p)} \! \in [0, R_\setY]$, the capacity region of the SD-BC with F-MSI at Receiver~$\setY$ ($R_\setZ^{(p)} \!= 0$) is the set of rate-tuples $(R, R_\setY^{(p)} \!, R_\setY^{(c)} \!, 0, R_\setZ)$ satisfying
\begin{subequations}\label{bl:capRegFullMessSIRecYCommMess}
\ba
R + R_\setY &\leq H (Y) \label{eq:fullMessSIRecYRyCommMess} \\
R + R_\setZ &\leq I (U;Z) \label{eq:fullMessSIRecYRzCommMess} \\
R + R_\setY^{(p)} + R_\setZ &\leq H (Y|U) + I (U;Z) \label{eq:fullMessSIRecYRypRzCommMess}
\ea
\end{subequations}
for some PMF of the form
\ba
p (u,x,y,z) = p (u,x) \, W (y,z|x). \label{eq:fullMessSIRecYCommMess}
\ea
\end{corollary}

\begin{proof}
For $V = U$ the constraints in \eqref{bl:capRegPartMessSI} and \eqref{bl:capRegFullMessSIRecYCommMess} are equivalent, and hence \eqref{bl:capRegFullMessSIRecYCommMess} is an inner bound on the capacity region. That \eqref{bl:capRegFullMessSIRecYCommMess} is also an outer bound follows from \eqref{eq:ry}, \eqref{eq:rz}, and \eqref{eq:rypRz}. 
\end{proof}

%With F-MSI at the stochastic receiver~$\setZ$, i.e., when $R_\setY^{(p)} = 0$, Theorem~\ref{th:partMessSI} gives:

\begin{corollary}[F-MSI at $\setZ$] \label{co:fullMessSIRecZCommMess}
For each $R_\setZ^{(p)} \! \in [0, R_\setZ]$, the capacity region of the SD-BC with F-MSI at Receiver~$\setZ$ ($R_\setY^{(p)} \!= 0$) is the set of rate-tuples $(R, 0, R_\setY, R_\setZ^{(p)} \!, R_\setZ^{(c)})$ satisfying
\begin{subequations} \label{bl:capRegFullMessSIRecZCommMess}
\ba
R + R_\setY &\leq H (Y) \label{eq:fullMessSIRecZRyCommMess} \\
R + R_\setZ &\leq I (X;Z) \label{eq:fullMessSIRecZRzCommMess} \\
R + R_\setY + R_\setZ^{(p)} \! &\leq I (X;Y,Z) \label{eq:fullMessSIRecZRyRzpCommMess}
\ea
\end{subequations}
for some PMF of the form
\ba
p (x,y,z) = p (x) \, W (y,z|x). \label{eq:fullMessSIRecZCommMess}
\ea
\end{corollary}

\begin{proof}
See Appendix~\ref{sec:pfCoFullMessSIRecZ}.
\end{proof}

From Corollary~\ref{co:fullMessSIRecZCommMess} we see that if both receivers have F-MSI, i.e., when $$R_\setY^{(p)} = 0 \quad \text{and} \quad R_\setZ^{(p)} = 0,$$ then \eqref{eq:fullMessSIRecZRyRzpCommMess} is redundant. Consequently, we recover:

\begin{corollary}[F-MSI at $\setY$ and $\setZ$]\label{co:fullMessSIRecYZCommMess} [From \cite[Theorem~1]{kramershamai07}.\footnote{For the setting without a common message, the capacity region of the general BC with F-MSI at both receivers was established in \cite[Theorem~1]{kramershamai07}. This result readily extends to the setting with a common message.}] The capacity region of the SD-BC with F-MSI at both receivers ($R_\setY^{(p)} \!= 0$ and $R_\setZ^{(p)} \!= 0$) is the set of rate-tuples $(R, R_\setY, R_\setZ)$ satisfying
\begin{subequations} \label{bl:capRegFullMessSIRecYZCommMess}
\ba
R + R_\setY &\leq H (Y) \label{eq:fullMessSIRecYZRyCommMess} \\
R + R_\setZ &\leq I (X;Z) \label{eq:fullMessSIRecYZRzCommMess} \\
\ea
\end{subequations}
for some PMF of the form
\ba
p (x,y,z) = p (x) \, W (y,z|x). \label{eq:fullMessSIRecYZCommMess}
\ea
\end{corollary}

\subsection{How P-MSI Affects Capacity} \label{sec:effectPMSI}

In this section we study how P-MSI at the receivers affects the capacity region of the SD-BC.

\begin{remark}[P-MSI at $\setY$] \label{re:partMessSIRecy}
P-MSI at the deterministic receiver~$\setY$ can increase capacity if, and only if, the stochastic receiver $\setZ$ has P-MSI ($R_{\setY}^{(p)} \! < R_{\setY}$) or a common message is transmitted ($R > 0$).

In particular, the ``only if''-direction implies:
\begin{equation}\label{eq:equivalent}
(0,R_\setY, 0, R_\setZ^{(p)} \!, R_\setZ^{(c)}) \in 
\mathscr C_{\textnormal {P-MSI}} \quad \Longleftrightarrow \quad (0,R_\setY, R_\setZ) \in \mathscr C,
\end{equation}
where $\mathscr C$ denotes the capacity region of the SD-BC without MSI.
\end{remark}

\begin{proof}
The ``only-if'' direction follows from Corollary~\ref{co:noMessSI}. The ``if-direction'' follows from Examples~\ref{ex:MSIIncCap} and \ref{ex:commMessPMSIRecYIncCap} ahead. More specifically, Example~\ref{ex:MSIIncCap} shows that F-MSI at Receiver~$\setY$ can increase capacity if Receiver~$\setZ$ already has F-MSI; and Example~\ref{ex:commMessPMSIRecYIncCap} shows that F-MSI at Receiver~$\setY$ can increase capacity if a common message is transmitted.\footnote{Continuity considerations imply that it is not necessary to assume F-MSI ($R_{\setY}^{(p)} \! = 0$ or $R_{\setZ}^{(p)} \! = 0$), but that the statements also hold for P-MSI of the form $R_{\setY}^{(p)} \! \in (0,R_\setY)$ or $R_{\setZ}^{(p)} \! \in (0,R_\setZ)$.}
\end{proof}

\begin{remark}[P-MSI at $\setZ$] \label{re:MSIRecZ}
%Consider the setting without common message and without P-MSI at the deterministic receiver~$\setY$ ($R=0$ and $R_\setZ^{(p)} \! = R_\setZ$).
P-MSI at the stochastic receiver~$\setZ$ can increase the capacity region of the SD-BC; and this holds irrespective of whether or not the deterministic receiver has P-MSI or the encoder transmits a common message.
\end{remark}

\begin{proof}
Assume no common message ($R = 0$). By Corollary~\ref{co:noMessSI} the capacity region without MSI at Receiver~$\setZ$ does not depend on whether or not Receiver~$\setY$ has P-MSI ($R_\setZ^{(p)} \! \in [0,R_\setZ]$), and hence we obtain from Example~\ref{ex:MSIIncCap} ahead that F-MSI at Receiver~$\setZ$ can increase capacity, irrespective of whether or not Receiver~$\setY$ has P-MSI.\footnote{Continuity considerations imply that it is not necessary to assume F-MSI ($R_{\setY}^{(p)} \! = 0$), i.e., that the statement also holds for P-MSI of the form $R_{\setY}^{(p)} \! \in (0,R_\setY)$.} Continuity considerations imply that the statement also holds with a common message, i.e., for some $R > 0$.
\end{proof}

\begin{example}[P-MSI without a common message] \label{ex:MSIIncCap}
Consider the SD-BC with binary input $X$ and binary outputs 
\begin{IEEEeqnarray}{l}
Y = X \qquad \text {and} \qquad Z = X \oplus S, \label{eq:exMSIIncCapChannel}
\end{IEEEeqnarray}
where $S \sim \ber (p)$ is independent of $X$, and where $p \in ( 0,1/2 )$. Assume that the transmitter conveys only private messages ($R = 0$).

Using the capacity results of Section~\ref{sec:pmsiCapResults}, we can characterize the capacity region of the SC-BC \eqref{eq:exMSIIncCapChannel} for the cases where no receiver has MSI and where the stochastic receiver~$\setZ$ has F-MSI:

\begin{itemize}
\item \underline{Assume no MSI ($R_\setY^{(p)} \! = R_\setY$ and $R_\setZ^{(p)} \!= R_\setZ$).} As shown in \cite[Section~5.4.2]{gamalkim11}, the capacity region $\mathscr C$ without MSI is the set of all rate-pairs $( R_\setY, R_\setZ )$ that for some $\alpha \in [0,1/2]$ satisfy
\begin{subequations} \label{bl:exMSIIncCapNoMSI}
\ba
R_\setY &\leq h_\textnormal{b}(\alpha) \\
R_\setZ &\leq 1 - h_\textnormal{b}(\alpha + p - 2 \alpha p).
\ea
\end{subequations}
In particular, \eqref{bl:exMSIIncCapNoMSI} implies that without MSI the sum-rate is at most 1, and that it is strictly smaller than 1 whenever $R_\setZ > 0$.

\item \underline{Assume F-MSI at Receiver~$\setZ$ ($R_\setY^{(p)} \! = 0$).} By Corollary~\ref{co:fullMessSIRecZCommMess} (with $X \sim \ber (1/2)$) the capacity region $\mathscr C_{\textnormal {P-MSI}}$ with F-MSI at Receiver~$\setZ$ is the set of rate-tuples $(0, R_\setY, R_\setZ^{(p)} \!, R_\setZ^{(c)})$ that satisfy  
%\footnote{The region is also achieved by the following simple scheme. The transmitter encodes $M_\setY, M_\setZ^{(p)} \!, M_\setZ^{(c)}$ using Ber$(1/2)$,  Ber$(p_1)$, and Ber $(1/2)$ codebooks; this way it produces the codewords $\vecx_\setY(M_\setY)$, $\vecx_{\setZ,p}(M_\setZ^{(p)})$, and $\vecx_{\setZ,c}(M_\setZ^{(c)})$. It sends the 	XOR of these three sequences over the channel. Receiver~$\setY$ first subtracts codeword $\vecx_{\setZ,c}(M_\setZ^{(c)})$ from its received sequence, and  decodes its desired message  $\vecx_\setY(M_\setY)$ from the obtained sequence 	treating codeword $\vecx_{\setZ,p}(M_\setZ^{(p)})$ as noise. Receiver~$\setZ$ 	first subtracts $\vecx_\setY(M_\setY)$ from its received sequence and 	successively decodes its desired messages $M_\setZ^{(c)}$ and $M_\setZ^{(p)}$.}
\begin{subequations} \label{bl:exMSIIncCapMSI}
\ba
%R_\setY &\leq 1 \\
R_\setZ &\leq 1-h_\textnormal{b}(p)\\
R_\setY + R_\setZ^{(p)} \! &\leq 1. 
\ea
\end{subequations}

For the case where Receiver~$\setY$ has no MSI and Receiver~$\setZ$ has F-MSI ($R_\setY^{(p)}\!= 0$ and $R_\setZ^{(p)} \!= R_\setZ$), Constraints~\eqref{bl:exMSIIncCapMSI} imply that the sum-rate is at most 1 but can also be 1 when $R_\setZ > 0$. 

For the case where both receivers have F-MSI ($R_\setZ^{(p)} \!= R_\setY^{(p)} \!= 0$), Constraints~\eqref{bl:exMSIIncCapMSI} imply that the sum-rate can exceed $1$.
\end{itemize}

From the above observations we see that the capacity of the studied SD-BC \eqref{eq:exMSIIncCapChannel} satisfies the following two: 
\begin{enumerate}
	\item The capacity region without MSI is strictly contained in the capacity region without MSI at Receiver~$\setY$ and with F-MSI at Receiver~$\setZ$.
	\item The capacity region without MSI at Receiver~$\setY$ and with F-MSI at Receiver~$\setZ$ is strictly contained in the capacity region with F-MSI at both receivers.
	\end{enumerate}
\end{example}

\begin{example}[P-MSI with a common message] \label{ex:commMessPMSIRecYIncCap}
Consider the SD-BC with input $X = ( X_1, X_2)$, where $X_1$ and $X_2$ are binary, and with outputs 
\ba \label{eq:examp}Y = X_1 + X_2\qquad \textnormal{and}\qquad 
Z = \begin{cases} X_2 &S = 0, \\ ? &S = 1, \end{cases}
\ea
where $S \sim \ber (p)$ is independent of $X$, and where $p \in (0,1)$.

In Appendix~\ref{sec:pfExCommMessPMSIRecYIncCap} we prove the following facts on the maximum sum-rate that is achievable on the SD-BC \eqref{eq:examp}:

\begin{itemize}
\item
\underline{Assume F-MSI at $\setY$ and no MSI at $\setZ$ ($R_\setY^{(p)} \!= R_\setY$ and $R_\setZ^{(p)} \!= 0$).}
Denote the set of PMFs $p(u,x,y,z)$ satisfying \eqref{eq:fullMessSIRecYCommMess} by $\mathcal{P}_u$. The maximum achievable sum-rate $R+R_\setY+R_\setZ$ is
\ba 
\max_{p(u,x,y,z) \in \mathcal{P}_u} \bigl\{ H(Y|U) + I(U;Z) \bigr\} = 1 - p + \Biggl[ p \, h_{\textnormal b} \biggl( \frac{1}{1+2^{1/p}} \biggr) + \frac{2^{1/p}}{1+2^{1/p}} \Biggr]. \label{eq:maxSRFullMessSIRecYCM}
\ea

Let $\mathcal{P}_u^\star$ denote the set of PMFs $p (u,x,y,z) \in \mathcal{P}_u$ that maximize the LHS of \eqref{eq:maxSRFullMessSIRecYCM}. W.r.t.\ every PMF $p (u,x,y,z) \in \mathcal{P}_u^\star$
\ba 
&I (U;Y) < \min \bigl\{ H (Y), I (U;Z) \bigr\}. \label{eq:maxSRFullMessSIRecYCMMutis}
\ea
Moreover, the largest common-message rate $R^\star_{\textnormal{F-MSI@}\setY}$ for which the maximum sum-rate is achievable is 
\ba \label{eq:Rcommonstar}
R^\star_{\textnormal{F-MSI@}\setY}= \max_{p (u,x,y,z) \in \mathcal{P}_{u}^\star}  \min{ \bigl\{ H (Y), I (U;Z) \bigr\}}.
\ea
(Note that---by continuity and because the set $\mathcal{P}_u^\star$ is compact---the maxima in \eqref{eq:maxSRFullMessSIRecYCM} and \eqref{eq:Rcommonstar} are attained.)
 
\item
\underline{Assume no MSI ($R_\setY^{(p)} \!= R_\setY$ and $R_\setZ^{(p)} \!= R_\setZ$).} The maximum achievable sum-rate is again \eqref{eq:maxSRFullMessSIRecYCM}, and it can be achieved only if $p (u,x,y,z) \in \mathcal{P}_u$. Let $\mathcal{P}_{uv}^\star$ denote the set of PMFs $p (u,v,x,y,z)$ that for some $p(u,x,y,z) \in \mathcal{P}_{u}^\star$ are of the form $$p (u,v,x,y,z)=p(v|u) \, p(u,x,y,z).$$ The largest common-message rate $R^\star_{\textnormal{no-MSI}}$ for which the maximum sum-rate is achievable satisfies
\ba 
R^\star_{\textnormal{no-MSI}} \leq \max_{p (u,v,x,y,z) \in \mathcal{P}_{uv}^\star} I (V;Y). \label{eq:commMessPMSIRecYIncCapRSR}
\ea
\end{itemize}

We can use the above observations to compare $R^\star_{\textnormal{no-MSI}}$ and $R^\star_{\textnormal{F-MSI@}\setY}$:
\begin{IEEEeqnarray}{rCl}
R^\star_{\textnormal{no-MSI}} & \stackrel{(a)}\leq &  \max_{p (u,v,x,y,z) \in\mathcal{P}_{uv}^\star}  I (V;Y) \nonumber \\
& \stackrel{(b)}\leq &  \max_{p (u,x,y,z) \in\mathcal{P}_{u}^\star}  I (U;Y) \nonumber \\
& \stackrel{(c)}< & \max_{p (u,x,y,z) \in\mathcal{P}_{u}^\star} \min \bigl\{ H (Y), I (U;Z) \bigr\} \nonumber \\
& \stackrel{(d)}= & R^\star_{\textnormal{F-MSI@}\setY},
\end{IEEEeqnarray}
where $(a)$ follows from \eqref{eq:commMessPMSIRecYIncCapRSR}; $(b)$ holds by definition of $\mathcal{P}_{uv}^\star$; $(c)$ follows from \eqref{eq:maxSRFullMessSIRecYCMMutis}; and $(d)$ follows from \eqref{eq:Rcommonstar}. The comparison reveals that for the SD-BC \eqref{eq:examp} the capacity region with F-MSI at Receiver~$\setY$ and without MSI at Receiver~$\setZ$ strictly contains that without MSI.

%This, \eqref{eq:commMessPMSIRecYIncCapRSR}, and the fact that every PMF $p (u,x,y,z)$ of the form \eqref{eq:fullMessSIRecYCommMess} that achieves the maximum achievable sum-rate satisfies \eqref{eq:maxSRFullMessSIRecYCMMutis} imply that the maximum sum-rate that can be achieved for some fixed common-message rate $R$ can be larger in the case with F-MSI at the deterministic receiver~$\setY$ ($R_\setZ^{(p)} \!= 0$) and without MSI at the stochastic receiver~$\setZ$ ($R_\setY^{(p)} \!= R_\setY$) than in the case where neither receiver has MSI ($R_\setY^{(p)} = R_\setY$ and $R_\setZ^{(p)} = R_\setZ$).
\end{example}

\subsection{Intuition on the Results} \label{sec:discuss}

In this section we provide some intuition on the results of Section~\ref{sec:effectPMSI}. To keep the exposition simple we consider only the case without a common message.
%First, we provide intuition on why in the setting without a common message P-MSI at the stochastic receiver only can increase capacity, whereas P-MSI at the deterministic receiver only cannot increase capacity. Then, we explain why P-MSI at the deterministic receiver can increase capacity if one of the following two holds: 1) the stochastic receiver has P-MSI; or 2) the stochastic receiver has no P-MSI, but the encoder also conveys a common message.

%Assume that
When a transmitter sends two independent messages over a BC to two receivers, then the transmission to each of the receivers is interfered by the transmission to the other receiver. Since the transmitter knows the two messages, it knows the two transmissions acausally and can thus partially mitigate the interference experienced at each receiver (see Marton's scheme \cite{marton79}).

Suppose now that a receiver has P-MSI. Such a receiver has partial knowledge of the transmission to the other receiver, and in general knowing interference at both the transmitter and the receiver is better (in terms of achievable rates) than knowing it only at the transmitter (cf.\ the state-dependent single-user channel with acausal state-information (SI) at the transmitter). Consequently, P-MSI at a receiver allows to better mitigate interference on the BC and hence to achieve larger rates. This provides some intuition for our finding that---in the setting without a common message---P-MSI at the SD-BC's stochastic receiver only can increase capacity. %can better mitigate the interference that it experiences due to the transmission to the other receiver. Therefore, P-MSI at only one receiver typically increases capacity. 

In contrast, we have seen that---in the setting without a common message---P-MSI at the SD-BC's deterministic receiver only cannot increase capacity. Also for this result we can obtain some intuition from the single-user channel whose transmission is subject to interference. Indeed, if the single-user channel's outputs can be computed from its inputs and the interference, then knowing the interference only at the transmitter is as beneficial as knowing it also at the receiver.\footnote{To see this, consider a deterministic state-dependent single-user channel $W (y|x,s)$, whose output is a function of its input and the state. That the capacity of this channel with acausal SI at the encoder does not depend on whether or not the state is revealed to the receiver can be seen as follows. If the receiver does not observe the state, then the capacity is given by the Gelfand-Pinsker formula \cite{gelfandpinsker80SI} $$C = \max_{p (u,x|s)} I (U;Y) - I (U;S),$$ which for $U = Y (X,S)$ evaluates to $\max_{p (x|s)} H (Y|S)$. If the receiver observes the state, then the capacity is $$C = \max_{p (x|s)} I (X;Y|S),$$ which is equivalent to $\max_{p (x|s)} H (Y|S)$, because $Y = Y (X,S)$. This, and the fact that the receiver can always ignore the state that it observes, prove the claim.}

Interference can be mitigated more efficiently if both receivers have P-MSI. For example, if both receivers have F-MSI ($R_\setY^{(p)} \! = 0$ and $R_\setZ^{(p)} \! = 0$) and, moreover, $R_\setY^{(c)} \! = R_\setZ^{(c)}$, then the transmitter can send the ``x-or'' of the messages $M_{\setY}^{(c)}$ and $M_{\setZ}^{(c)}$ as a ``common message'' to both receivers, and each receiver can recover its message by first recovering the common message and then subtracting the message for the other receiver, which it knows. This provides some intuition for our finding that on the SD-BC P-MSI at both receiver's is better (in terms of achievable rates) than P-MSI at the stochastic receiver only.

%For the setting with a common message we have seen that---in contrast to the setting without a common message---P-MSI can also increase capacity if it is available only at the deterministic receiver. A reason for this might be that the encoder conveys the common message in the cloud-center of Marton's inner bound, and that the cloud-center might be easier to decode at the stochastic receiver than at the deterministic receiver. In such cases it would be desirable to add an extra message for the stochastic receiver to the cloud-center; and if the deterministic receiver knows this extra message due to its P-MSI, then this could be done without making the decoding at the deterministic receiver unreliable.

\section{Feedback on the SD-BC} \label{sec:fb}

This section investigates how feedback can be used on the SD-BC. The feedback that we consider is perfect or rate-limited, and it is only from the stochastic receiver~$\setZ$. (Recall that feedback from the deterministic receiver~$\setY$ is useless.) For simplicity of exposition, we assume that there is no common message $(R=0)$. But our ideas easily extend to the common-message setting.

%We show that feedback can increase the capacity of the SD-BC with or without P-MSI and, in particular, the sum-rate capacity when the deterministic receiver has P-MSI (Theorem~\ref{th:fbCanIncCapSDBC}). To this end we introduce a feedback code, and we exhibit an example of an SD-BC for which our feedback code can achieve rate-pairs that are not achievable without feedback (Example~\ref{ex:addEras}). Our code allows the encoder to create P-MSI at the stochastic receiver. This is useful, because---as we have seen in Section~\ref{sec:effectPMSI}---P-MSI allows the stochastic receiver to better mitigate experienced interference (Remark~\ref{re:MSIRecZ}). We also show that feedback cannot increase the capacity of the SD-BC if the stochastic receiver~$\setZ$ has F-MSI.

\subsection{Preliminaries: An Enhanced BC}
Consider the \emph{enhanced BC} with P-MSI of Figure~\ref{fig:model_enh}, which is obtained from the SD-BC with P-MSI by revealing the stochastic outputs $\{Z_i\}$ also to the deterministic receiver~$\setY$. It plays an important role in the feedback code that we present in the next section.
\begin{figure}[ht]
\vspace{-2mm}

\begin{center}
\def\pgfsysdriver{pgfsys-dvipdfm.def}
\begin{tikzpicture}[circuit logic US]
% Define a few styles and constants
	\tikzstyle{sensor}=[draw, minimum width=4em, text centered, minimum height=2em]
	\tikzstyle{stategen}=[draw, text width=2.2em, text centered, minimum height=2em]
	\tikzstyle{delay}=[draw, text width=0.5em, text centered, minimum height=1em]
	\tikzstyle{naveqs} = [sensor, minimum width=4em, minimum height=5.5em]
	\def\blockdist{2.4}
	\def\edgedist{2.5}
    \node (naveq) [naveqs] {$\channel {y,z}{x}$};
    \path (naveq.west)+(-0.7*\blockdist,0) node (enc) [sensor] {Encoder};
    \path (naveq.east)+(0.8*\blockdist,0.3*\blockdist) node (dec1) [sensor] {Rec.~$\setY$};
    \path (naveq.east)+(0.8*\blockdist,-0.3*\blockdist) node (dec2) [sensor] {Rec.~$\setZ$};
    \path (dec1.east)+(0.3*\blockdist,0) node (guess1) [text centered] {};
    \path (dec2.east)+(0.3*\blockdist,0) node (guess2) [text centered] {};
    \path (dec1)+(0,0.5*\blockdist) node (msi1) [text centered] {};
    \path (dec2)+(0,-0.5*\blockdist) node (msi2) [text centered] {};     
    \path (enc.west)+(-0.3*\blockdist,0) node (sour) [text centered] {};
%    \path (dec2)+(-0.3*\blockdist,-0.45*\blockdist) node (fb1) [text centered] {};
%    \path (fb1.north)+(-1.3*\blockdist,0) node (del) [delay] {\small D};
    		
    \path [draw, ->] (enc) -- node [above] {$X_{i}$} 
        (naveq.west |- enc);        
    \path [draw, ->] (sour) -- node [left] {$\begin{matrix}M_\setY \\ M_\setZ\end{matrix}$ \quad}
    		(enc.west |- sour);
    	\path [draw, <-] (dec1) -- node [above] {$Y_i$}
	(naveq.east |- dec1);
	  \path [draw, <-] (dec1) -- node [below] {${\color{red}Z_i}$}
    		(naveq.east |- dec1);
    \path [draw, <-] (dec2) -- node [above] {$Z_i$}
    		(naveq.east |- dec2);
    	\path [draw, <-] (dec1) -- node [right] {$M^{(c)}_\setZ$} 
        (msi1);
    	\path [draw, <-] (dec2) -- node [right] {$M^{(c)}_\setY$} 
        (msi2);    		
    	\path [draw, ->] (dec1) -- node [right] { \; $ \widehat{M_\setY}$} 
        (guess1);
    \path [draw, ->] (dec2) -- node [right] { \; $\widehat{M_\setZ}$} 
        (guess2);
%    \path [draw, dashed] (fb1) -- node [above] {} (fb1.north |- dec2.south);
%    \path [draw, dashed, - latex new, arrow head=0.2cm] (fb1.north) -- node [below] {$Z_i$ / $W_i$} (del);
%%    \path [draw, dashed] (del.west) -- node [above] {$Z^{i-1}$ / $W^{i-1}$} (enc.south |- del.west);
%	\path [draw, dashed] (del.west) -- (enc.south |- del.west);
%    \path [draw, dashed, latex new -, arrow head=0.2cm] (enc.south) -- node [left] {} (enc.south |- del.west);
    
\end{tikzpicture}

\caption{Enhanced BC with P-MSI.}
\label{fig:model_enh}
\end{center}
\vspace{-2mm}

\end{figure}
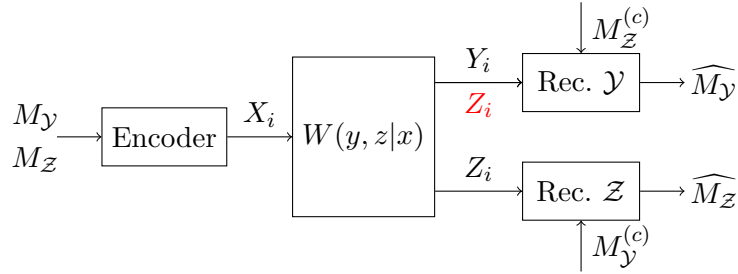

The capacity region of the enhanced BC is defined similarly as that of the SD-BC (see Section~\ref{sec:chMod}). We denote it by $\mathscr C^{( \textnormal{enh} )}_{\textnormal{P-MSI}}$, and in the special case without MSI by 
$\mathscr C^{( \textnormal{enh} )}$.

\begin{proposition}[Enhanced BC with P-MSI]\label{pr:capRegEnhBC}
The capacity region $\mathscr C^{( \textnormal{enh} )}_{\textnormal{P-MSI}}$ of the enhanced BC with P-MSI is the set of rate-tuples $(R_\setY^{(p)} \!, R_\setY^{(c)} \!, R_\setZ^{(p)} \!, R_\setZ^{(c)})$ satisfying
\begin{subequations}\label{bl:capRegEnhPMSI}
\ba
R_\setZ &\leq I (U;Z) \label{eq:rzEnhMSI} \\
R_\setY + R_\setZ^{(p)} \!&\leq I (X;Y,Z) \label{eq:ryEnhMSI} \\
R_\setY^{(p)} \!+ R_\setZ &\leq I (X;Y,Z|U) + I (U;Z) \label{eq:rypRzEnhMSI}
\ea
\end{subequations}
for some PMF of the form
\ba
p (u,x,y,z) = p (u,x) \, W (y,z|x). \label{eq:PMSIPMFEnh}
\ea
\end{proposition}

\begin{proof}
The result is an immediate consequence of \cite[Theorem~3]{kramershamai07}, which characterizes the capacity region of the BC with P-MSI and degraded message sets. This holds because, as we argue next, the capacity region of the enhanced BC remains unchanged if the stronger receiver~$\setY$ must decode also the pair $(M_\setZ^{(p)}, M_\setZ^{(c)})$. To see this, note that a coding scheme is reliable on the enhanced BC with P-MSI if, and only if, the following two conditions hold: 1)~Receiver~$\setZ$ can decode $(M_\setZ^{(p)} \!, M_\setZ^{(c)})$ reliably from $M_\setY^{(c)}$ and the outputs $Z^n$; and 2)~Receiver~$\setY$ can decode $(M_\setY^{(p)} \!, M_\setY^{(c)})$ reliably from $M_\setZ^{(c)}$ and the outputs ($Y^n, Z^n)$. But these two conditions are equivalent to the conditions that result when in 2)~we require that Receiver~$\setY$ can---in addition to $(M_\setY^{(p)} \!, M_\setY^{(c)})$---decode also the pair $(M_\setZ^{(p)}, M_\setZ^{(c)})$ reliably.
%This holds because the enhanced BC is degraded and because receiver~$\setY$ has to reliably decode receiver $\setZ$'s MSI $M_{\setY}^{(c)}$.}
%remains The reason for this is that the enhanced BC is a degraded BC whose stronger receiver is $\setY$, and that, as we argue next, on the degraded BC we can w.l.g.\ assume that the stronger receiver~$\setY$ must also decode the weaker receiver~$\setZ$'s message $M_\setZ^{(p)}$. Indeed, since the weaker receiver~$\setZ$ can decode $(M_\setZ^{(p)} \!, M_\setZ^{(c)})$ from the outputs that it observes and $M_\setY^{(c)}$, and since receiver~$\setY$ can decode $(M_\setY^{(p)} \!, M_\setY^{(c)})$ from the outputs that it observes and $M_\setZ^{(c)}$, it follows that the stronger receiver~$\setY$ can decode $M_\setZ^{(p)}$ from the outputs that it observes and $(M_\setY^{(p)} \!, M_\setY^{(c)} \!, M_\setZ^{(c)})$.
\end{proof}

\begin{corollary}[Enhanced BC without MSI] \label{co:enhWithoutMSI}
The capacity region $\mathscr C^{( \textnormal{enh} )}$ of the enhanced BC without MSI ($R_\setY^{(p)} \!= R_\setY$ and $R_\setZ^{(p)} \!= R_\setZ$) is the set of rate-tuples $( R_\setY, 0, R_\setZ^{(p)} \!, R_\setZ^{(c)} )$ satisfying \begin{subequations}\label{bl:capRegEnh}
\ba 
R_\setY &\leq I (X;Y,Z|U) \label{eq:ryEnh} \\
R_\setZ &\leq I (U;Z) \label{eq:rzEnh}
\ea
\end{subequations}
for some PMF of the form \eqref{eq:PMSIPMFEnh}. 
\end{corollary}

From Proposition~\ref{pr:capRegEnhBC} and Corollary~\ref{co:enhWithoutMSI} it follows that on the enhanced BC P-MSI at Receiver~$\setY$ only cannot increase capacity:

\begin{corollary}[Enhanced BC without MSI at $\setZ$]
If Receiver~$\setZ$ has no MSI ($R_\setY^{(p)} \!= R_\setY$), then, irrespective of $R_\setZ^{(p)} \! \in [0,R_\setZ]$,
\begin{equation} \label{eq:equivalentEnh}
(R_\setY, 0, R_\setZ^{(p)} \!, R_\setZ^{(c)}) \in 
\mathscr C^{( \textnormal{enh} )}_{\textnormal {P-MSI}} \quad \Longleftrightarrow \quad (R_\setY, R_\setZ) \in \mathscr C^{( \textnormal{enh} )}.
\end{equation}
\end{corollary}

\subsection{Coding Scheme with Rate-Limited Feedback} \label{sec:scheme}

Assume that the feedback is rate-limited (see Figure~\ref{fig:model} when the dashed links transport the feedback signal $W_i$).\footnote{Recall from Section~\ref{sec:chMod} that perfect feedback is equivalent to rate-limited feedback of rate $R_{\textnormal{FB}} \geq \log |\setZ|$.} We first describe our feedback scheme for the SD-BC without MSI at the stochastic receiver~$\setZ$. Later, we shall generalize it to the case with P-MSI at both receivers.

The scheme has two phases: for some fixed $\alpha \in (0,1)$, Phase~1 comprises the first $\alpha n$ channel uses $[1:\alpha n]$, and Phase~2 comprises the remaining $(1-\alpha) n$ channel uses $[\alpha n + 1 : n]$. We next describe Phase~1 and Phase~2, beginning with Phase~1.

\underline{Phase~1:} In Phase~1 the transmitter codes for the enhanced BC with P-MSI at the deterministic receiver~$\setY$ at rates 
\begin{subequations}\label{eq:condfb}
\ba \label{eq:cond13}
\bigl( \tilde R_\setY^{(p)} \!, 0, \tilde R_\setZ^{(p)} \!, \tilde R_\setZ^{(c)} \bigr) \in \mathscr C^{( \textnormal{enh} )}_{\textnormal{P-MSI}}.
\ea
At the end of the phase the stochastic receiver~$\setZ$ decodes its intended rate-$(\tilde R_\setZ^{(p)} \!, \tilde R_\setZ^{(c)})$ messages. Moreover, it compresses its Phase-1 channel-outputs and uses the rate-limited feedback-link to send the compression index that it obtains to the encoder.\footnote{Receiver~$\setZ$ uses the feedback link only between the two phases, i.e., only at Time~$\alpha n$, and for every $i \neq \alpha n$ the feedback signal $W_i$ thus takes value in a size-1 set $\setW_i$. This allows the Receiver~$\setZ$ to choose the Time-$\alpha n$ feedback $W_{\alpha n}$ from a size-$2^{n R_{\textnormal{FB}}}$ set $\setW_{\alpha n}$ while guaranteeing that the rate-limitation \eqref{eq:rateLimitation} be met. By using $B$ instances of the two phases, and by starting the transmission with sufficiently many instances of Phase~1, one could use the feedback link more evenly and (for $B \rightarrow \infty$) guarantee that each alphabet $\setW_i, \,\, i \in [1:n]$ be of size at most $2^{R_{\textnormal{FB}}}$.} The deterministic receiver~$\setY$ performs no action in Phase~1.

\underline{Phase~2:} In Phase~2 the encoder transmits the compression index, which Receiver~$\setZ$ sent over the feedback link in Phase~1, along with fresh message-information. It transmits the compression index at rate $R_\setY^{(c)}$ and the fresh message-information at rates $( R_\setY^{(p)} \!, 0, R_\setZ^{(p)} \!, R_\setZ^{(c)} )$, where
\ba \label{eq:cond23}
( R_\setY^{(p)} \!, R_\setY^{(c)} \!, R_\setZ^{(p)} \!, R_\setZ^{(c)} ) \in \mathscr C_{\textnormal {P-MSI}}.
\ea
Note that Receiver~$\setZ$ knows the rate-$R_\setY^{(c)}$ compression-index and can thus use it as P-MSI. At the end of the phase Receiver~$\setZ$ decodes its intended rate-$(R_\setZ^{(p)} \!, R_\setZ^{(c)})$ messages; and Receiver~$\setY$ performs the following actions: 
\begin{enumerate}
\item Based on its Phase-2 channel-outputs $Y^n_{\alpha n + 1}$, it decodes the rate-$R_\setY^{(p)}$ message and the rate-$R_{\setY}^{(c)}$ compression-index that were sent to it in Phase~2. 
\item Based on its estimate of the compression index and its Phase-1 channel-outputs $Y^{\alpha n}$, it decodes Receiver~$\setZ$'s Phase-1 channel-outputs $Z^{\alpha n}$. 
\item Based on its Phase-1 channel-outputs $Y^{\alpha n}$ and its estimate of Receiver~$\setZ$'s Phase-1 channel-outputs $Z^{\alpha n}$, it decodes the rate-$\tilde R_\setY^{(p)}$ message that was sent to it in Phase~1.
\end{enumerate}

%ence, in Phase~2 the encoder codes for the SD-BC with P-MSI at \emph{both} receivers, and to guarantee a small phase-2 decoding error probability it is thus enough to have $$. Note that only receiver~$\setY$ needs to be able to reconstruct receiver~$\setZ$'s channel outputs $Z^{\alpha n}$. Since receiver~$\setY$ can use its own phase-1 outputs $Y^{\alpha n}$ as side-information, receiver~$\setZ$'s compression index consists of $\alpha n H (Z|Y)$ bits, where the conditional entropy is w.r.t.\ a PMF $p (u,x,y,z)$ of the form \eqref{eq:PMSIPMFEnh} for which the rate-tuple $( \tilde R_\setY^{(p)} \!, 0, \tilde R_\setZ^{(p)} \!, \tilde R_\setZ^{(c)} )$ satisfies \eqref{bl:capRegEnhPMSI} ahead. 

By \eqref{eq:cond13} and \eqref{eq:cond23} we can guarantee that the probability of a decoding error tend to 0 as $n$ tends to infinity whenever Receiver~$\setY$ can recover Receiver~$\setZ$'s Phase-1 channel-outputs $Z^{\alpha n}$. As we argue next, we can guarantee that---with probability tending to $1$ as $n$ tends to infinity---the latter hold whenever the feedback-rate $R_{\textnormal{FB}}$ and the rate $R_\setY^{(c)}$ at which the compression index is sent in Phase~2 satisfy
\begin{IEEEeqnarray}{rCl}
\alpha H (Z|Y) & < & (1 - \alpha) R_\setY^{(c)} \label{eq:1111}\\
(1 - \alpha) R_\setY^{(c)} & \leq & R_{\textnormal{FB}}, \label{eq:2222}
\end{IEEEeqnarray}
\end{subequations}
where the conditional entropy $H (Z|Y)$ is computed w.r.t.\ some PMF $p (u,x,y,z)$ of the form \eqref{eq:PMSIPMFEnh} for which the rate-tuple $( \tilde R_\setY^{(p)} \!, 0, \tilde R_\setZ^{(p)} \!, \tilde R_\setZ^{(c)} )$ satisfies \eqref{bl:capRegEnhPMSI}. Indeed, Condition~\eqref{eq:1111} guarantees that---with probability tending to $1$ as $n$ tends to infinity---Receiver~$\setY$ can decode Receiver~$\setZ$'s Phase-1 channel-outputs $Z^{\alpha n}$ from its own Phase-1 channel-outputs and the compression index that is sent to it in Phase~2, and Condition~\eqref{eq:2222} guarantees that the compression index can be sent over the feedback link. To see this, recall that the compression index is sent in Phase~2, and that the rate $R_\setY^{(c)}$ at which it is sent is thus computed w.r.t.\ the $(1-\alpha) n$ channel uses that Phase~2 comprises. Moreover, we can guarantee that---with probability tending to $1$ as $n$ tends to infinity---Receiver~$\setY$ can decode Receiver~$\setZ$'s Phase-1 channel-outputs $Z^{\alpha n}$ based on the compression index and its own Phase-1 channel-outputs $Y^n$ whenever the rate of the compression index---computed w.r.t.\ the first $\alpha n$ channel uses---exceeds $H (Z|Y)$.
 
From the above we conclude that our feedback code achieves any rate-tuple of the form
\ba 
\alpha ( \tilde R_\setY^{(p)} \!, 0, \tilde R_\setZ^{(p)} \!, \tilde R_\setZ^{(c)} ) + (1 - \alpha ) ( R_\setY^{(p)} \!, 0, R_\setZ^{(p)} \!, R_\setZ^{(c)} )
\ea
for which Conditions~\eqref{eq:condfb} hold. A more detailed analysis of the scheme can be found in Appendix~\ref{sec:pfPrEnhMSI}.\\

We can easily adapt the above feedback code to allow P-MSI also at the stochastic receiver~$\setZ$, i.e., $R_\setY^{(p)} \! \in [0, R_\setY]$ and $R_\setZ^{(p)} \! \in [0, R_\setZ]$. To this end we modify the code as follows:
\begin{itemize}
\item In Phase~1 the encoder codes for the enhanced BC with P-MSI at \emph{both} receivers.
\item In Phase~2 the compression index that is sent to the deterministic receiver constitutes \emph{additional} P-MSI at the stochastic receiver, i.e., P-MSI beyond the one that the stochastic receiver already had in the beginning.
\end{itemize}
In Phase~1 of the adapted code the encoder thus codes at rates
\begin{subequations} \label{bl:rateConstsFBCodePMSIBoth}
\ba
( \tilde R_\setY^{(p)} \!, \tilde R_\setZ^{(c)} \!, \tilde R_\setZ^{(p)} \!, \tilde R_\setZ^{(c)} ) \in \mathscr C^{( \textnormal{enh} )}_{\textnormal{P-MSI}}; \label{eq:ratesPhase1FBCodePMSIBoth}
\ea
and in Phase~2 the encoder sends fresh message-information at rates $( R_\setY^{(p)} \!, \hat R_\setY^{(c)} \!, R_\setZ^{(p)} \!, R_\setZ^{(c)} )$, where
\ba
\hat R_\setY^{(c)} \! \leq R_\setY^{(c)} \! - \frac{\alpha}{1 - \alpha} H (Z|Y),
\ea
and
\ba
( R_\setY^{(p)} \!, R_\setY^{(c)} \!, R_\setZ^{(p)} \!, R_\setZ^{(c)} ) \in \mathscr C_{\textnormal {P-MSI}}.
\ea
\end{subequations}
The adapted feedback code achieves any rate-tuple of the form
\ba 
\alpha ( \tilde R_\setY^{(p)} \!, \tilde R_\setY^{(c)} \!, \tilde R_\setZ^{(p)} \!, \tilde R_\setZ^{(c)} ) + (1 - \alpha ) ( R_\setY^{(p)} \!, \hat R_\setY^{(c)} \!, R_\setZ^{(p)} \!, R_\setZ^{(c)} )
\ea
for which Conditions~\eqref{eq:2222} and \eqref{bl:rateConstsFBCodePMSIBoth} hold.

The following proposition summarizes which rate-tuples our feedback scheme can achieve:

\begin{proposition}[Performance of the feedback code]\label{pr:enhMSI}
Fix any PMF $p(u,x,y,z)$ of the form~\eqref{eq:PMSIPMFEnh}, and let $( \tilde R_\setY^{(p)} \!, \tilde R_\setY^{(c)} \!, \tilde R_\setZ^{(p)} \!, \tilde R_\setZ^{(c)} )$ be any rate-tuple that---for the fixed PMF $p(u,x,y,z)$---satisfies \eqref{bl:capRegEnhPMSI}. In addition, pick any rate-tuple $( R_\setY^{(p)} \!, R_\setY^{(c)} \!, R_\setZ^{(p)} \!, R_\setZ^{(c)} ) \in \mathscr C_{\textnormal {P-MSI}}$, any nonnegative number
\ba 
\alpha \leq \min \Biggl\{ \frac{R_\setY^{(c)}}{R_\setY^{(c)} \!+ H (Z|Y)}, \frac{R_{\textnormal{FB}}}{H (Z|Y)} \Biggr\}, \label{eq:alpha}
\ea
where the conditional entropy $H (Z|Y)$ is computed w.r.t.\ the fixed PMF $p(u,x,y,z)$, and any nonnegative rate
\ba
\hat R_\setY^{(c)} \!\leq R_\setY^{(c)} \!- \frac{\alpha}{1 - \alpha} H (Z|Y). \label{eq:hatRYc}
\ea
The capacity region of the SD-BC with P-MSI and rate-limited feedback of rate $R_{\textnormal{FB}}$ from the stochastic receiver~$\setZ$ contains the rate-tuple
\ba 
\alpha ( \tilde R_\setY^{(p)} \!, \tilde R_\setY^{(c)} \!, \tilde R_\setZ^{(p)} \!, \tilde R_\setZ^{(c)} ) + (1-\alpha) ( R_\setY^{(p)} \!, \hat R_\setY^{(c)} \!, R_\setZ^{(p)} \!, R_\setZ^{(c)} ). \label{eq:FBRateTuple}
\ea
\end{proposition}

\begin{proof}
See Appendix~\ref{sec:pfPrEnhMSI}.
\end{proof}

%To obtain a more general version of the feedback code, one could make the following modifications.
Our feedback code can be generalized along the following guidelines:

\begin{itemize}
\item To obtain a feedback code for the general BC, in Phase~2 one can replace the capacity-achieving code for the SD-BC with P-MSI by a ``good'' code for the general BC with P-MSI, e.g., by the code of \cite{asadiongjohnson15}.

\item Instead of recovering $Z^{\alpha n}$ losslessly, after Phase~2 Receiver~$\setY$ could recover a lossy version of these outputs.\footnote{In particular, this generalization would allow us to extend the feedback code to continuous output alphabets $\setZ$.} To allow for this generalization, the enhanced BC would have to be adapted so that Receiver~$\setY$ does not observe Receiver~$\setZ$'s output but only a lossy version of it. In general, the capacity of such an enhanced BC (with P-MSI) is unknown; and hence the encoder would have to use a ``good'' rather than a capacity-achieving code for it, e.g., the code of \cite{asadiongjohnson15} for the general BC with P-MSI.

\item Instead of partitioning the transmission into two phases, one could use a block-Markov framework as in \cite{wuwigger14}.
\end{itemize}

In the absence of MSI, extending our feedback code along the above guidelines results in the feedback code of \cite{wuwigger14}, which---in the absence of MSI---is thus more general than our code.\footnote{In particular, in the absence of MSI the rate region that is achievable with our feedback code is contained in the rate region that is achievable with the feedback code of \cite{wuwigger14}.}

%\begin{remark}[Feedback code of \cite{wuwigger14}] \label{re:wuwiggerstronger}
%For the case without MSI, our feedback code for the SD-BC can be viewed as a special case of the feedback code of \cite{wuwigger14} for the general BC. In particular, Proposition~\ref{pr:noSIFBInc} ahead, which we use to show that feedback can increase the capacity region of the SD-BC without MSI, can be derived from \cite[Corollary~3]{wuwigger14}. 
%
%We think that our feedback code is nevertheless interesting for the following two reasons: 1)~Our feedback code is perhaps more natural than that of \cite{wuwigger14}, because we present it for the setting where the receivers can have P-MSI. This allows us to stay in the same framework---namely the one with P-MSI---when we use the feedback to create P-MSI. %2)~Our feedback code is simpler than that of \cite{wuwigger14}, and it allows us to propose some intuition on why feedback can increase the capacity of the SD-BC (see Section~\ref{sec:int}).
%\end{remark}

Using Proposition~\ref{pr:enhMSI}, we next identify sufficient conditions for feedback to increase the capacity of the SD-BC. For simplicity we assume that the stochastic receiver~$\setZ$ has no MSI, i.e., that $R_\setY^{(p)} \!= 0$. Proposition~\ref{pr:partSIRecYInc} treats the case where Receiver~$\setY$ has P-MSI ($R_\setZ^{(p)} \!< R_\setZ$), and Proposition~\ref{pr:noSIFBInc} treats the case without MSI ($R_\setZ^{(p)} \! = R_\setY$).

\begin{proposition}[Sufficient conditions with P-MSI at $\setY$]\label{pr:partSIRecYInc}
Consider an SD-BC with P-MSI only at the deterministic receiver~$\ry$ ($R_\setY^{(p)} \!= R_\setY$ and $R_\setZ^{(p)} \!\in [0, R_\setZ)$). If there exists a rate-triple $( R_\setY^{(p)} \!, R_\setZ^{(p)} \!,R_\setZ^{(c)} )$ satisfying
\begin{equation}\label{eq:cond_boundary}
( R_\setY^{(p)} \!, R_\setZ^{(p)} \!+ R_\setZ^{(c)} ) \in \bigl(\partial \mathscr C \cap ( \mathscr C_{\textnormal {enh}} \setminus \partial \mathscr C_{\textnormal {enh}} )\bigr),
\end{equation}
and if, for some PMF $p(u,x,y,z)$ of the form \eqref{eq:noMessSIPMF}, Conditions~\eqref{bl:capRegNoMessSI} and 
\ba
R_\setY < H (Y) \quad \text{and} \quad 0 < I (U;Y) \label{eq:thPartSIRecY}
\ea
hold, then, irrespective of $R_{\textnormal{FB}} > 0$, $( R_\setY^{(p)} \!, R_\setZ^{(p)} \!, R_\setZ^{(c)} )$ is in the interior of the feedback capacity region, and feedback thus increases the capacity region.
\end{proposition}

\begin{proof}
See Appendix~\ref{sec:pfSufficientConditions}.
\end{proof}

\begin{proposition}[Sufficient conditions without MSI]\label{pr:noSIFBInc}
Consider an SD-BC without MSI ($R_\setY^{(p)} \!= R_\setY$ and $R_\setZ^{(p)} \!= R_\setZ$). If there exists a rate-pair $(R_\setY, R_\setZ)$ satisfying
\ba\label{eq:cond_boundaryNoMSI}
( R_\setY, R_\setZ ) \in \bigl(\partial \mathscr C \cap ( \mathscr C_{\textnormal {enh}} \setminus \partial \mathscr C_{\textnormal {enh}} )\bigr),
\ea
and if, for some PMF $p(v,u,x,y,z)$ of the form \eqref{eq:partMessSIPMF}, Conditions~\eqref{bl:capRegNoMessSI} and
\ba 
R_\setY < H (Y) \quad \text{and} \quad 0 < I (V;Y) - I (V;Z) \label{eq:thNoSI}
\ea
hold, then, irrespective of $R_{\textnormal{FB}} > 0$, $( R_\setY, R_\setZ )$ is in the interior of the feedback capacity region, and feedback thus increases the capacity region.
\end{proposition}

\begin{proof}
See Appendix~\ref{sec:pfSufficientConditions}.
\end{proof}

Proposition~\ref{pr:noSIFBInc} is used in the analysis of Example~\ref{ex:addEras} ahead (see Appendix~\ref{sec:proof}), which proves that feedback can increase the capacity of the SD-BC without MSI. In this analysis one chooses $U = V$, which turns Marton coding into the simpler superposition coding with no satellite for Receiver~$\setZ$. There are other SD-BCs for which one has to apply Proposition~\ref{pr:noSIFBInc} with full Marton coding, i.e., with $U \neq V$, in order to show that feedback increases their capacity (see Remark~\ref{re:propNoSIFBIncVNeeded} in Appendix~\ref{sec:proof}).

It is perhaps surprising that Proposition~\ref{pr:noSIFBInc} can be used to show that feedback can increase the capacity of the SD-BC without MSI: The boundary points of the no-feedback capacity region of the SD-BC without MSI are known to be achievable using only satellites but no cloud-center, i.e., with $V = \emptyset$. This notwithstanding, Proposition~\ref{pr:noSIFBInc} relies on the assumption that one can achieve boundary points using full Marton coding with a cloud-center, i.e., with $V \neq \emptyset$ (see \eqref{eq:thNoSI}).

\subsection{How Feedback Affects Capacity}

We use the sufficient conditions of 
Propositions~\ref{pr:partSIRecYInc} and \ref{pr:noSIFBInc} to show that feedback can increase the capacity of the SD-BC without P-MSI and the sum-rate capacity of the SD-BC with P-MSI at the deterministic receiver~$\setY$. (The assumption that the stochastic receiver~$\setZ$ has no MSI is made for simplicity.) Key to the proof are the following two observations:

\begin{enumerate}
\item P-MSI to the stochastic receiver can increase the capacity of the SD-BC.
\item The capacity of the enhanced BC (with P-MSI) is typically larger than that of the SD-BC (with P-MSI).
\end{enumerate}

\begin{theorem}[Feedback can help] \label{th:fbCanIncCapSDBC}
If the stochastic receiver~$\setZ$ has no MSI ($R_{\setY}^{(p)} \!= R_\setY$), then, irrespective of whether or not the deterministic receiver~$\setY$ has P-MSI ($R_{\setZ}^{(p)} \!\in [0,R_\setZ]$), rate-limited feedback of any positive rate can increase the capacity region of the SD-BC. In particular, it can increase the sum-rate capacity if Receiver~$\setY$ has P-MSI ($R_\setZ^{(p)} \!< R_\setZ$).
\end{theorem}

The theorem follows from the following example, which we analyze in Appendix~\ref{sec:proof}:

\begin{example}[Feedback helps]\label{ex:addEras}
Consider the SD-BC whose input is $X = (X_1, X_2)$, where $X_1$ and $X_2$ are binary, and whose outputs are $Y = X_1 + X_2$ and
\bas
Z = \begin{cases} X_2 &S = 0, \\ ? &S = 1, \end{cases}
\eas
where $S \sim \ber (p), \,\, p \in (0,1)$ is independent of $X$. If the deterministic receiver~$\setY$ has P-MSI ($R_\setZ^{(p)} \!< R_\setZ$) and the stochastic receiver~$\setZ$ has no MSI ($R_\setY^{(p)} \!= R_\setY$), then rate-limited feedback of any positive rate increases the sum-rate capacity. If no receiver has MSI ($R_\setZ^{(p)} \!= R_\setZ$ and $R_\setY^{(p)} \!= R_\setY)$ and $p > 1/2$, then such feedback increases the capacity region.
\end{example}

As our next result shows, the setting where the stochastic receiver~$\setZ$ has F-MSI is different:

\begin{theorem}[With F-MSI at $\setZ$, feedback is useless]\label{th:fullMessSIRecZ}
If the stochastic receiver~$\setZ$ has F-MSI ($R_\setY^{(p)} \!= 0$), then, irrespective of whether or not the deterministic receiver~$\setY$ has P-MSI ($R_\setZ^{(p)} \!\in [0,R_\setZ]$), even perfect feedback cannot increase the capacity region of the SD-BC. 
\end{theorem}

\begin{proof}
See Appendix~\ref{sec:pfFullMessSIRecZ}.
\end{proof}

From Theorems~\ref{th:fbCanIncCapSDBC} and \ref{th:fullMessSIRecZ} it follows that whether feedback can increase the capacity of the SD-BC with P-MSI depends on the P-MSI. Of course, it also depends on the SD-BC. The following is an example of an SD-BC---which does not consists of two noninterfering single-user channels and is neither deterministic nor physically-degraded---whose capacity without MSI at the stochastic receiver~$\setZ$ cannot be increased by feedback, irrespective of whether or not the deterministic receiver~$\setY$ has P-MSI. The example is analyzed in Appendix~\ref{sec:pfFbNeedNotHelp}.

\begin{example}[Also without MSI at $\setZ$, feedback can be useless]\label{ex:fbNeedNotHelp}
Consider the SD-BC with input $X \in \setX$ and outputs
\ba
Y = f (X) \qquad \text{and} \qquad Z = \begin{cases} X &\text{if } S = 0, \\ ? &\text{if } S = 1, \end{cases} \label{eq:ZExFBUseless}
\ea
where $S \sim \ber (p)$ is independent of $X$, where $p \in (0,1)$, and where $? \notin \setX$. Let $\bar p = 1 - p$ and $$\setX_y = \bigl\{ x \in \setX \colon f (x) = y \bigr\}, \quad y \in \setY.$$ For the case where the stochastic receiver~$\setZ$ has no MSI ($R_\setY^{(p)} \!= R_\setY$), we show the following statement in Appendix~\ref{sec:pfFbNeedNotHelp}. Irrespective of $R_\setZ^{(p)} \!\in [0,R_\setZ]$, the capacity region is the set of rate-tuples $(R_\setY, 0, R_\setZ^{(p)} \!, R_\setZ^{(c)})$ satisfying
\begin{subequations}\label{bl:FBDoesNotHelp}
\ba 
R_\setY &\leq H (Y) \\
R_\setZ &\leq \bar p I (U;Y) + \bar p \sum_{y \in \setY} p (y) \log |\setX_y| \\
R_\setY + R_\setZ &\leq H (Y) - p I (U;Y) + \bar p \sum_{y \in \setY} p (y) \log |\setX_y|
\ea
\end{subequations}
for some PMF $p (u,x,y,z)$ of the form \eqref{eq:noMessSIPMF} under which $X$ is a function of $(Y,U)$. Further, we show that this holds also in the presence of feedback. These observations imply that feedback cannot increase the capacity region of the considered SD-BC without MSI at Receiver~$\setZ$.
\end{example}

\subsection{Strictly Causal SI on the State-Dependent SD-BC}

In this section we consider the state-dependent SD-BC, whose transition law is governed by an IID state-sequence of finite support $\setS$, i.e., $$W (y,z|x,s) = \ind {\{ y = f (x) \}} \, W (z|x,s), \quad s \in \setS.$$ We study how strictly-causal SI at the encoder affects the capacity of this channel. Because the setting with strictly-causal SI at the encoder is related to that with feedback, for the state-dependent SD-BC with strictly-causal SI at the encoder we can easily obtain results that are equivalent to those of Theorems~\ref{th:fbCanIncCapSDBC} and \ref{th:fullMessSIRecZ} for the SD-BC with feedback.

\begin{theorem}[Strictly-causal SI]
Consider the state-dependent SD-BC $$W (y,z|x,s) = \ind {\{ y = f (x) \}} \, W (z|x,s), \quad s \in \setS.$$ If the stochastic receiver~$\setZ$ has no MSI ($R_{\setY}^{(p)} \!= R_\setY$), then, irrespective of whether or not the deterministic receiver~$\setY$ has P-MSI ($R_{\setZ}^{(p)} \!\in [0,R_\setZ]$), strictly-causal SI can increase the capacity region of the SD-BC. In particular, it can increase the sum-rate capacity if Receiver~$\setY$ has P-MSI ($R_\setZ^{(p)} \!< R_\setZ$). If Receiver~$\setZ$ has F-MSI ($R_\setY^{(p)} \!= 0$), then, irrespective of whether or not Receiver~$\setY$ has P-MSI ($R_\setZ^{(p)} \!\in [0,R_\setZ]$), strictly-causal SI cannot increase capacity.
\end{theorem}

\begin{proof}
The Functional Representation lemma (Lemma~\ref{le:funcRep}) allows us to view the state-less SD-BC $W (y,z|x)$ as a state-dependent SD-BC $W (y,z|x,s)$ whose stochastic output $Z$ can be computed from its input and state. If the state is revealed strictly-causally to the encoder, then the encoder can compute Receiver~$\setZ$'s past channel outputs from the past states and channel inputs, and hence feedback cannot be better than strictly-causal SI. Consequently, the first two claims of the theorem follow as a corollary to Theorem~\ref{th:fbCanIncCapSDBC}.

The last claim of the theorem (i.e., that if Receiver~$\setZ$ has F-MSI ($R_\setY^{(p)} \!= 0$), then, irrespective of whether or not Receiver~$\setY$ has P-MSI ($R_\setZ^{(p)} \!\in [0,R_\setZ]$), strictly-causal SI cannot increase capacity) can be established along the line of argument in the proof of Theorem~\ref{th:fullMessSIRecZ}.
\end{proof}

\section*{Acknowledgements}

Helpful discussions with A.\ Lapidoth are gratefully acknowledged.
The authors would also like to thank the anonymous reviewers and the Associate Editor whose valuable comments helped substantially improve the quality of the paper.

%
%%%%%%%%%%%%%%%%%%%%%%%%%%%%%%%%%%%%%%%%%%%%%%%%%%%%%%%%%%%%%%%%%%
%%%%%%%%%%%%%%%%%%%%%%%%%%%%%%%%%%%%%%%%%%%%%%%%%%%%%%%%%%%%%%%%%%
%% APPENDIX
%%%%%%%%%%%%%%%%%%%%%%%%%%%%%%%%%%%%%%%%%%%%%%%%%%%%%%%%%%%%%%%%%%
%%%%%%%%%%%%%%%%%%%%%%%%%%%%%%%%%%%%%%%%%%%%%%%%%%%%%%%%%%%%%%%%%%
%
\begin{appendix}
%
%%%%%%%%%%%%%%%%%%%%%%%%%%%%%%%%%%%%%%%%%%%%%%%%%%%%%%%%%%%%%%%%%%
%% HEADER APPENDIX
%%%%%%%%%%%%%%%%%%%%%%%%%%%%%%%%%%%%%%%%%%%%%%%%%%%%%%%%%%%%%%%%%%

\section{Proof of Theorem~\ref{th:partMessSI}}\label{sec:thPartMessSI}

\subsection{Direct Part}\label{sec:dir}

\textit{Codebook Generation:} Fix positive real numbers $\tilde \epsilon > \epsilon > 0$, a PMF $p(v,u,x) = p(v,u) \, p(x|u)$, and rates $\tilde R^{(c)} \!, \, \tilde R_\setY, \, \tilde R_\setZ > 0$ for which $\tilde R^{(c)} \!\leq \min \{ \tilde R_\setY, \tilde R_\setZ \}$. For each $( m, m_\setY^{(c)} \!, m_\setZ^{(c)} ) \in \setM \times \setM_\setY^{(c)} \!\times \setM_{\setZ}^{(c)}$ draw $2^{n \tilde R^{(c)}}$ $n$-tuples $\vecv$ from the product distribution $\prod^n_{i = 1} p (v_i)$, label them by $k \in [ 1 : 2^{n \tilde R^{(c)}} ]$, where $\vecv (k)$ denotes the $n$-tuple labelled by $k$, and place them in a bin $\mathscr B ( m, m_\setY^{(c)} \!, m_\setZ^{(c)} )$. For each $n$-tuple $\vecv$ draw $2^{n ( \tilde R_\setY - \tilde R^{(c)} )}$ $n$-tuples $\vecy \in \setY^n$ independently of each other and each from the product distribution $\prod^n_{i = 1} p (y_i|v_i)$, and label them by $\ell \in [ 1 : 2^{n ( \tilde R_\setY - \tilde R^{(c)} )} ]$, where $\vecy (k,\ell)$ denotes the $\ell$-th $n$-tuple corresponding to $\vecv (k)$. Randomly allocate the $2^{n \tilde R_\setY}$ $n$-tuples to $2^{n R_\setY^{(p)}}$ bins $\mathscr B_\setY ( m, m_\setY^{(c)} \!, m_\setZ^{(c)} \!, m_\setY^{(p)} )$, where $m_\setY^{(p)} \!\in \setM_{\setY}^{(p)}$. Similarly, for each $n$-tuple $\vecv$ draw $2^{n ( \tilde R_\setZ - \tilde R^{(c)} )}$ $n$-tuples $\vecu \in \setU^n$ independently of each other and each from the product distribution $\prod^n_{i = 1} p (u_i|v_i)$, and label them by $j \in [ 1 : 2^{n ( \tilde R_\setZ - \tilde R^{(c)} )} ]$, where $\vecu (k,j)$ denotes the $j$-th $n$-tuple $\vecu$ corresponding to $\vecv (k)$. Randomly allocate the $2^{n \tilde R_\setZ}$ $n$-tuples to $2^{n R_\setZ^{(p)}}$ bins $\mathscr B_\setZ ( m, m_\setY^{(c)} \!, m_\setZ^{(c)} \!, m_\setZ^{(p)} )$, where $m_\setZ^{(p)} \!\in \setM_{\setZ}^{(p)}$.

\textit{Encoding:} If there are $n$-tuples $\vecv \in \mathscr B ( M, M_\setY^{(c)} \!, M_\setZ^{(c)} )$, $\vecy \in \mathscr B_\setY ( M, M_\setY^{(c)} \!, M_\setZ^{(c)} \!, M_\setY^{(p)} )$, and $\vecu \in \mathscr B_\setZ ( M, M_\setY^{(c)} \!, M_\setZ^{(c)} \!, M_\setZ^{(p)} )$ for which $(\vecv, \vecy, \vecu) \in \setT^{(n)}_\epsilon (V,Y,U)$, then the encoder draws the length-$n$ channel-input-sequence $X^n$ from the product distribution $\prod^n_{i = 1} p (x_i|y_i,u_i)$. Otherwise the encoding is unsuccessful.

\textit{Decoding:} Upon observing $\bigl( Y^n,M_\setZ^{(c)} \bigr)$, Receiver~$\setY$ decodes $(m,m_\setY)$ if it is the unique element of $\setM \times \setM_\setY$ for which $Y^n \in \mathscr B_\setY ( m, m_\setY^{(c)} \!, M_\setZ^{(c)} \!, m_\setY^{(p)} ) \cap \setT^{(n)}_{\tilde \epsilon} (Y)$. Otherwise it declares an error. Upon observing $\bigl( Z^n,M_\setY^{(c)} \bigr)$, Receiver~$\setZ$ decodes $(m,m_\setZ)$ if it is the unique element of $\setM_\setZ$ for which Bin~$\mathscr B_\setZ ( m, M_\setY^{(c)} \!, m_\setZ^{(c)} \!, m_\setZ^{(p)} )$ contains an $n$-tuple $\vecu$ satisfying $(\vecu,Z^n) \in \setT^{(n)}_{\tilde \epsilon} (U,Z)$. Otherwise it declares an error.

\textit{Analysis of the Error Probability:} We first analyze the probability that the encoding is unsuccessful. By symmetry, this probability does not depend on the messages' realizations, and we thus assume w.l.g.\ that $M = M_\setY^{(p)} \!= M_\setY^{(c)} \! = M_\setZ^{(p)} \!= M_\setZ^{(c)} \! = 1$. For ease of notation, let $\mathscr B = \mathscr B (1,1,1)$, $\mathscr B_\setY = \mathscr B_\setY (1,1,1,1)$, and $\mathscr B_\setZ = \mathscr B_\setZ (1,1,1,1)$. Note that the encoder can choose one out of
\bas
N = \sum_{k, \ell, j} \ind { (\vecv (k),\vecy (k,\ell),\vecu (k,j)) \in \setT^{(n)}_\epsilon } \ind {\vecy (k,\ell) \in \mathscr B_\setY} \ind {\vecu (k,j) \in \mathscr B_\setZ}
\eas
triples $(\vecv,\vecy,\vecu)$, where $k \in [1 : 2^{n \tilde R^{(c)}}]$, $\ell \in [1 : 2^{n ( \tilde R_\setY - \tilde R^{(c)} )}]$, and $j \in [1 : 2^{n ( \tilde R_\setZ - \tilde R^{(c)} )}]$. The encoding is unsuccessful if $N = 0$. As in the proof of the mutual covering lemma \cite[Lemma~8.1]{gamalkim11}, we obtain from Chebyshev's inequality that
\ba 
\distof {N = 0} \leq \Bigdistof { \bigl( N - \smEx {}{N} \bigr)^2 \geq \smEx {}{N}^2 } \leq \frac{\textnormal {Var} ( N ) }{\smEx {}{N}^2}. \label{eq:prUB}
\ea
To conclude that---on average over the realization of the code---the encoding is with high probability successful, it thus suffices to show that $$\frac{\textnormal {Var} ( N ) }{\smEx {}{N}^2} \rightarrow 0 \,\, (n \rightarrow \infty).$$ Define
\begin{subequations}
\ba 
E^{k,\ell,j} &= \ind {(\vecv (k),\vecy (k,\ell),\vecu (k,j)) \in \setT^{(n)}_\epsilon } \ind {\vecy (k,\ell) \in \mathscr B_\setY} \ind {\vecu (k,j) \in \mathscr B_\setZ} \\
T^{k,\ell,j} &= \ind {(\vecv (k),\vecy (k,\ell),\vecu (k,j)) \in \setT^{(n)}_\epsilon },
\ea
\end{subequations}
and note that
\ba 
\bigEx {}{N^2} &= \sum_{k,\ell,j} \sum_{k^\prime \!,\ell^\prime \!,j^\prime \!} \BigEx {}{E^{k,\ell,j} E^{k^\prime \!,\ell^\prime \!,j^\prime}} \\
&= \sum_{k,\ell,j} \sum_{k^\prime \! \neq k,\ell^\prime \!,j^\prime} \BigEx {}{E^{k,\ell,j} E^{k^\prime \!,\ell^\prime \!,j^\prime}} \nonumber \\
&\quad + \sum_{k,\ell,j} \sum_{\ell^\prime \! \neq \ell,j^\prime \! \neq j} \BigEx {}{E ^{k,\ell,j} E^{k,\ell^\prime \!,j^\prime}} \nonumber \\
&\quad + \sum_{k,\ell,j} \sum_{j^\prime \! \neq j} \BigEx {}{E^{k,\ell,j} E^{k,\ell,j^\prime}} \nonumber \\
&\quad + \sum_{k,\ell,j} \sum_{\ell^\prime \! \neq \ell} \BigEx {}{E^{k,\ell,j} E^{k,\ell^\prime \!,j}} \nonumber \\
&\quad + \sum_{k, \ell, j} \BigEx {}{E^{k,\ell,j}}, \label{eq:expNSq}
\ea
where we used that $E^{k,\ell,j} E^{k,\ell,j} = E^{k,\ell,j}$. If $k \neq k^\prime \!$, then $E^{k,\ell,j}$ and $E^{k^\prime \!,\ell^\prime \!,j^\prime}$ are independent and
\ba 
&\sum_{k,\ell,j} \sum_{k^\prime \! \neq k,\ell^\prime \!,j^\prime} \BigEx {}{E^{k,\ell,j} E^{k^\prime \!,\ell^\prime \!,j^\prime}} \nonumber \\
&\quad = \sum_{k,\ell,j} \sum_{k^\prime \! \neq k,\ell^\prime \!,j^\prime} \BigEx {}{E^{k,\ell,j}} \BigEx {}{E^{k^\prime \!,\ell^\prime \!,j^\prime}} \\
&\quad \leq \Biggl( \sum_{k,\ell,j} \BigEx {}{E^{k,\ell,j}} \Biggr)^{\!\! 2} \\
&\quad = \smEx {}{N}^2. \label{eq:kDiff}
\ea
If $k = k^\prime \!$, $\ell \neq \ell^\prime \!$, and $j \neq j^\prime \!$, then
\ba 
&\sum_{k,\ell,j} \sum_{\ell^\prime \! \neq \ell,j^\prime \! \neq j} \BigEx {}{E ^{k,\ell,j} E^{k,\ell^\prime \!,j^\prime}} \nonumber \\
&\quad \stackrel{(a)}\leq 2^{- 2 n (R_\setY^{(p)} \!+ R_\setZ^{(p)})} \sum_{k,\ell,j} \sum_{\ell^\prime \! \neq \ell,j^\prime \! \neq j} \!\!\!\!\! \BigEx {}{T^{k,\ell,j} T^{k,\ell^\prime \!,j^\prime}} \\
&\quad \stackrel{(b)}\leq 2^{- 2 n (R_\setY^{(p)} \!+ R_\setZ^{(p)})} \sum_{k,\ell,j} \sum_{\ell^\prime \! \neq \ell,j^\prime \! \neq j} 2^{-n ( 2 I (Y;U|V) - \delta (\epsilon) )} \\
&\quad = 2^{n (2 \tilde R_\setY \!- 2 R_\setY^{(p)} \!+ 2 \tilde R_\setZ \!- 2 R_\setZ^{(p)} \!- 3 \tilde R^{(c)} \!- 2 I (Y;U|V) + \delta (\epsilon)) }, \label{eq:lMDiff}
\ea
where $(a)$ holds because $\ind {\vecy (k,\ell) \in \mathscr B_\setY}$, $\ind {\vecu (k,j) \in \mathscr B_\setZ}$, $\ind {\vecy (k,\ell^\prime) \in \mathscr B_\setY}$, and $\ind {\vecu (k,j^\prime) \in \mathscr B_\setZ}$ are independent of each other and of $T^{k,\ell,j} T^{k,\ell^\prime \!,j^\prime}$, and because
\bas
\bigEx {}{\ind {\vecy (k,\ell) \in \mathscr B_\setY}} &= 2^{-n R_\setY^{(p)}} \\
\bigEx {}{\ind {\vecu (k,j) \in \mathscr B_\setZ}} &= 2^{-n R_\setZ^{(p)}};
\eas
and where $(b)$ holds because of the properties of typical sequences. (Recall that $\delta (\cdot)$ denotes any function of $\epsilon$ that converges to $0$ as $\epsilon$ approaches $0$.) If $k = k^\prime \!$, $\ell = \ell^\prime \!$, and $j \neq j^\prime \!$, then
\ba 
&\sum_{k,\ell,j} \sum_{j^\prime \! \neq j} \BigEx {}{E^{k,\ell,j} E^{k,\ell,j^\prime}} \nonumber \\
&\quad \stackrel{(a)}\leq 2^{- n (R_\setY^{(p)} \!+ 2 R_\setZ^{(p)})} \sum_{k,\ell,j} \sum_{j^\prime \! \neq j} \BigEx {}{T^{k,\ell,j} T^{k,\ell,j^\prime}} \\
&\quad \stackrel{(b)}\leq 2^{- n (R_\setY^{(p)} \!+ 2 R_\setZ^{(p)})} \sum_{k,\ell,j} \sum_{j^\prime \! \neq j} 2^{-n ( 2 I (Y;U|V) - \delta (\epsilon) )} \\
&\quad = 2^{n ( \tilde R_\setY - R_\setY^{(p)} \!+ 2 \tilde R_\setZ - 2 R_\setZ^{(p)} \!- 2 \tilde R^{(c)} \!- 2 I (Y;U|V) + \delta (\epsilon) ) }, \label{eq:mDiff}
\ea
where $(a)$ holds because $\ind {\vecy (k,\ell) \in \mathscr B_\setY} \ind {\vecy (k,\ell) \in \mathscr B_\setY} = \ind {\vecy (k,\ell) \in \mathscr B_\setY}$ and because $\ind {\vecy (k,\ell) \in \mathscr B_\setY}$, $\ind {\vecu (k,j) \in \mathscr B_\setZ}$, and $\ind {\vecu (k,j^\prime) \in \mathscr B_\setZ}$ are independent of each other and of $T^{k,\ell,j} T^{k,\ell,j^\prime} \!$; and where $(b)$ follows from the properties of typical sequences. Similarly, if $k = k^\prime \!$, $j = j^\prime \!$, and $\ell \neq \ell^\prime \!$, then
\ba 
&\sum_{k,\ell,j} \sum_{\ell^\prime \! \neq \ell} \BigEx {}{E^{k,\ell,j} E^{k,\ell^\prime \!,j}} \nonumber \\
&\quad = 2^{n ( 2 \tilde R_\setY - 2 R_\setY^{(p)} \!+ \tilde R_\setZ - R_\setZ^{(p)} \!- 2 \tilde R^{(c)} \!- 2 I (Y;U|V) + \delta (\epsilon) ) }. \label{eq:lDiff}
\ea
Finally, if $k = k^\prime \!$, $j = j^\prime \!$, and $\ell = \ell^\prime \!$, then
\ba 
&\sum_{k,\ell,j} \bigEx {}{E^{k,\ell,j}} = \smEx {}{N}. \label{eq:allEq}
\ea
Using that $\textnormal {Var} ( N ) = \bigEx {}{N^2} - \smEx {}{N}^2$, we obtain from \eqref{eq:prUB}, \eqref{eq:expNSq}, \eqref{eq:kDiff}, \eqref{eq:lMDiff}, \eqref{eq:mDiff}, \eqref{eq:lDiff}, and \eqref{eq:allEq} that $$\distof {N = 0} \rightarrow 0 \,\, (n \rightarrow \infty)$$ holds whenever the RHS of each of the equations \eqref{eq:lMDiff}, \eqref{eq:mDiff}, \eqref{eq:lDiff}, and \eqref{eq:allEq} is---asymptotically---negligibly-small compared to $\smEx {}{N}^2$. Note that
\ba 
\smEx {}{N} &= \sum_{k,\ell,j} \BigEx {}{E^{k,\ell,j}} \\
&\stackrel{(a)}= 2^{- n (R_\setY^{(p)} \!+ R_\setZ^{(p)}) } \sum_{k,\ell,j} \BigEx {}{T^{k,\ell,j}} \\
&\stackrel{(b)}\geq 2^{- n (R_\setY^{(p)} \!+ R_\setZ^{(p)}) } \sum_{k,\ell,j} 2^{-n ( I(Y;U|V) + \delta (\epsilon) )} \\
&= 2^{n ( \tilde R_\setY - R_\setY^{(p)} \!+ \tilde R_\setZ - R_\setZ^{(p)} \!- \tilde R^{(c)} \!- I (Y;U|V) - \delta (\epsilon) ) },
\ea
where $(a)$ is true because $\ind {\vecy (k,\ell) \in \mathscr B_\setY}$ and $\ind {\vecu (k,j) \in \mathscr B_\setZ}$ are independent of each other and of $T^{k,\ell,j}$; and where $(b)$ holds for all sufficiently-large $n$ by the properties of typical sequences. This proves that---on average over the realization of the code---the probability that the encoding is unsuccessful converges to $0$ as $n$ tends to infinity whenever
\begin{subequations}\label{bl:succEnc}
\ba 
\tilde R^{(c)} \!&> 3 \delta (\epsilon) \\
\tilde R_\setY - R_\setY^{(p)} \!&> 3 \delta (\epsilon) \\
\tilde R_\setZ - R_\setZ^{(p)} \!&> 3 \delta (\epsilon) \\
\tilde R_\setY - R_\setY^{(p)} \!+ \tilde R_\setZ \!- R_\setZ^{(p)} \!- \tilde R^{(c)} \!&> I (Y;U|V) + \delta (\epsilon).
\ea
\end{subequations}

Suppose now that the encoding was successful. Under this assumption, we next analyze the probability that the decoding is unsuccessful. Because Receiver~$\setY$ observes the $n$-tuple $\vecy$ that the encoder selected, and because $\setT^{(n)}_{\epsilon} (Y) \subseteq \setT^{(n)}_{\tilde \epsilon} (Y)$, it is clear that $Y^n \in \mathscr B_\setY ( M, M_\setY^{(c)} \!, M_\setZ^{(c)} \!, M_\setY^{(p)} ) \cap \setT^{(n)}_{\tilde \epsilon} (Y)$. Hence, the pair $(M,M_\setY)$ satisfies the decoding requirements. Moreover, for all sufficiently-large $n$ it holds with high probability that $(U^n,Z^n) \in \setT^{(n)}_{\tilde \epsilon} (P_{U,Z})$, and hence that the pair $(M,M_\setZ)$ satisfies the decoding requirements. From this we conclude that we need to worry only about the event that also other message-pairs $(m,m_\setY) \in \setM \times \setM_\setY \setminus \bigl\{ (M,M_\setY) \bigr\}$ or $(m,m_\setZ) \in \setM \times \setM_\setZ \setminus \bigl\{ (M,M_\setZ) \bigr\}$ satisfy the decoding requirements. For Receiver~$\setY$ this can happen if an $n$-tuple $\vecy$ that corresponds to the same length-$n$ sequence $\vecv$ as $Y^n$ satisfies the decoding requirements, or if an $n$-tuple $\vecy$ that corresponds to a different length-$n$ sequence $\vecv$ than $Y^n$ satisfies the decoding requirements; and similarly for Receiver~$\setZ$. Hence, by the properties of typical sequences, the probability that the decoding is unsuccessful converges to $0$ as $n$ tends to infinity if
\begin{subequations}\label{bl:succDec}
\ba 
\tilde R_\setY - \tilde R^{(c)} \!&< H (Y|V) - \delta (\tilde \epsilon) \\
\tilde R_\setY + R + R_{\setY}^{(c)} \!&< H (Y) - \delta (\tilde \epsilon) \\
\tilde R_\setZ - \tilde R^{(c)} \!&< I (U;Z|V) - \delta (\tilde \epsilon) \\
\tilde R_\setZ + R + R_{\setZ}^{(c)} \!&< I (U;Z) - \delta (\tilde \epsilon).
\ea
\end{subequations}
Now let $\epsilon$ and $\tilde \epsilon$ tend to $0$. Then, also $\delta (\epsilon)$ and $\delta (\tilde \epsilon)$ tend to $0$, and hence we conclude that a rate-tuple $(R, R^{(p)}_\setY \!, R^{(c)}_\setY \!, R^{(p)}_\setZ \!, R^{(c)}_\setZ)$ is achievable if there exist rates $\tilde R^{(c)} \!$, $\tilde R_\setY$, and $\tilde R_\setZ$ for which
\begin{subequations}
\ba 
- \tilde R^{(c)} \!&< 0 \\
R^{(p)}_\setY \!- \tilde R_\setY &< 0 \\
R^{(p)}_\setZ \!- \tilde R_\setZ &< 0 \\
R_\setY^{(p)} \!+ R_\setZ^{(p)} \!+ \tilde R^{(c)} \!- \tilde R_\setY - \tilde R_\setZ &< - I (Y;U|V) \\
\tilde R^{(c)} \!- \tilde R_\setY &< 0 \\
-\tilde R^{(c)} \!+ \tilde R_\setY &< H (Y|V) \\
R + R_{\setY}^{(c)} \!+ \tilde R_\setY &< H (Y) \\
\tilde R^{(c)} \!- \tilde R_\setZ &< 0 \\
- \tilde R^{(c)} \!+ \tilde R_\setZ &< I (U;Z|V) \\
R + R_{\setZ}^{(c)} \!+ \tilde R_\setZ &< I (U;Z).
\ea
\end{subequations}
A Fourier-Motzkin elimination reveals that the above inequalities hold if the rate-tuple $(R, R^{(p)}_\setY \!, R^{(c)}_\setY \!, R^{(p)}_\setZ \!, R^{(c)}_\setZ)$ lies in the interior of $\mathscr C_{\textnormal{P-MSI}}$.\\

At first sight, it is perhaps surprising that we bin the cloud-center. As the following remark shows, this is simply an equivalent alternative to rate-splitting. The benefit is that binning the cloud-center requires only one auxiliary rate, namely $\tilde R^{(c)}$, whereas rate-splitting requires two auxiulary rates, namely one for each rate $R_\setY^{(p)}$ and $R_\setZ^{(p)}$.

\begin{remark}\label{re:binningCCEqRS}
The effect of binning the cloud-center is that of rate-splitting. More precisely, instead of generating $2^{n \tilde R^{(c)}}$ cloud-centers $\vecv$ for each triple $( m, m_\setY^{(c)} \!, m_\setZ^{(c)} )$, we could execute the following three stepts: 1)~divide the messages $M_\setY^{(p)}$ and $M_\setZ^{(p)}$ into two parts, i.e., $M_\setY^{(p)} \!= (M_{\setY,s}^{(p)}, M_{\setY,c}^{(p)})$ and $M_\setZ^{(p)} \!= (M_{\setZ,s}^{(p)}, M_{\setZ,c}^{(p)})$; 2)~generate for each tuple $( m,m_\setY^{(c)} \!, m_\setZ^{(c)} \!, m_{\setY,c}^{(p)}, m_{\setZ,c}^{(p)} )$ a cloud-center $\vecv$; and 3)~allocate the associated satellites $\vecy$ and $\vecu$ to $2^{n R_{\setY,s}^{(p)}}$ and $2^{n R_{\setZ,s}^{(p)}}$ instead of $2^{n R_\setY^{(p)}}$ and $2^{n R_\setZ^{(p)}}$ bins, respectively.
\end{remark}

To see that binning the cloud-center is tantamount to rate-splitting, note that with rate-splitting we generate $2^{n (R_{\setY,c}^{(p)} \!+ R_{\setZ,c}^{(p)})}$ cloud-centers $\vecv$ per triple $( m, m_\setY^{(c)} \!, m_\setZ^{(c)} )$, namely one for each pair $( m_{\setY,c}^{(p)}, m_{\setZ,c}^{(p)} )$. Therefore, we can identify $R_{\setY,c}^{(p)} + R_{\setZ,c}^{(p)}$ in the rate-splitting code with $\tilde R^{(c)}$ in the code where we also bin the cloud-center. Moreover, associating every cloud-center $\vecv$ that corresponds to some triple $( m, m_\setY^{(c)} \!, m_\setZ^{(c)} )$ with a different pair $( m_{\setY,c}^{(p)}, m_{\setZ,c}^{(p)} )$ and allocating the satellites $\vecy$ and $\vecu$ to $2^{n R_{\setY,s}^{(p)}}$ and $2^{n R_{\setZ,s}^{(p)}}$ bins, respectively, is tantamount to not associating the cloud-center with anything while allocating the satellites $\vecy$ and $\vecu$ to $2^{n R_\setY^{(p)}}$ and $2^{n R_\setZ^{(p)}}$ bins, respectively. From these observations it follows that the only difference between binning the cloud-center and rate-splitting is that for the latter $\tilde R^{(c)} \!= R_{\setY,c}^{(p)} + R_{\setZ,c}^{(p)}$ must satisfy the upper bound $\tilde R^{(c)} \!\leq R_\setY^{(p)} \!+ R_\setZ^{(p)} \!$. This upper bound is not, however, restrictive: if $\tilde R^{(c)} \!> R_\setY^{(p)} \!+ R_\setZ^{(p)} \!$, then we have more cloud-centers $\vecv$ than message-tuples $(m, m_\setY^{(p)} \!, m_\setY^{(c)} \!, m_\setZ^{(p)} \!, m_\setZ^{(c)})$; and this cannot be better than having for each message-tuple a different cloud-center.

\subsection{Converse}\label{sec:conv}

Let $Q \sim \unif [ 1 : n ]$ be independent of $( M, M_\setY, M_\setZ )$, and introduce
\begin{subequations}\label{bl:auxVar}
\ba
V_Q &= \bigl( M, M_\setY^{(c)} \!, M_\setZ^{(c)} \!, Y^n_{Q+1}, Z^{Q-1} \bigr), \\
U_Q &= \bigl( M, M_\setY^{(c)} \!, M_\setZ, Y^n_{Q+1}, Z^{Q-1} \bigr), \\
V &= ( V_Q, Q ), \, U = ( U_Q, Q ), \\
X &= X_Q, \, Y = Y_Q, \, Z = Z_Q.
\ea
\end{subequations}
Note that $V$, $U$, $X$, and $(Y,Z)$ form a Markov chain in that order, i.e., that their PMF is of the form \eqref{eq:partMessSIPMF}.

The rate of the message-pair $( M, M_\setY )$ intended to Receiver~$\setY$ satisfies
\ba 
R + R_\setY - \epsilon_n &\stackrel{(a)}\leq \frac{1}{n} I \bigl( M, M_\setY ; Y^n, M_\setZ^{(c)} \bigr) \\
&\stackrel{(b)}= I \bigl( M, M_\setY;Y_Q \bigl| M_\setZ^{(c)} \!,Y^{Q-1},Q \bigr) \label{eq:recyFB} \\
&\stackrel{(c)}\leq H (Y),
\ea
where $(a)$ follows from Fano's inequality; $(b)$ follows from the chain-rule and the independence of $( M, M_\setY )$ and $M_\setZ^{(c)} \!$; and $(c)$ holds because conditional entropy is nonnegative and conditioning cannot increase entropy.

The rate of the message-pair $( M, M_\setZ )$ intended to Receiver~$\setZ$ satisfies
\ba 
R + R_\setZ - \epsilon_n &\stackrel{(a)}\leq \frac{1}{n} I \bigl( M, M_\setZ ; Z^n,M_\setY^{(c)} \bigr) \\
&\stackrel{(b)}= I \bigl( M, M_\setZ;Z_Q \bigl| M_\setY^{(c)} \!,Z^{Q-1},Q \bigr) \label{eq:reczFB} \\
&\stackrel{(c)}\leq I (U;Z),
\ea
where $(a)$ follows from Fano's inequality; $(b)$ holds because of the chain-rule and because $(M, M_\setZ)$ and $M_\setY^{(c)}$ are independent; and $(c)$ is true because conditioning cannot increase entropy.

We prove the sum-rate constraints using Csisz\'{a}r's sum-identity, which states that, for every tuple $(A^n,B^n,T)$,
\bas
0 &= \frac{1}{n} \sum^n_{i = 1} \bigl[ I (A^n_{i+1};B_i|B^{i-1},T) \! - \! I (B^{i-1};A_i|A_{i+1}^n,T) \bigr] \\
&= I (A^n_{Q+1};B_Q|B^{Q-1},T,Q) \! - \! I (B^{Q-1};A_Q|A_{Q+1}^n,T,Q),
\eas
where $Q \sim \unif [ 1 : n ]$ is independent of $(A^n,B^n,T)$. The first sum-rate constraint that we prove is that
\ba 
&R + R_\setY^{(p)} \!+ R_\setZ - \epsilon_n \nonumber \\
&\quad \stackrel{(a)}\leq \frac{1}{n} \Bigl[ I \bigl( M_\setY^{(p)};Y^n, M, M_\setY^{(c)} \!, M_\setZ \bigr) + I \bigl( M,M_\setZ;Z^n,M_\setY^{(c)} \bigr) \Bigr] \nonumber \\
&\quad \stackrel{(b)}= I \bigl( M_\setY^{(p)};Y_Q\bigl|M, M_\setY^{(c)} \!,M_\setZ,Y^n_{Q+1},Q \bigr) \nonumber \\
&\qquad + I \bigl( M, M_\setZ;Z_Q\bigl|M_\setY^{(c)} \!,Z^{Q-1},Q \bigr) \\
&\quad \stackrel{(c)}= I \bigl( M_\setY^{(p)};Y_Q\bigl|M, M_\setY^{(c)} \!,M_\setZ,Y^n_{Q+1},Z^{Q-1},Q \bigr) \nonumber \\
&\qquad - I \bigl( Z^{Q-1};Y_Q\bigl|M, M_\setY,M_\setZ,Y^n_{Q+1},Q \bigr) \nonumber \\
&\qquad + I \bigl(Z^{Q-1};Y_Q\bigl|M, M_\setY^{(c)} \!,M_\setZ,Y^n_{Q+1},Q \bigr) \nonumber \\
&\qquad + I \bigl(M, M_\setY^{(c)} \!,M_\setZ,Y^n_{Q+1},Z^{Q-1},Q;Z_Q\bigr) \nonumber \\
&\qquad - I \bigl(M_\setY^{(c)} \!,Z^{Q-1},Q;Z_Q\bigr) \nonumber \\
&\qquad - I \bigr(Y^n_{Q+1};Z_Q\bigl|M, M_\setY^{(c)} \!,M_\setZ,Z^{Q-1},Q\bigr) \\
&\quad \stackrel{(d)}\leq I \bigl( M_\setY^{(p)};Y_Q\bigl|M, M_\setY^{(c)} \!,M_\setZ,Y^n_{Q+1},Z^{Q-1},Q \bigr) \nonumber \\
&\qquad + I \bigl(M, M_\setY^{(c)} \!,M_\setZ,Y^n_{Q+1},Z^{Q-1},Q;Z_Q\bigr) \\
&\quad \stackrel{(e)}\leq H \bigl(Y_Q\bigl|M,M_\setY^{(c)} \!,M_\setZ,Y^n_{Q+1},Z^{Q-1},Q\bigr) \nonumber \\
&\qquad + I \bigl(M,M_\setY^{(c)} \!,M_\setZ,Y^n_{Q+1},Z^{Q-1},Q;Z_Q\bigr) \\
&\quad \stackrel{(f)}= H (Y|U) + I (U;Z), \label{eq:ryPrivRz}
\ea
where $(a)$ follows from Fano's inequality; $(b)$ follows from the chain-rule and the independence of $M$, $M_\setY^{(p)} \!$, $M_\setY^{(c)} \!$, and $M_\setZ$; $(c)$ follows from the chain-rule; $(d)$ holds because of Csisz\'{a}r's sum-identity and because mutual information is nonnegative; $(e)$ holds because conditional entropy is nonnegative; and $(f)$ follows from \eqref{bl:auxVar}. Similarly, we obtain the sum-rate constraint
\ba
&R + R_\setY + R_\setZ^{(p)} \!- \epsilon_n \nonumber \\
&\quad \stackrel{(a)}\leq \frac{1}{n} \Bigl[ I \bigl(M,M_\setY^{(c)};Y^n,M_\setZ^{(c)} \bigr) + I \bigl( M_\setY^{(p)};Y^n,M,M_\setY^{(c)} \!,M_\setZ \bigr) \Bigr. \nonumber \\
&\qquad + I \bigl( M_\setZ^{(p)};Z^n,M,M_\setY^{(c)} \!,M_\setZ^{(c)} \bigr) \Bigr] \\
&\quad \stackrel{(b)}= I \bigl(M,M_\setY^{(c)};Y_Q\bigl|M_\setZ^{(c)} \!,Y^{n}_{Q+1},Q\bigr) \nonumber \\
&\qquad + I \bigl(M_\setY^{(p)};Y_Q\bigl|M,M_\setY^{(c)} \!,M_\setZ,Y^n_{Q+1},Q\bigr) \nonumber \\
&\qquad + I \bigl(M_\setZ^{(p)};Z_Q\bigl|M,M_\setY^{(c)} \!,M_\setZ^{(c)} \!,Z^{Q-1},Q\bigr) \\
&\quad \stackrel{(c)}= I \bigl(M,M_\setY^{(c)} \!,M_\setZ^{(c)} \!,Y^{n}_{Q+1},Z^{Q-1},Q;Y_Q\bigr) \nonumber \\
&\qquad - I \bigl(M_\setZ^{(c)} \!,Y^n_{Q+1},Q;Y_Q\bigr) \nonumber \\
&\qquad - I \bigl(Z^{Q-1};Y_Q\bigl|M,M_\setY^{(c)} \!,M_\setZ^{(c)} \!,Y^{n}_{Q+1},Q\bigr) \nonumber \\
&\qquad + I \bigl(M_\setY^{(p)};Y_Q\bigl|M,M_\setY^{(c)} \!,M_\setZ,Y^n_{Q+1},Z^{Q-1},Q\bigr) \nonumber \\
&\qquad - I \bigl(Z^{Q-1};Y_Q\bigl|M,M_\setY,M_\setZ,Y^n_{Q+1},Q\bigr) \nonumber \\
&\qquad + I \bigl(Z^{Q-1};Y_Q\bigl|M,M_\setY^{(c)} \!,M_\setZ,Y^n_{Q+1},Q\bigr) \nonumber \\
&\qquad + I \bigl(M_\setZ^{(p)};Z_Q\bigl|M,M_\setY^{(c)} \!,M_\setZ^{(c)} \!,Y^n_{Q+1},Z^{Q-1},Q\bigr) \nonumber \\
&\qquad - I \bigl(Y^n_{Q+1};Z_Q\bigl|M,M_\setY^{(c)} \!,M_\setZ,Z^{Q-1},Q\bigr) \nonumber \\
&\qquad + I \bigl(Y^n_{Q+1};Z_Q\bigl|M,M_\setY^{(c)} \!,M_\setZ^{(c)} \!,Z^{Q-1},Q\bigr) \\
&\quad \stackrel{(d)}\leq I \bigl(M,M_\setY^{(c)} \!,M_\setZ^{(c)} \!,Y^{n}_{Q+1},Z^{Q-1},Q;Y_Q\bigr) \nonumber \\
&\qquad + I \bigl(M_\setY^{(p)};Y_Q\bigl|M,M_\setY^{(c)} \!,M_\setZ,Y^n_{Q+1},Z^{Q-1},Q\bigr) \nonumber \\
&\qquad + I \bigl(M_\setZ^{(p)};Z_Q\bigl|M,M_\setY^{(c)} \!,M_\setZ^{(c)} \!,Y^n_{Q+1},Z^{Q-1},Q\bigr) \\
&\quad \stackrel{(e)}\leq I \bigl(M,M_\setY^{(c)} \!,M_\setZ^{(c)} \!,Y^{n}_{Q+1},Z^{Q-1},Q;Y_Q\bigr) \nonumber \\
&\qquad + H \bigl(Y_Q\bigl|M,M_\setY^{(c)} \!,M_\setZ,Y^n_{Q+1},Z^{Q-1},Q\bigr) \nonumber \\
&\qquad + I \bigl(M_\setZ^{(p)};Z_Q\bigl|M,M_\setY^{(c)} \!,M_\setZ^{(c)} \!,Y^n_{Q+1},Z^{Q-1},Q\bigr) \\
&\quad \stackrel{(f)}= I (V;Y) + H (Y|U) + I (U;Z|V),
\ea
where $(a)$ follows from Fano's inequality; $(b)$ holds because of the chain-rule and because $M$, $M_\setY^{(p)} \!$, $M_\setY^{(c)} \!$, $M_\setZ^{(p)} \!$, and $M_\setZ^{(c)}$ are independent; $(c)$ follows from the chain-rule; $(d)$ is obtained by using Csisz\'{a}r's sum-identity twice and by using that mutual information is nonnegative; $(e)$ holds because conditional entropy is nonnegative; and $(f)$ follows from \eqref{bl:auxVar}. Our last sum-rate constraint is that
\ba 
&2 R + R_\setY + R_\setZ - \epsilon_n \nonumber \\
&\quad \stackrel{(a)}\leq \frac{1}{n} I \bigl(M, M_\setY^{(c)};Y^n,M_\setZ^{(c)}\bigr) + H (Y|U) + I (U;Z) \\
&\quad \stackrel{(b)}\leq I \bigl(M, M_\setY^{(c)};Y_Q\bigl|M_\setZ^{(c)} \!,Y^n_{Q+1},Q\bigr) + H (Y|U) + I (U;Z) \nonumber \\
&\quad \stackrel{(c)}\leq I (V;Y) + H (Y|U) + I (U;Z),
\ea
where $(a)$ follows from Fano's inequality and \eqref{eq:ryPrivRz}; $(b)$ follows from the chain-rule and the independence of $M$, $M_\setY^{(c)}$, and $M_\setZ^{(c)}$; and $(c)$ holds because conditioning cannot increase entropy.

So far, we have shown that the set of rate-tuples satisfying $\eqref{bl:capRegPartMessSI}$ for some PMF of the form \eqref{eq:partMessSIPMF} is an outer bound on the capacity region of the SD-BC with P-MSI. To conclude, it remains to establish that we can w.l.g.\ restrict $X$ to be a function of $(Y,U)$. To this end we note that, by the Functional Representation lemma (Lemma~\ref{le:funcRep}), there exists some $\hat{U}$ that is of finite support and independent of $(Y,U)$ for which $X$ is a function of $(Y, U, \hat{U})$. Further, we note that the RHS of each constraint in \eqref{bl:capRegPartMessSI} is either unaffected or increases if we replace $U$ by the pair~$(U,\hat{U})$. From this we conclude that we can w.l.g.\ restrict $X$ to be a function of $(Y,U)$.

\section{Binning the Cloud-Center (or Rate-Splitting) is Necessary to Achieve the Capacity Region of Theorem~\ref{th:partMessSI}} \label{sec:binCCNeeded}

Throughout this section, we consider the case where the encoder conveys only private messages ($R = 0$). We show that with a coding scheme as in Appendix~\ref{sec:dir} but where the cloud-center is not binned, we cannot---in general---achieve the capacity region $\mathscr C_{\textnormal {P-MSI}}$ of the SD-BC with P-MSI. Recall from the proof-sketch after Theorem~\ref{th:partMessSI} that---without binning the cloud-center---we can achieve every rate-tuple $(R_\setY^{(p)} \!, R_\setY^{(c)} \!, R_\setZ^{(p)} \!, R_\setZ^{(c)})$ that satisfies \eqref{bl:capRegPartMessSI} and \eqref{bl:ccNotBinned} for some PMF $p (v,u,x,y,z)$ of the form \eqref{eq:partMessSIPMF}. Moreover, we can show that w.l.g.\ we can restrict $X$ to be a function of $(Y,U)$ (this follows from the Functional Representation lemma (Lemma~\ref{le:funcRep}); a similar argument can be found in Appendix~\ref{sec:conv}). Note that the rate-region that can be achieved without binning the cloud-center is contained in the capacity region $\mathscr C_{\textnormal {P-MSI}}$ of the SD-BC with P-MSI. As the following example shows, the containment can be strict:

\begin{example}[Binning the cloud-center] \label{ex:binningCCNeeded}
Consider the SD-BC with binary input $X$ and binary outputs $Y = X$ and
\ba
Z = \begin{cases} X &\text{if } S = 0, \\ ? &\text{if } S = 1, \end{cases} \label{eq:exBinningCCNeededDefZ}
\ea
where $S \sim \ber (p), \,\, p \in [ 0,1 ]$ is independent of $X$, and assume that the encoder conveys only private messages ($R = 0$). By Corollary~\ref{co:fullMessSIRecZCommMess} (with $X \sim \ber (1/2)$) the capacity region $\mathscr C_{\textnormal {P-MSI}}$ with F-MSI at the stochastic receiver~$\setZ$ ($R_\setY^{(p)} \!= 0$) is the set of rate-tuples $(0, R_\setY, R_\setZ^{(p)} \!, R_\setZ^{(c)})$ that satisfy
\begin{subequations}
\ba 
R_\setY &\leq 1 \\
R_\setZ &\leq 1 - p \\
R_\setY + R_\setZ^{(p)} \!&\leq 1.
\ea
\end{subequations}
This implies that, irrespective of $p \in [0,1]$, we can achieve the rate-tuple
\ba
(0, R_\setY, R_\setZ^{(p)} \!, R_\setZ^{(c)}) = (0, p, 1 - p, 0). \label{eq:rateTupleWithBinning}
\ea
As we argue shortly, if we do not bin the cloud-center, then we can achieve the rate-tuple \eqref{eq:rateTupleWithBinning} only in the degenerate cases where $p \in \{ 0,1 \}$. This implies that, for every $p \in (0,1)$, if we do not bin the cloud-center, then we can achieve only a strict subset of the capacity region.

We next show that, if we do not bin the cloud-center, then we can achieve the rate-tuple \eqref{eq:rateTupleWithBinning} only in the degenerate cases where $p \in \{ 0,1 \}$. To this end recall that---without binning the cloud-center---the achievable rate-tuples $(R_\setY^{(p)} \!, R_\setY^{(c)} \!, R_\setZ^{(p)} \!, R_\setZ^{(c)})$ are the ones that satisfy \eqref{bl:capRegPartMessSI} and \eqref{bl:ccNotBinned} for some PMF $p (v,u,x,y,z)$ of the form \eqref{eq:partMessSIPMF}, and where one can w.l.g.\ restrict $X$ to be a function of $(Y,U)$. Fix some PMF $p (v,u,x,y,z)$ of the form \eqref{eq:partMessSIPMF} that satisfies that $X$ is a function of $(Y,U)$, and assume that $p > 0$. As we argue shortly, \eqref{eq:ryRzp} implies that, for every $p \in (0,1]$, at least one of the following two holds:
\ba 
R_\setY + R_\setZ^{(p)} \!< H (X) \leq 1 \quad \text{or} \quad I (U;X|V) = 0. \label{eq:binningCCNeededC1}
\ea
Moreover, it follows from \eqref{eq:ccNotBinnedRZp} that
\ba 
R_\setZ^{(p)} \!&\stackrel{(a)} \leq I (U;Z|V) \\
&\stackrel{(b)} = I (U;Z,S|V) \\
&\stackrel{(c)} = I (U;Z|S,V) \\
&\stackrel{(d)} = (1 - p) I (U;X|V), \label{eq:binningCCNeededRZ}
\ea 
where $(a)$ is \eqref{eq:ccNotBinnedRZp}; $(b)$ holds because $S$ is a function of $Z$; $(c)$ holds because $S$ is independent of $(V,U)$; and $(d)$ follows from \eqref{eq:exBinningCCNeededDefZ}. From \eqref{eq:binningCCNeededC1} and \eqref{eq:binningCCNeededRZ} it follows that, for every $p \in ( 0, 1 ]$, at least one of the following two holds:
\ba 
R_\setY + R_\setZ^{(p)} \!< H (X) \leq 1 \quad \text{or} \quad R_\setZ^{(p)} \!= 0.
\ea
Consequently, for every $p \in ( 0, 1 ]$, if $R_\setZ^{(p)}$ is strictly positive, then $R_\setY + R_\setZ^{(p)}$ must be strictly smaller than $1$. In particular, this implies that if we do not bin the cloud-center, then we can achieve the rate-tuple $(0, R_\setY, R_\setZ^{(p)} \!, R_\setZ^{(c)}) = (0, p, 1 - p, 0)$ only in the degenerate cases where $p \in \{ 0,1 \}$.

To conclude, it remains to show that, for every $p \in (0,1]$, \eqref{eq:ryRzp} implies \eqref{eq:binningCCNeededC1}. To this end we observe from \eqref{eq:ryRzp} that
\ba 
R_\setY + R_\setZ^{(p)} \!&\stackrel{(a)}\leq I (V;Y) + H (Y|U) + I (U;Z|V) \\
&\stackrel{(b)}\leq I(X;Y,Z) \\
&\stackrel{(c)}= H (X) \\
&\stackrel{(d)}\leq 1, \label{eq:exBinningCCNeededImp1}
\ea
where $(a)$ is \eqref{eq:ryRzp}; $(b)$ follows from \eqref{eq:ub1c} in the proof of Corollary~\ref{co:fullMessSIRecZCommMess}, which can be found in Appendix~\ref{sec:pfCoFullMessSIRecZ}; $(c)$ holds because $Y = X$; and $(d)$ holds because $X$ is binary. Note that \eqref{eq:exBinningCCNeededImp1} can hold with equality only if the following two equalities hold
\ba 
I (U;Y|Z,V) = I (V;Z|Y) = 0. \label{eq:exBinningCCNeededImp1Eq}
\ea
Indeed, Inequality~$(d)$ in the derivation of \eqref{eq:ub1c} holds with equality only if \eqref{eq:exBinningCCNeededImp1Eq} holds, and hence Inequality~$(b)$ in the derivation of \eqref{eq:exBinningCCNeededImp1} holds with equality only if \eqref{eq:exBinningCCNeededImp1Eq} holds. Note that
\ba 
I (U;Y|Z,V) &\stackrel{(a)} = I (U;X|Z,V) \\
&\stackrel{(b)} = (1 - p) I (U;X|X,V) + p \, I (U;X|Z = ?,V) \\
& = p \, I (U;X|V), \label{eq:exBinningCCNeededImp1Reform}
\ea
where $(a)$ holds because $Y = X$; and $(b)$ follows from \eqref{eq:exBinningCCNeededDefZ}. We are now ready to conclude the proof of our claim that, for every $p \in (0,1]$, \eqref{eq:ryRzp} implies \eqref{eq:binningCCNeededC1}: We have shown that \eqref{eq:ryRzp} implies \eqref{eq:exBinningCCNeededImp1}. Moreover, \eqref{eq:exBinningCCNeededImp1} can hold with equality only if \eqref{eq:exBinningCCNeededImp1Eq} holds, and by \eqref{eq:exBinningCCNeededImp1Reform} this implies that, for every $p \in (0,1]$, \eqref{eq:exBinningCCNeededImp1} can hold with equality only if $I (U;X|V) = 0$. Consequently, \eqref{eq:binningCCNeededC1} holds for every $p \in (0,1]$.

\end{example}

\section{Proof of Corollary~\ref{co:fullMessSIRecZCommMess}} \label{sec:pfCoFullMessSIRecZ}

For $U = X$ and $V = Y$ the constraints in \eqref{bl:capRegPartMessSI} and \eqref{bl:capRegFullMessSIRecZCommMess} are equivalent. Hence, \eqref{bl:capRegFullMessSIRecZCommMess} is an inner bound on the capacity region. We next argue that \eqref{bl:capRegFullMessSIRecZCommMess} is also an outer bound on the capacity region. To this end fix any PMF of the form \eqref{eq:partMessSIPMF} satisfying that $X$ is a function of $(Y,U)$. By \eqref{eq:partMessSIPMF} $U$, $X$, and $Z$ form a Markov chain in that order, and hence
\ba
I (U;Z) \leq I (X;Z). \label{eq:ub1b}
\ea
Moreover,
\ba 
&I (V;Y) + H (Y|U) + I (U;Z|V) \nonumber \\
&\quad \stackrel{(a)}= H (Y|U) + I (U;Y,Z) - I (U;Y|V) \nonumber \\
&\qquad - I (U;Z|Y) + I (U;Z|V) \\
&\quad \stackrel{(b)}= H (Y|U) + I (U;Y,Z) - I (U;Y|V) \nonumber \\
&\qquad - I (X;Z|Y) + I (U;Z|V) \\
&\quad \stackrel{(c)}= H (Y|U) + I (U;Y,Z) - I (U;Y,Z|V) \nonumber \\
&\qquad - I (V;Z|Y) + I (U;Z|V) \\
&\quad \stackrel{(d)}\leq H (Y|U) + I (U;Y,Z) \\
&\quad \stackrel{(e)}= H (Y) + I (U;Z|Y) \\
&\quad \stackrel{(f)}= I (X;Y,Z), \label{eq:ub1c}
\ea
where $(a)$ follows from the chain-rule and the fact that, by \eqref{eq:partMessSIPMF}, $V$, $U$, and $Y$ form a Markov chain in that order; $(b)$ holds because $X$ is a function of $(Y,U)$ and because, by \eqref{eq:partMessSIPMF}, $U$, $X$, and $Z$ form a Markov chain in that order; $(c)$ holds because
\ba
I (X;Z|Y) - I (U;Z|Y,V) &= I (X;Z|Y) - I (X;Z|Y,V) \\
&= I (V;Z|Y),
\ea
where we used that $X$ is a function of $(Y,U)$; $(d)$ holds because mutual information is nonnegative and because conditioning cannot increase entropy; $(e)$ follows from the chain-rule; and $(f)$ holds because $Y$ is a function of $X$, because $X$ is a function of $(Y,U)$, and because of the chain-rule. From \eqref{eq:ry}, \eqref{eq:rz} and \eqref{eq:ub1b}, as well as \eqref{eq:ryRzp} and \eqref{eq:ub1c} we conclude that \eqref{bl:capRegFullMessSIRecZCommMess} is an outer bound on the capacity region.

\section{Analysis of Example~\ref{ex:commMessPMSIRecYIncCap}} \label{sec:pfExCommMessPMSIRecYIncCap}

\begin{itemize} 
\item \underline{Assume F-MSI at $\setY$ and no MSI at $\setZ$ ($R_\setY^{(p)} = R_\setY$ and $R_\setZ^{(p)} = 0$).} The capacity region is given in Corollary~\ref{co:fullMessSIRecYCommMess}, and it is not hard to see that the maximum achievable sum-rate $R+R_\setY+R_\setZ$ is %(see \eqref{eq:fullMessSIRecYRypRzCommMess})
\ba 
\max_{p(u,x,y,z)\in\mathcal{P}_u}\ \bigl\{ H (Y|U) + I (U;Z) \bigr\}.\label{eq:pfExCommMessMaxSRFMSIY}
\ea
In the proof of Example~\ref{ex:addEras} (see Equations~\eqref{eq:maxSR2}--\eqref{bl:condMaxSumRExFBInc} in Appendix~\ref{sec:proof}), the RHS of \eqref{eq:pfExCommMessMaxSRFMSIY} is shown to equal the RHS of \eqref{eq:maxSRFullMessSIRecYCM}. Further, it is shown that every PMF $p (u, x, y, z) \in \mathcal{P}_{u}^\star$ satisfies \eqref{eq:maxSRFullMessSIRecYCMMutis}. Finally, we obtain \eqref{eq:Rcommonstar} by examining \eqref{bl:capRegFullMessSIRecYCommMess}.

\item \underline{Assume no MSI ($R_\setY^{(p)} = R_\setY$ and $R_\setZ^{(p)} = R_\setZ$).} By Theorem~\ref{th:partMessSI} the capacity region is the set of rate-tuples $( R, R_\setY, R_\setZ )$ satisfying
\begin{subequations}\label{bl:capRegNoMSICommMess}
\ba
R+R_\setY &\leq H (Y) \label{eq:ryCM} \\
R+R_\setZ &\leq I (U;Z) \label{eq:rzCM} \\
R+R_\setY+R_\setZ &\leq I (V;Y) + H (Y|U) + I (U;Z|V) \label{eq:ryRzCM} \\
R+R_\setY+R_\setZ &\leq H (Y|U) + I (U;Z) \label{eq:rypRzCM} \\
2 R+R_\setY+R_\setZ &\leq I (V;Y) + H (Y|U) + I (U;Z) \label{eq:RlastCM}
\ea
\end{subequations}
for some PMF $p (v,u,x,y,z)$ of the form \eqref{eq:partMessSIPMF}. From~\eqref{eq:rypRzCM} we see that the maximum achievable sum-rate cannot be larger than the maximum sum-rate in \eqref{eq:pfExCommMessMaxSRFMSIY}, and that it can be achieved only if $p(u,x,y,z)\in\mathcal{P}_u$. If we set $R=0$ (no common message) and let $V$ be deterministic, then we see that \eqref{eq:pfExCommMessMaxSRFMSIY} is achievable, and hence we obtain that \eqref{eq:pfExCommMessMaxSRFMSIY} is the maximum achievable sum-rate. Finally, \eqref{eq:commMessPMSIRecYIncCapRSR} follows from \eqref{eq:RlastCM}, because the maximum achievable sum-rate is \eqref{eq:pfExCommMessMaxSRFMSIY}.

%Using that \eqref{eq:rypRzCM} is the same as \eqref{eq:pfExCommMessMaxSRFMSIY}, and that the marginal PMF $p (u,x,y,z)$ of any PMF $p (v,u,x,y,z)$ of the form \eqref{eq:partMessSIPMF} must be of the form \eqref{eq:fullMessSIRecYCommMess}, we find that the maximum achievable sum-rate is given by \eqref{eq:maxSRFullMessSIRecYCM}; and that 

%Every PMF $p (v, u, x, y, z)\in \mathcal{P}_{uv}^\star$ satisfies\eqref{eq:maxSRFullMessSIRecYCMMutis}. 
\end{itemize}

\section{Proof of Proposition~\ref{pr:enhMSI}}\label{sec:pfPrEnhMSI}

We assume that $R_\setY^{(c)} \!> 0$, because otherwise the statement is obvious. We will show that the error probability of the described code with rate-limited feedback converges to zero as the blocklength approaches infinity. To this end we introduce the following three error events:
\begin{IEEEeqnarray}{rCl}
E_0 &:& \quad \text{Receiver~$\setY$ cannot recover $Z^{\alpha n} \!$} \\
E_1 &:& \quad \text{The messages sent in Phase~1 are not decoded correctly} \\
E_2 &:& \quad \text{The messages sent in Phase~2 (including the compression index)} \\ \nonumber \\
& &\quad \text{are not decoded correctly}.
\end{IEEEeqnarray}
Using these error events, we can upper-bound the code's error probability by
\ba 
\distof {\text{error}} &\leq \distof {E_0 \cup E_1 \cup E_2} \\
&= \distof {E_2} + \distof {E_0 | E_2^c} + \distof {E_1|E_0^c}. \label{eq:fbCodeErrorProb}
\ea
To show that the code's error probability converges to zero, we will show that each term on the RHS of \eqref{eq:fbCodeErrorProb} converges to zero. For the third term we have
\ba 
\distof {E_1|E_0^c} \rightarrow 0 \,\, (n \rightarrow \infty),
\ea
because $( \tilde R_\setY^{(p)} \!, \tilde R_\setY^{(c)} \!, \tilde R_\setZ^{(p)} \!, \tilde R_\setZ^{(c)} ) \in \mathscr C^{( \textnormal{enh} )}_{\textnormal{P-MSI}}$. Also, because Event~$E_2^c$ occurs only if Receiver~$\setY$ recovers the correct compression index, and because a compression index of rate $H (Z|Y)$ is enough for Receiver~$\setY$ to recover $Z^{\alpha n}$ using also the side-information $Y^{\alpha n}$, the second term satisfies
\ba 
\distof {E_0|E_2^c} \rightarrow 0 \,\, (n \rightarrow \infty).
\ea
Finally, we note that the first term satisfies
\ba 
\distof {E_2} \rightarrow 0 \,\, (n \rightarrow \infty),
\ea
because
\ba
\Bigl( R_\setY^{(p)} \!, \hat R_\setY^{(c)} + \frac{\alpha}{1 - \alpha} H (Z|Y), R_\setZ^{(p)} \!, R_\setZ^{(c)} \Bigr) \in \mathscr C_{\textnormal {P-MSI}},
\ea
where we used \eqref{eq:hatRYc} and that $( R_\setY^{(p)} \!, R_\setY^{(c)} \!, R_\setZ^{(p)} \!, R_\setZ^{(c)} ) \in \mathscr C_{\textnormal {P-MSI}}$. (The ``rate'' $\frac{\alpha}{1-\alpha} H (Z|Y)$ accounts for the $\alpha n H (Z|Y)$ compression bits from Phase~1 that have to be sent during the $(1 - \alpha) n$ channel uses that Phase~2 comprises.)

\section{Proof of Propositions~\ref{pr:partSIRecYInc} and \ref{pr:noSIFBInc}}\label{sec:pfSufficientConditions}

Before we prove Propositions~\ref{pr:partSIRecYInc} and \ref{pr:noSIFBInc}, we sketch our proof of Proposition~\ref{pr:partSIRecYInc}. A first step in the proof is to note that---under the conditions stated in Proposition~\ref{pr:partSIRecYInc}---it is possible to identify a rate-triple $( R_\setY^{(p)} \!, R_{\setZ}^{(p)} \!, R_\setZ^{(c)} )$ that satisfies the following two:
\begin{enumerate}
\item $(R_\setY^{(p)} \!, R_\setZ^{(p)} \! + R_{\setZ}^{(c)} )$ lies on the boundary of the no-feedback capacity region $\mathscr C$ of the SD-BC without MSI and in the interior of the capacity region $\mathscr C_{\textnormal{enh}}$ of the enhanced BC without MSI.
\item There exists a strictly-positive rate $R_\setY^{(c)}$ for which $(R_\setY^{(p)} \!, R_\setY^{(c)} \!, R_\setZ^{(p)} \!, R_\setZ^{(c)} )$ is contained in $\mathscr C_{\textnormal{P-MSI}}$.
\end{enumerate}
As we argue next, the above conditions guarantee that, irrespective of $R_{\textnormal{FB}} > 0$, the feedback code of Section~\ref{sec:scheme} achieves a rate-tuple outside the no-feedback capacity region $\mathscr C$ of the SD-BC without MSI. By Remark~\ref{re:partMessSIRecy} $\mathscr{C}$ is also the no-feedback capacity region of the SD-BC with P-MSI at Receiver~$\setY$ only, and hence it follows that the conditions in Proposition~\ref{pr:partSIRecYInc} are sufficient for feedback to increase the capacity region of the SD-BC with P-MSI at Receiver~$\setY$ only. It remains to show that the feedback code of Section~\ref{sec:scheme} achieves a rate-tuple outside the no-feedback capacity region $\mathscr C$. Condition~1 guarantees that in Phase~1 of the feedback code, when the encoder codes for the enhanced BC, message information can be sent at rates $(\tilde{R}_{\setY}, 0, \tilde{R}_{\setZ}^{(p)} \!, \tilde{R}_{\setZ}^{(c)} )$ satisfying
\begin{equation}\label{eq:ineq}
\tilde{R}_{\setY} > R_\setY^{(p)} \!, \quad \tilde{R}_{\setZ}^{(p)} \! > R_\setZ^{(p)} \!, \quad \tilde{R}_{\setZ}^{(c)} \! > R_{\setZ}^{(c)} \!.
\end{equation}
Recall that the deterministic receiver~$\setY$ needs resolution information to decode its Phase-$1$ message. Condition~2 guarantees that in Phase~2 of the feedback code the encoder can send the required resolution information at rate $R_\setY^{(c)}$ and, simultaneously, fresh message-information at rates $( R_\setY^{(p)} \!, 0, R_{\setZ}^{(p)} \!, R_\setZ^{(c)} )$. Specifically, Proposition~\ref{pr:enhMSI} implies that, for every sufficiently-small $\alpha$, the rate-tuple 
\ba 
\alpha ( \tilde R_\setY, 0, \tilde R_\setZ^{(p)} \!, \tilde R_\setZ^{(c)} ) + (1-\alpha) ( R_\setY^{(p)} \!, 0, R_\setZ^{(p)} \!, R_\setZ^{(c)} ) \label{eq:FBRateTuple2}
\ea
is achievable. Since $( R_\setY^{(p)} \!, R_\setZ^{(p)} \! + R_{\setZ}^{(c)} )$ lies on the boundary of the no-feedback capacity region $\mathscr C$, we conclude from \eqref{eq:ineq} that the rate-tuple \eqref{eq:FBRateTuple2} is not achievable without feedback. This proves that the feedback code of Section~\ref{sec:scheme} achieves a rate-tuple outside the no-feedback capacity region $\mathscr C$.

The proof of Proposition~\ref{pr:noSIFBInc} is similar: the main difference is that we set $R_\setZ^{(c)} \!= \tilde R_\setZ^{(c)} \!= 0$.\\

We next provide the details of the proof of Propositions~\ref{pr:partSIRecYInc} and \ref{pr:noSIFBInc}. We shall use the following lemma:

\begin{lemma}\label{le:part1}
Fix a rate-tuple $( R_\setY^{(p)} \!, 0, R_\setZ^{(p)} \!, R_\setZ^{(c)} ) \in \mathscr C_{\textnormal{P-MSI}}$. If $R_\setZ^{(p)} \!< R_\setZ$ (or, equivalently, $R_\setZ^{(c)} \!> 0$), and if for some PMF $p(u,x,y,z)$ of the form \eqref{eq:noMessSIPMF} we have
\begin{subequations}\label{bl:lePartSIRecY}
\ba 
R_\setY^{(p)} \!&< H (Y) \label{eq:lePartSIRecYRecY} \\
R_\setZ^{(p)} \!+ R_\setZ^{(c)} \!&\leq I (U;Z) \label{eq:lePartSIRecYRecZ} \\
R_\setY^{(p)} \!+ R_\setZ^{(p)} \!+ R_\setZ^{(c)} \!&\leq H (Y|U) + I (U;Z) \label{eq:lePartSIRecYSumRate} \\
I (U;Y) &> 0,
\ea
\end{subequations}
then there exists some positive rate ${R}_\setY^{(c)} \!> 0$ satisfying
\begin{equation}\label{eq:new_ach}
( R_\setY^{(p)} \!, {R}_\setY^{(c)} \!, R_\setZ^{(p)} \!, R_\setZ^{(c)} ) \in \mathscr C_{\textnormal{P-MSI}}.
\end{equation}

If $R_\setZ^{(p)} \!= R_\setZ$ (or, equivalently, $R_\setZ^{(c)} \!= 0$), and if for some PMF $p (v,u,x,y,z)$ of the form \eqref{eq:partMessSIPMF} we have \eqref{bl:lePartSIRecY} and
\ba 
I (V;Y) - I (V;Z) > 0, \label{eq:leNoSI}
\ea
then there exists some positive rate $ R_\setY^{(c)} \!> 0$ satisfying \eqref{eq:new_ach}.
\end{lemma}

\begin{proof}[Proof of Lemma~\ref{le:part1}]
Suppose that \eqref{bl:lePartSIRecY} holds for some PMF $p(u,x,y,z)$ of the form \eqref{eq:noMessSIPMF}. A rate-tuple $\bigl( R_\setY^{(p)},R_\setY^{(c)},R_\setZ^{(p)},R_\setZ^{(c)} \bigl)$ is in $\mathscr C_{\textnormal{P-MSI}}$ if it satisfies \eqref{bl:capRegPartMessSI} for some conditional PMF $p (v|u)$ and
\ba 
p (v,u,x,y,z) = p (v|u) \, p(u,x,y,z).
\ea
By inspection of \eqref{bl:capRegPartMessSI} and \eqref{bl:lePartSIRecY}, we see that this holds if
\begin{subequations} \label{bl:pfleConds}
\ba
R_\setY^{(c)} &\leq H (Y) \! - \! R_\setY^{(p)} \\
R_\setY^{(c)} &\leq H (Y|U) \! + \! I (U;Z) \! + \! I (V;Y) \! - \! I (V;Z) \! - \! R_\setY^{(p)} \! - \! R_\setZ^{(p)} \\
R_\setY^{(c)} &\leq H (Y|U) \! + \! I (U;Z) \! + \! I (V;Y) \! - \! R_\setY^{(p)} \! -\! R_\setZ^{(p)} \! - \! R_\setZ^{(c)}.
\ea
\end{subequations}
And from \eqref{bl:lePartSIRecY} it follows that if
\ba 
I (V;Z) - I (V;Y) < R_\setZ^{(c)} \quad \text{and} \quad I (V;Y) > 0, \label{eq:lePartSIRecYSuff}
\ea
then there exists some $R_\setY^{(c)} > 0$ for which the rate-tuple $\bigl( R_\setY^{(p)},R_\setY^{(c)},R_\setZ^{(p)},R_\setZ^{(c)} \bigl)$ satisfies \eqref{bl:pfleConds}. This proves the claim for the case where $R_\setZ^{(p)} \!= R_\setZ$ (or, equivalently, $R_\setZ^{(c)} \!= 0$), because in this case \eqref{eq:lePartSIRecYSuff} and \eqref{eq:leNoSI} are equivalent. To prove the claim for the case where $R_\setZ^{(p)} \!< R_\setZ$ (or, equivalently, $R_\setZ^{(c)} \!> 0$), fix $\varepsilon \in (0,1)$, let $S \sim \ber (1-\varepsilon)$ be independent of $(U,X,Y,Z)$, and choose
\ba 
V = \begin{cases} U & S = 0, \\ ? & S = 1. \end{cases}
\ea
For this choice of $V$ \eqref{eq:lePartSIRecYSuff} holds for every sufficiently-small $\varepsilon$, because
\bas 
I (V;Z) = \varepsilon I (U;Z) \leq \varepsilon \log |\setZ| \quad \text{and} \quad
0 < I (V;Y) = \varepsilon I (U;Y).
\eas
This proves the claim for the case where $R_\setZ^{(p)} \!< R_\setZ$.
\end{proof}

Note that if $R_\setZ^{(p)} \!< R_\setZ$, then Lemma~\ref{le:part1} does not ask for much: By the assumption that $( R_\setY^{(p)} \!, 0, R_\setZ^{(p)} \!, R_\setZ^{(c)} ) \in \mathscr C_{\textnormal{P-MSI}}$, there must exist some PMF $p(u,x,y,z)$ of the form \eqref{eq:noMessSIPMF} for which \eqref{eq:lePartSIRecYRecZ} and \eqref{eq:lePartSIRecYSumRate} are satisfied and \eqref{eq:lePartSIRecYRecY} holds with nonstrict inequality (this follows from Remark~\ref{re:partMessSIRecy} and Corollary~\ref{co:noMessSI}). Hence, all we are asking for is that $R_\setY < H (Y)$ and $I (U;Y) > 0$ hold. (Roughly speaking, all we are asking for is that the transmission of $M_\setZ$ interfere with the transmission of $M_\setY$.)

\begin{proof}[Proof of Proposition~\ref{pr:partSIRecYInc} and \ref{pr:noSIFBInc}]
Assume that there exists some rate-tuple $( R_\setY^{(p)} \!, R_\setY^{(c)} \!, R_\setZ^{(p)} \!, R_\setZ^{(c)} )$ that satisfies the following three conditions:\footnote{In particular, \eqref{eq:cond1} implies that $( R_\setY^{(p)} \!, R_\setZ^{(p)} \!, R_\setZ^{(c)} ) \in \partial \mathscr C^{(R_\setY^{(c)} \!= 0)}_{\textnormal{P-MSI}}$, where $\mathscr C^{(R_\setY^{(c)} \!= 0)}_{\textnormal{P-MSI}} \subset (\reals_0^+)^3$ denotes the set of rate-triples $( R_\setY^{(p)} \!, R_\setZ^{(p)} \!, R_\setZ^{(c)} )$ satisfying $( R_\setY^{(p)} \!, 0, R_\setZ^{(p)} \!, R_\setZ^{(c)} ) \in \mathscr C_{\textnormal{P-MSI}}$ (cf.~\eqref{eq:equivalent}).\label{footnotelabel}}
\begin{subequations}\label{eq:cond}
\ba
R_\setY^{(c)} \!&> 0 \label{eq:cond3} \\
( R_\setY^{(p)} \!, R_\setY^{(c)} \!,R_\setZ^{(p)} \!, R_\setZ^{(c)} ) &\in \mathscr C_{\textnormal{P-MSI}} \label{eq:cond2} \\
( R_\setY^{(p)} \!, R_\setZ^{(p)} \!+ R_\setZ^{(c)}) \!&\in \bigl( \partial \mathscr C \cap ( \mathscr C_{\textnormal {enh}} \setminus \partial \mathscr C_{\textnormal {enh}} )\bigr). \label{eq:cond1}
\ea
\end{subequations}
Now look for rates
\begin{equation}\label{eq:larger_rates}
\tilde R_\setY^{(p)} \!> R_\setY^{(p)} \!, \quad \tilde R_\setZ^{(p)} \!> R_\setZ^{(p)} \!, \quad \tilde R_\setZ^{(c)} \!> R_\setZ^{(c)}
\end{equation}
satisfying that $(\tilde R_\setY^{(p)} \!, \tilde R_\setZ^{(p)} \!+ \tilde R_\setZ^{(c)}) \in \mathscr C_{\textnormal {enh}}$. By Assumption~\eqref{eq:cond1} such rates exist.

Proposition~\ref{pr:enhMSI} and Assumptions~\eqref{eq:cond3} and \eqref{eq:cond2} guarantee that, irrespective of $R_{\textnormal{FB}} > 0$, there exists some $\alpha \in (0,1)$ for which the rate-tuple
\ba 
\alpha (\tilde R_\setY^{(p)} \!, 0, \tilde R_\setZ^{(p)} \!, \tilde R_\setZ^{(c)} ) + (1- \alpha) ( R_\setY^{(p)} \!,0 , R_\setZ^{(p)} \!, R_\setZ^{(c)} ) \label{eq:rTuple}
\ea
is in the feedback capacity region of the SD-BC with P-MSI at Receiver~$\setY$. Note that \eqref{eq:equivalent}, \eqref{eq:cond1}, and \eqref{eq:larger_rates} imply that the rate-tuple in \eqref{eq:rTuple} is not in $\mathscr C_{\textnormal{P-MSI}}$ (see also Footnote~\ref{footnotelabel}). To conclude the proof of Proposition~\ref{pr:partSIRecYInc}, it thus suffices to show that, under the conditions in Proposition~\ref{pr:partSIRecYInc}, it is possible to find a rate-tuple $( R_\setY^{(p)} \!, R_\setY^{(c)} \!, R_\setZ^{(p)} \!, R_\setZ^{(c)} )$ satisfying \eqref{eq:cond}. But this follows from the first part of Lemma~\ref{le:part1}.\\

The proof of Proposition~\ref{pr:noSIFBInc} is similar. Assume that there exists some rate-tuple $( R_\setY^{(p)} \!, R_\setY^{(c)} \!, R_\setZ^{(p)} \!, 0 )$ that satisfies the following three conditions:
\begin{subequations}\label{eq:condsNoSI}
\ba
R_\setY^{(c)} \!&> 0 \label{eq:cond3NoSI} \\
( R_\setY^{(p)} \!, R_\setY^{(c)} \!,R_\setZ^{(p)} \!, 0 ) &\in \mathscr C_{\textnormal{P-MSI}} \label{eq:cond2NoSI} \\
( R_\setY^{(p)} \!, R_\setZ^{(p)} ) &\in \bigl( \partial \mathscr C \cap ( \mathscr C_{\textnormal {enh}} \setminus \partial \mathscr C_{\textnormal {enh}} )\bigr). \label{eq:cond1NoSI}
\ea
\end{subequations}
Now look for rates
\begin{equation}\label{eq:larger_ratesNoSI}
\tilde R_\setY^{(p)} \!> R_\setY^{(p)} \!, \quad \tilde R_\setZ^{(p)} \!> R_\setZ^{(p)}
\end{equation}
satisfying that $(\tilde R_\setY^{(p)} \!, \tilde R_\setZ^{(p)} ) \in \mathscr C_{\textnormal {enh}}$. By Assumption~\eqref{eq:cond1NoSI} such rates exist.

Proposition~\ref{pr:enhMSI} and Assumptions~\eqref{eq:cond3NoSI} and \eqref{eq:cond2NoSI} guarantee that, irrespective of $R_{\textnormal{FB}} > 0$, there exists some $\alpha \in (0,1)$ for which the rate-pair
\ba 
\alpha (\tilde R_\setY^{(p)} \!, \tilde R_\setZ^{(p)} ) + (1 - \alpha) ( R_\setY^{(p)} \!, R_\setZ^{(p)} ) \label{eq:rTupleNoSI}
\ea
is in the feedback capacity region of the SD-BC with P-MSI at Receiver~$\setY$. Note that \eqref{eq:cond1NoSI} and \eqref{eq:larger_ratesNoSI} imply that the rate-tuple in \eqref{eq:rTupleNoSI} is not in $\mathscr C$. To conclude the proof of Proposition~\ref{pr:noSIFBInc}, it thus suffices to show that, under the conditions in Proposition~\ref{pr:noSIFBInc}, it is possible to find a rate-tuple $( R_\setY^{(p)} \!, R_\setY^{(c)} \!, R_\setZ^{(p)} \!, 0 )$ satisfying \eqref{eq:condsNoSI}. But this follows from the second part of Lemma~\ref{le:part1}.
\end{proof}

\section{Analysis of Example~\ref{ex:addEras}}\label{sec:proof}

The proof of Example~\ref{ex:addEras} hinges on Propositions~\ref{pr:partSIRecYInc} and \ref{pr:noSIFBInc}, which state sufficient conditions for feedback to increase the capacity region of the SD-BC with and without P-MSI at Receiver~$\setY$. It can be roughly outlined as follows: Let $\bar p = 1 - p$. For the SD-BC of Example~\ref{ex:addEras}, we first parametrize $\partial \mathscr C \cap \bigl\{ R_\setZ \in [ R_\setZ^\ast, \bar p ] \bigr\}$ by $R_\setZ$, where $R_\setZ^\ast$ denotes the maximum rate $R_\setZ$ for which we can achieve the sum-rate capacity without MSI, and where $\bar p$ is the maximum rate $R_\setZ$ that is achievable without MSI. We then show that for the boundary point with $R_\setZ = R_\setZ^\ast$ we can satisfy all the conditions in Proposition~\ref{pr:partSIRecYInc}, and that for some boundary point with $R_\setZ \in (R_\setZ^\ast, \bar p)$ we can satisfy all the conditions in Proposition~\ref{pr:noSIFBInc}. This allows us to conclude from Propositions~\ref{pr:partSIRecYInc} and \ref{pr:noSIFBInc} that---on the considered SD-BC---feedback increases the sum-rate capacity with P-MSI at Receiver~$\setY$ and the capacity region without MSI.\\

Recall that the no-feedback capacity region $\mathscr C$ of the SD-BC without MSI is the set of rate-tuples satisfying \eqref{bl:capRegNoMessSI} for some PMF of the form \eqref{eq:noMessSIPMF}, and where we can restrict $X$ to be a function of $(Y,U)$. For the SD-BC of Example~\ref{ex:addEras}, the fact that $X$ is a function of $(Y,U)$ implies that we can partition the support $\setU$ of $U$ into two disjoint sets $\setU_0$ and $\setU_1$, where
\ba 
p_{X_2,Y|U} (1,1|u) = 0, \,\, \forall \, u \in \setU_0 \quad \text{and} \quad p_{X_2,Y|U} (0,1|u) = 0, \,\, \forall \, u \in \setU_1.
\ea
For each $u \in \setU_0$ introduce
\begin{subequations}
\ba 
p_0 (u) &= p_{X_2,Y|U} (0,0|u) \\
p_1 (u) &= p_{X_2,Y|U} (0,1|u) \\
p_2 (u) &= p_{X_2,Y|U} (1,2|u),
\ea
\end{subequations}
and note that $p_0 (u) + p_1 (u) + p_2 (u) = 1$. As we argue next, we can w.l.g.\ assume that to each $u \in \setU_0$ there corresponds some $u^\prime \in \setU_1$ satisfying $p_U (u) = p_U (u^\prime)$ and
\begin{subequations}\label{bl:correspU0U1}
\ba 
p_0 (u) = p_{X_2,Y|U} (0,0|u) &= p_{X_2,Y|U} (1,2|u^\prime) \\
p_1 (u) = p_{X_2,Y|U} (0,1|u) &= p_{X_2,Y|U} (1,1|u^\prime) \\
p_2 (u) = p_{X_2,Y|U} (1,2|u) &= p_{X_2,Y|U} (0,0|u^\prime),
\ea
\end{subequations}
This allows us to simplify \eqref{eq:rzCapRegNoMessSI} and \eqref{eq:srCapRegNoMessSI} to
\begin{subequations}\label{bl:repRZRS}
\ba
R_\setZ &\leq \bar p \biggl[ 1 - 2 \sum_{u \in \setU_0} p (u) h_\textnormal{b} \bigl(p_2 (u)\bigr) \biggr] \\
R_\setY + R_\setZ &\leq \bar p + 2 \sum_{u \in \setU_0} p (u) \Biggl[ p \, h_\textnormal{b} \bigl(p_2 (u)\bigr) + \bigl( 1 - p_2 (u) \bigr) h_\textnormal{b} \biggl( \frac{p_0 (u)}{p_0 (u) + p_1 (u)} \biggr) \Biggr]. \label{eq:maxSR1}
\ea
\end{subequations}
To show that to each $u \in \setU_0$ there corresponds some $u^\prime \in \setU_1$ satisfying $p_U (u) = p_U (u^\prime)$ and \eqref{bl:correspU0U1}, we first note that all $u \in \setU_0$ and $u^\prime \in \setU_1$ satisfying \eqref{bl:correspU0U1} must satisfy $H (X_2|U=u) = H (X_2|U=u^\prime)$ and $H (Y|U=u) = H (Y|U=u^\prime)$. Using this, that $I (U;Z) = \bar p \, I (U;X_2)$, and that entropy is concave, we can now argue that the claim must hold by symmetry.

Having obtained \eqref{bl:repRZRS}, we are now ready to determine the sum-rate capacity. First, note that the RHS of \eqref{eq:maxSR1} is maximum only if $p_0 (u) = p_1 (u) = \bigl( 1 - p_2 (u) \bigr) / 2$, and that in this case \eqref{eq:maxSR1} simplifies to
\ba
R_\setY + R_\setZ &\leq \bar p + 2 \sum_{u \in \setU_0} p (u) \Bigl[ p \, h_\textnormal{b} \bigl(p_2 (u)\bigr) + \bigl( 1 - p_2 (u) \bigr) \Bigr]. \label{eq:maxSR2}
\ea
Because the function $p_2 (u) \mapsto p \, h_\textnormal{b} \bigl(p_2 (u)\bigr) + \bigl( 1 - p_2 (u) \bigr)$ is strictly concave in $p_2 (u)$, we readily find that the RHS of \eqref{eq:maxSR2} is maximum if, and only if, (iff) for all $u \in \setU_0$ we have
\ba 
p_2 (u) = \frac{1}{1+2^{1/p}}. \label{eq:p2uOpt}
\ea
Hence, we find that the sum-rate capacity is obtained by evaluating the RHS of \eqref{eq:maxSR2} for the choice \eqref{eq:p2uOpt}, and that every PMF $p (u,x,y,z)$ of the form \eqref{eq:noMessSIPMF} that achieves the maximum sum-rate must satisfy the following two:
\begin{subequations}\label{bl:condMaxSumRExFBInc}
\ba 
&0 < I (U;Y) < I (U;Z) \label{eq:condPartSIRecY} \\
&R_\setZ \leq I (U;Z) = \bar p \Biggl[ 1 - \binent { \frac{1}{1+2^{1/p}} } \Biggr] = R_\setZ^\ast,
\ea
\end{subequations}
where $R_\setZ^\ast$ is the maximum rate $R_\setZ$ for which we can achieve the sum-rate capacity without MSI.

We next parametrize $\partial \mathscr C \cap \bigl\{ R_\setZ \in [ R_\setZ^\ast, \bar p ] \bigr\}$ by $R_\setZ$, i.e., for each $R_\setZ \in [ R_\setZ^\ast, \bar p ]$ we determine
\ba 
R_\setY (R_\setZ) = \max \bigl\{ R_\setY \geq 0 \colon ( R_\setY, R_\setZ ) \in \mathscr C \bigr\}.
\ea
To this end we first show that $R_\setY (R_\setZ) + R_\setZ$ is strictly decreasing in $R_\setZ$. Indeed, because $\mathscr C$ is convex we have for all $\epsilon \in (0, R_\setZ - R_\setZ^\ast]$, for $\alpha = \epsilon / ( R_\setZ - R_\setZ^\ast )$, and for $\bar \alpha = 1 - \alpha$
\ba 
&R_\setY (R_\setZ - \epsilon) + R_\setZ - \epsilon \nonumber \\
&\quad \geq \bar \alpha \bigl( R_\setY (R_\setZ) + R_\setZ \bigr) + \alpha \bigl( R_\setY (R_\setZ^\ast) + R_\setZ^\ast \bigr) \\
&\quad > R_\setY (R_\setZ) + R_\setZ,
\ea
where the last inequality holds because $R_\setY (R_\setZ) + R_\setZ$ is maximum if $R_\setZ = R_\setZ^\ast$ and strictly smaller than its maximum if $R_\setZ > R_\setZ^\ast$. As we argue next, each pair $\bigl( R_\setY (R_\setZ), R_\setZ \bigr)$ satisfies
\begin{subequations}\label{bl:bdChar}
\ba 
R_\setY (R_\setZ) &= H (Y|U) \\
R_\setZ &= I (U;Z)
\ea
\end{subequations}
for some PMF $p (u,x,y,z)$ of the form \eqref{eq:noMessSIPMF}. Note that the claim holds for $R_\setZ = R_\setZ^\ast$, and hence we assume that $R_\setZ^\ast < R_\setZ$. Because $\bigl( R_\setY (R_\setZ), R_\setZ \bigr) \in \partial \mathscr C$, Corollary~\ref{co:noMessSI} implies that for some PMF $p (u,x,y,z)$ of the form \eqref{eq:noMessSIPMF} we must have
\ba 
R_\setZ &\leq I (U;Z) \\
R_\setY (R_\setZ) &= \min \bigl\{ H (Y|U) + I (U;Z) - R_\setZ, H (Y) \bigr\}.
\ea
For contradiction, assume that $R_\setZ < I (U;Z)$. Then,
\ba 
R_\setY (R_\setZ) + R_\setZ &\stackrel{(a)}\geq \bar \alpha \bigl( H (Y|U) + I (U;Z) \bigr) + \alpha \bigl( R_\setY (R_\setZ^\ast) + R_\setZ^\ast \bigr) \\
&\stackrel{(b)}> R_\setY (R_\setZ) + R_\setZ,
\ea
where $(a)$ holds for $\alpha = \bigl( I (U;Z) - R_\setZ \bigr) / (I (U;Z) - R_\setZ^\ast)$ and $\bar \alpha = 1 - \alpha$, because $R_\setZ = \bar \alpha \, I (U;Z) + \alpha \, R_\setZ^\ast$, and because the capacity region is convex; and $(b)$ is true because $R_\setY (R_\setZ) + R_\setZ \leq H (Y|U) + I (U;Z)$, because $\alpha > 0$, and because $R_\setZ^\ast < R_\setZ$ and $R_\setY (R_\setZ) + R_\setZ$ is strictly decreasing in $R_\setZ$. This is a contradiction, and hence the claim follows.

From \eqref{bl:repRZRS}, \eqref{eq:maxSR2}, and \eqref{bl:bdChar} we obtain that each pair $\bigl( R_\setY (R_\setZ), R_\setZ \bigr), \, R_\setZ \in [ R_\setZ^\ast, \bar p ]$ must be of the form
\begin{subequations}
\ba 
R_\setZ &= \bar p \biggl[ 1 - 2 \sum_{u \in \setU_0} p (u) h_\textnormal{b} \bigl(p_2 (u)\bigr) \biggr] \\
R_\setY (R_\setZ) &= \bar p + 2 \sum_{u \in \setU_0} p (u) \Bigl[ p \, h_\textnormal{b} \bigl(p_2 (u)\bigr) + 1 - p_2 (u) \Bigr] - R_\setZ.
\ea
\end{subequations}
This can also be written as
\begin{subequations}\label{bl:bdChar2}
\ba 
R_\setZ &= \sum_{u \in \setU_0} 2 p (u) \bar p \Bigl[ 1 - h_\textnormal{b} \bigl(p_2 (u)\bigr) \Bigr] \\
R_\setY (R_\setZ) &= \sum_{u \in \setU_0} 2 p (u) \Bigl[ h_\textnormal{b} \bigl(p_2 (u)\bigr) + 1 - p_2 (u) \Bigr],
\ea
\end{subequations}
where we used that
\ba 
\sum_{u \in \setU_0} 2 p (u) = 1.
\ea
Note that if $p_2 (u) > 1/2$, then replacing $p_2 (u)$ by $1 - p_2 (u)$ in \eqref{bl:bdChar2} does not affect the value of $R_\setZ$ but increases that of $R_\setY (R_\setZ)$. Hence, for each $u \in \setU_0$ satisfying that $p (u) > 0$ we have $p_2 (u) \in [0,1/2]$, and hence we can w.l.g.\ assume that $p_2 (u) \in [0,1/2]$ holds for all $u \in \setU_0$. Since $h_\textnormal{b} ( \cdot )$ is invertible on $[0,1/2]$ and $\bar p$ is a constant, we can thus write
\begin{subequations}
\ba 
R_\setZ &= \sum_{u \in \setU_0} 2 p (u) R_\setZ^{(u)} \\
R_\setY (R_\setZ) &= \sum_{u \in \setU_0} 2 p (u) \Biggl[ 2 - \frac{1}{\bar p} R_\setZ^{(u)} - h_\textnormal{b}^{-1} \biggl( 1 - \frac{R_\setZ^{(u)}}{\bar p} \biggr) \Biggr],
\ea
\end{subequations}
where $R_\setZ^{(u)} = \bar p \bigl[ 1 - h_\textnormal{b} (p_2 (u)) \bigr]$. Note that $R_\setZ^{(u)}$ can assume any value in $[0,\bar p]$ depending on $p_2 (u) \in [0,1/2]$. Since $h_\textnormal{b} ( \cdot )$ is strictly increasing and strictly concave on $[0,1/2]$, its inverse $h_\textnormal{b}^{-1} ( \cdot )$ is strictly increasing and strictly convex on $[0,1]$. In particular, this implies that the mapping
\ba 
x \mapsto 2 - \frac{1}{\bar p} x - \binentinv {1 - \frac{x}{\bar p}}
\ea
is strictly concave on $[0,1]$. Recalling that $R_\setY (R_\setZ)$ is the maximum rate $R_\setY$ that is achievable for a given $R_\setZ$, we thus obtain from Jensen's inequality that for each $u \in \setU_0$ satisfying that $p (u) > 0$ we must have $R_\setZ^{(u)} = R_\setZ$. This allows us to parametrize $\partial \mathscr C \cap \bigl\{ R_\setZ \in [ R_\setZ^\ast, \bar p ] \bigr\}$ by $R_\setZ$:
\ba 
R_\setY (R_\setZ) &= 2 - \frac{1}{\bar p} R_\setZ - h_\textnormal{b}^{-1} \biggl( 1 - \frac{R_\setZ}{\bar p} \biggr), \quad R_\setZ \in [R_\setZ^\ast, \bar p]. \label{eq:paramBoundaryRZ}
\ea

With \eqref{eq:paramBoundaryRZ} at hand, we next argue that for $p > 1/2$ there exists some rate-pair $\bigl( R_\setY (R_\setZ), R_\setZ \bigr)$ with $R_\setZ < \bar p$ satisfying that \eqref{bl:capRegNoMessSI} and \eqref{eq:thNoSI} of Proposition~\ref{pr:noSIFBInc} hold for some PMF $p (u,x,y,z)$ of the form \eqref{eq:noMessSIPMF} and for $V = U$. Because $\bigl( R_\setY (R_\setZ), R_\setZ \bigr) \in \partial \mathscr C$ and because for $V = U$ the PMF $p (v,u,x,y,z)$ is of the form \eqref{eq:partMessSIPMF}, this is almost enough to conclude from Proposition~\ref{pr:noSIFBInc} that feedback increases the capacity region without MSI: once we have established the claim, all that remains to be shown is that the identified rate-pair satisfies
\ba 
\bigl( R_\setY (R_\setZ), R_\setZ \bigr) \in \mathscr C_{\textnormal {enh}} \setminus \partial \mathscr C_{\textnormal {enh}}.
\ea
To establish the claim, let $\setU_0 = \{ u \}$. In particular, this implies that $\setU_1 = \{ u^\prime \}$, where \eqref{bl:correspU0U1} holds and $p (u) = p (u^\prime) = 1/2$. For every $R_\setZ \in [ R_\setZ^\ast, \bar p ]$ set $p_2 (u) = h_\textnormal{b}^{-1} ( 1 - R_\setZ / \bar p )$ and
\begin{subequations}
\ba 
\frac{1 - p_2 (u)}{2} = p_{X_2,Y|U} (0,0|u) &= p_{X_2,Y|U} (1,2|u^\prime) \\
\frac{1 - p_2 (u)}{2} = p_{X_2,Y|U} (0,1|u) &= p_{X_2,Y|U} (1,1|u^\prime) \\
p_2 (u) = p_{X_2,Y|U} (1,2|u) &= p_{X_2,Y|U} (0,0|u^\prime).
\ea
\end{subequations}
For this choice we have
\begin{subequations}\label{bl:boundary}
\ba
R_\setY (R_\setZ) &= H (Y|U) = h_\textnormal{b} \bigl( p_2 (u) \bigr) + 1 - p_2 (u) \label{eq:boundaryRY} \\
R_\setZ &= I (U;Z) = \bar p \Bigl( 1 - h_\textnormal{b} \bigl( p_2 (u) \bigr) \Bigr),
\ea
\end{subequations}
Moreover, it holds that
\ba 
H (Y) &= h_\textnormal{b} \biggl( \frac{1-p_2 (u)}{2} \biggr) + \frac{1 + p_2 (u)}{2},
\ea
and from \eqref{eq:boundaryRY} we thus obtain
\ba 
I (U;Y) &= H (Y) - H ( Y|U ) \\
&= h_\textnormal{b} \biggl( \frac{1 - p_2 (u)}{2} \biggr) - h_\textnormal{b} \bigl( p_2 (u) \bigr) + \frac{ 3 p_2 (u) - 1 }{2}.
\ea
Note that for all $p_2 (u) \in \bigl[ 0,1/(1+2^{1/p}) \bigr]$ we have
\ba 
R_\setY < H (Y).
\ea
Moreover, if $p > 1/2$, then we obtain for $p_2 (u) = 0$ that
\ba 
I (U;Z) = \bar p < \frac{1}{2} = I (U;Y).
\ea
Since $I (U;Y)$ and $I (U;Z)$ are continuous in $p_2 (u)$, this implies that there exists an open interval $(0, \kappa) \subset \, \bigr( 0, 1/(1+2^{1/p}) \bigl)$ for which
\ba 
I (U;Y) > I (U;Z), \quad p_2 (u) \in (0, \kappa). \label{eq:condNoSI}
\ea
Indeed, this implies our claim: for $p > 1/2$ there exists a rate-pair $\bigl( R_\setY (R_\setZ), R_\setZ \bigr)$ with $R_\setZ < \bar p$ for which \eqref{bl:capRegNoMessSI} and \eqref{eq:thNoSI} of Proposition~\ref{pr:noSIFBInc} hold for some PMF $p (u,x,y,z)$ of the form \eqref{eq:noMessSIPMF} and for $V = U$.

As we argue next, to conclude the analysis of Example~\ref{ex:addEras} it now suffices to show that on the enhanced BC $\bigl( R_\setY (R_\setZ), R_\setZ \bigr) \in \mathscr C_{\textnormal {enh}} \setminus \partial \mathscr C_{\textnormal {enh}}$ holds for all $R_\setZ \in [ R_\setZ^\ast, \bar p )$. Indeed, we already showed that this is enough to conclude from Proposition~\ref{pr:noSIFBInc} that feedback increases the capacity region without MSI. It is, moreover, enough to conclude from Proposition~\ref{pr:partSIRecYInc} that feedback increases the sum-rate capacity with P-MSI at Receiver~$\setY$, because $\bigl( R_\setY (R_\setZ^\ast), R_\setZ^\ast \bigr)$ is in $\partial \mathscr C$ and satisfies \eqref{bl:capRegNoMessSI} and \eqref{eq:thPartSIRecY} for some PMF of the form \eqref{eq:noMessSIPMF}, where we used \eqref{eq:condPartSIRecY} and \eqref{bl:bdChar} to obtain that it satisfies \eqref{eq:thPartSIRecY}.

To conclude, it now remains to show that on the enhanced BC $\bigl( R_\setY (R_\setZ), R_\setZ \bigr) \in \mathscr C_{\textnormal {enh}} \setminus \partial \mathscr C_{\textnormal {enh}}$ holds for all $R_\setZ \in [ R_\setZ^\ast, \bar p )$. Let $U \sim \ber (1/2)$ be a binary random variable, and set
\begin{subequations}
\ba 
\frac{1-p_2 (u)}{2} &= p_{X_2,Y|U} (0,0|0) = p_{X_2,Y|U} (1,2|1) \\
\frac{1-p_2 (u)}{2} &= p_{X_2,Y|U} (0,1|0) = p_{X_2,Y|U} (1,1|1) \\
\epsilon p_2 (u) &= p_{X_2,Y|U} (1,1|0) = p_{X_2,Y|U} (0,1|1) \\
( 1 - \epsilon ) p_2 (u) &= p_{X_2,Y|U} (1,2|0) = p_{X_2,Y|U} (0,0|1),
\ea
\end{subequations}
where $\epsilon$ and $p_2 (u)$ take values in the set $[0,1]$. If $\epsilon = 0$ and $p_2 (u) \in \bigl[ 0, 1/(1+2^{1/p}) \bigr]$, then, under the above PMF $p (u,x_2,y)$, \eqref{bl:capRegNoMessSI} and \eqref{bl:capRegEnh} both evaluate to \eqref{bl:boundary}. Note that
\ba 
I (U;Z) &= \bar p \Bigl[ 1 - h_\textnormal{b} \bigl( p_2 (u) \bigr) \Bigr]
\ea
does not depend on $\epsilon$, but that $I ( X ; Y, Z | U )$ does. Define
\be
f \bigl( \epsilon, p_2 (u) \bigr) = I ( X ; Y,Z | U ),
\ee
and note that it satisfiess
\ba 
f \bigl( \epsilon, p_2 (u) \bigr) & = I ( X ; Y,Z | U ) \\
&= \bar p \, H ( X_1,X_2 | U ) + p \, H ( Y | U ) \nonumber \\
& = \bar p \Bigl[ h_\textnormal{b} \bigl( p_2 (u) \bigr) + 1 - p_2 (u) + p_2 (u) \binent { \epsilon } \Bigr] \nonumber \\
&\quad - p \Biggl[ \frac{1 - p_2 (u)}{2} \log \biggl( \frac{1 - p_2 (u)}{2} \biggr) \Biggr. \nonumber \\
&\quad + \biggl( \frac{1 - ( 1 - 2 \epsilon ) p_2 (u)}{2} \biggr) \log \biggl( \frac{1 - ( 1 - 2 \epsilon ) p_2 (u)}{2} \biggr) \Biggr. \nonumber \\
&\quad + (1 - \epsilon) p_2 (u) \log \bigl( (1 - \epsilon) p_2 (u) \bigr) \Biggr].
\ea
Note that for all $p_2 (u) \in \bigl( 0, 1/(1+2^{1/p}) \bigr]$ we have
\bas
\frac{\partial f \bigl( \epsilon, p_2 (u) \bigr)}{\partial \epsilon} &= \bar p \, p_2 (u) \log \biggl( \frac{1 - \epsilon}{\epsilon} \biggr) + p \, p_2 (u) \log \biggl( \frac{2 \, \bar \epsilon \, p_2 (u)}{1 - (1 - 2 \epsilon) p_2 (u)} \biggr) \\
&\rightarrow \infty \,\, (\epsilon \downarrow 0).
\eas
In particular, this implies that there exists some $\epsilon > 0$ so that $f \bigl( \epsilon, p_2 (u) \bigr) > f \bigl( 0, p_2 (u) \bigr)$. Since the capacity region of the enhanced BC contains every rate-pair $(\tilde R_\setY, \tilde R_\setZ)$ that for some $\epsilon > 0$ satisfies
\begin{subequations}
\ba 
\tilde R_\setY &= I ( X ; Y,Z | U ) \\
\tilde R_\setZ &= I ( U ; Z ),
\ea
\end{subequations}
this proves our claim that for all $R_\setZ \in [ R_\setZ^\ast, \bar p )$ we have $\bigl( R_\setY (R_\setZ), R_\setZ \bigr) \in \mathscr C_{\textnormal {enh}} \setminus \partial \mathscr C_{\textnormal {enh}}$.

\begin{remark}\label{re:propNoSIFBIncVNeeded}
To prove that feedback can increase the capacity region of the SD-BC without MSI, we used Proposition~\ref{pr:noSIFBInc} with the choice $V = U$. This raises the question whether choosing $V \neq U$ can help. It turns out that it can. To see this extend the SD-BC of Example~\ref{ex:addEras} by a parallel channel of capacity larger than $\log 3$, and assume that Receiver~$\setZ$ additionally observes the output of this parallel channel. (E.g., assume that, in addition to $Z$, Receiver~$\setZ$ noiselessly observes an input $X_3$, which can assume $4$ different values.) For the constructed SD-BC, it is easy to see that a rate-tuple satisfying \eqref{bl:capRegNoMessSI} for some PMF of the form \eqref{eq:noMessSIPMF} is a boundary point of $\mathscr C$ only if the auxiliary random variable $U$ comprises a capacity-achieving input to the parallel channel. (E.g., if Receiver~$\setZ$ observes the pair $( Z,X_3 )$, then $X_3$ must have a uniform prior and be deterministic given $U$.) But this implies that $I (U;Y) \leq \log 3 < I (U;Z)$. Hence, we cannot invoke Proposition~\ref{pr:noSIFBInc} with the choice $V = U$ to show that feedback can increase the capacity region. However, we can invoke Proposition~\ref{pr:noSIFBInc} with the following choice of $V$: choose $V$ to be the random variable that we obtain when we discard the input to the parallel channel from $U$.
\end{remark}

\section{Proof of Theorem~\ref{th:fullMessSIRecZ}}\label{sec:pfFullMessSIRecZ}

To prove the theorem, we show that every rate-tuple in the feedback capacity region of the SD-BC with F-MSI at the stochastic receiver~$\setZ$ satisfies \eqref{bl:capRegFullMessSIRecZCommMess} (with $R = 0$) for some PMF of the form \eqref{eq:fullMessSIRecZCommMess}. Let $\rndQ \sim \unif [ 1 : n ]$ be independent of $( \rndM_\setY, \rndM_\setZ )$, and introduce $(X,Y,Z) = (X_Q,Y_Q,Z_Q)$.

The rate of message $M_\setY$ satisfies
\ba 
R_\setY - \epsilon_n &\stackrel{(a)}\leq \frac{1}{n} I \bigl( M_\setY;Y^n,M_\setZ^{(c)} \bigr) \\
&\stackrel{(b)}\leq H \bigl( Y_Q \bigl| M_\setZ^{(c)} \!,Y^{Q-1},Q \bigl) \\
&\stackrel{(c)}\leq H (Y),
\ea
where $(a)$ follows from Fano's inequality; $(b)$ holds because of the chain-rule, because $M_\setY$ and $M_\setZ^{(c)}$ are independent, and because conditional entropy is nonnegative; and $(c)$ holds because conditioning cannot increase entropy. The rate of message $M_\setZ$ satisfies
\ba 
R_\setZ - \epsilon_n &\stackrel{(a)}\leq \frac{1}{n} I (M_\setZ;Z^n,M_\setY) \\
&\stackrel{(b)}\leq I (M_\setZ;Z_Q|M_\setY,Z^{Q-1},Q) \\
&\stackrel{(c)}\leq I (X;Z),
\ea
where $(a)$ follows form Fano's inequality; $(b)$ holds because of the chain-rule, and because $M_\setY$ and $M_\setZ$ are independent; and $(c)$ holds because conditioning cannot increase entropy, and because $Z$, $X$, and $( M_\setY, M_\setZ, Z^{Q-1}, Q )$ form a Markov chain in that order.

We next establish the sum-rate constraint. To this end we first note that if the probability of a decoding error is small, then---with high probability---$M_\setY$ and $M_\setZ^{(p)}$ are computable from $( Y^n,M_\setZ^{(c)} )$ and $( Z^n,M_\setY )$, respectively, and hence $( M_\setY,M_\setZ^{(p)} )$ is computable from $( Y^n,Z^n,M_\setZ^{(c)} )$. We thus obtain
\ba 
&R_\setY + R_\setZ^{(p)} \!- \epsilon_n \nonumber \\
&\quad\stackrel{(a)}\leq \frac{1}{n} I \bigl( M_\setY,M_\setZ^{(p)} ; Y^n,Z^n,M_\setZ^{(c)} \bigr) \\
&\quad\stackrel{(b)}= I \bigl( M_\setY,M_\setZ^{(p)} ; Y_Q,Z_Q\bigl|M_\setZ^{(c)} \!,Y^{Q-1},Z^{Q-1},Q \bigr. \bigr) \\
&\quad\stackrel{(c)}\leq I (X;Y,Z),
\ea
where $(a)$ follows from Fano's inequality; $(b)$ follows from the chain-rule and the independence of $\bigl( M_\setY,M_\setZ^{(p)} \bigr)$ and $M_\setZ^{(c)}$; and $(c)$ holds because conditioning cannot increase entropy, and because $(Y,Z)$, $X$, and $(M_\setY,M_\setZ,Y^{Q-1},Z^{Q-1},Q)$ form a Markov chain in that order.

\section{Analysis of Example~\ref{ex:fbNeedNotHelp}}\label{sec:pfFbNeedNotHelp}

By Remark~\ref{re:partMessSIRecy} and Corollary~\ref{co:noMessSI}, the no-feedback capacity region without MSI at the stochastic receiver~$\setZ$ is the set of rate-tuples satisfying \eqref{bl:capRegNoMessSI} for some PMF $p (u,x,y,z)$ of the form \eqref{eq:noMessSIPMF}, where we can w.l.g.\ restrict $X$ to be a function of $(Y,U)$. Note that for the SD-BC of Example~\ref{ex:fbNeedNotHelp}
\ba 
I (U;Z) &\stackrel{(a)}= I (U;Z,S) \\
&\stackrel{(b)}= I (U;Z|S) \\
&\stackrel{(c)}= \bar p \, I (U;X) \\
&\stackrel{(d)}= \bar p \, I (U;Y) + \bar p \, I (U;X|Y) \\
&\stackrel{(e)}= \bar p \, I (U;Y) + \bar p \, H (X|Y) \\
&= \bar p \, I (U;Y) + \bar p \, \sum_{y \in \setY} p (y) H (X|Y = y) \\
&\stackrel{(f)}\leq \bar p \, I (U;Y) + \bar p \, \sum_{y \in \setY} p (y) \log |\setX_y|,
\ea
where $(a)$ holds because $S$ is computable from $Z$; $(b)$ follows from the chain-rule and the independence of $S$ and $U$; $(c)$ follows from \eqref{eq:ZExFBUseless}; $(d)$ holds because $Y$ is a function of $X$, and because of the chain-rule; $(e)$ holds because $X$ is a function of $(Y,U)$; and $(f)$ holds because $Y = y$ implies that $X \in \setX_y$, and because the uniform distribution maximizes entropy. Note that, irrespective of $H (Y)$ and $I (U;Y)$, we can achieve $(f)$ with equality. Using that
\ba 
&H (Y|U) + \bar p \, I (U;Y) + \bar p \sum_{y \in \setY} p (y) \log |\setX_y| \nonumber \\
&\quad = H (Y) - p \, I (U;Y) + \bar p \sum_{y \in \setY} p (y) \log |\setX_y|,
\ea
we conclude that, indeed, the capacity region without feedback is the set of rate-tuples $(R_\setY, 0, R_\setZ^{(p)} \!, R_\setZ^{(c)})$ satisfying \eqref{bl:FBDoesNotHelp} for some PMF $p (u,x,y,z)$ of the form \eqref{eq:noMessSIPMF}.

Consider now the case with feedback. As we argue next, the feedback capacity region of any SD-BC without MSI at the stochastic receiver~$\setZ$ ($R_\setY^{(p)} \!= R_\setY$) is contained in the set of all rate-tuples satisfying \eqref{eq:ryCapRegNoMessSI} and \eqref{eq:rzCapRegNoMessSI} (with $R = 0$) as well as
\ba 
R_\setY + R_\setZ &\leq I (X;Y,Z|U) + I (U;Z)
\ea
for some PMF $p (u,x,y,z)$ of the form \eqref{eq:noMessSIPMF}, irrespective of whether or not the deterministic receiver~$\setY$ has MSI ($R_{\setZ}^{(p)} \!\in [0,R_\setZ]$). Indeed, let $Q \sim \unif [1:n]$ be independent of $(M_\setY,M_\setZ)$, and introduce
\begin{subequations}
\ba
U_Q &= (M_\setZ, Y^{Q-1}, Z^{Q-1}), \\
U &= (U_Q, Q), \\
X &= X_Q, \, Y = Y_Q, \, Z = Z_Q.
\ea
\end{subequations}
Note that---also in the presence of feedback---$U$, $X$, and $(Y,Z)$ form a Markov chain in that order, i.e., that their PMF is of the form \eqref{eq:noMessSIPMF}. Using Fano's inequality, it is not hard to show that
\ba
R_\setY - \epsilon_n &\leq \frac{1}{n} I (M_\setY;Y^n,M_\setZ) \leq H (Y) \\
R_\setZ - \epsilon_n &\leq \frac{1}{n} I (M_\setZ;Z^n) \leq I (U;Z) \\
R_\setY + R_\setZ - \epsilon_n &\leq \frac{1}{n} \bigl[ I (M_\setY;Y^n,Z^n,M_\setZ) + I (M_\setZ;Z^n) \bigr] \\
&\leq I (X;Y,Z|U) + I (U;Z),
\ea
which proves the claim.

We are now ready to conclude that feedback cannot increase the capacity region. To this end we note that for the SD-BC of Example~\ref{ex:fbNeedNotHelp}
\ba 
&I (X;Y,Z|U) + I (U;Z) \nonumber \\
&\quad \stackrel{(a)}= H (Y|U) + I (X;Z|Y,U) + I (U;Z) \\
&\quad \stackrel{(b)}= H (Y|U) + I (X;Z,S|Y,U) + I (U;Z,S) \\
&\quad \stackrel{(c)}= H (Y|U) + I (X;Z|Y,U,S) + I (U;Z|S) \\
&\quad \stackrel{(d)}= H (Y|U) + \bar p \, H (X|Y,U) + \bar p \, I (U;X) \\
&\quad \stackrel{(e)}= H (Y|U) + \bar p \, H (X|Y,U) + \bar p \, I (U;Y) + \bar p \, I (U;X|Y) \\
&\quad = H (Y) - p \, I (U;Y) + \bar p \, H (X|Y,U) + \bar p \, I (U;X|Y) \\
&\quad \stackrel{(f)}= H (Y) - p \, I (U;Y) + \bar p \, H (X|Y) \\
&\quad \stackrel{(g)}\leq H (Y) - p \, I (U;Y) + \bar p \sum_{y \in \setY} p (y) \log |\setX_y|,
\ea
where $(a)$ holds because of the chain-rule, and because $Y$ is a function of $X$; $(b)$ holds because $S$ is computable from $Z$; $(c)$ follows from the chain-rule and the independence of $S$ and $(U,X)$; $(d)$ follows from \eqref{eq:ZExFBUseless}; $(e)$ holds because $Y$ is a function of $X$, and because of the chain-rule; $(f)$ follows from the chain-rule; and $(g)$ holds because $Y = y$ implies that $X \in \setX_y$, and because the uniform distribution maximizes entropy.

%%%%%%%%%%%%%%%%%%%%%%%%%%%%%%%%%%%%%%%%%%%%%%%%%%%%%%%%%%%%%%%%%%
%% END OF APPENDIX
%%%%%%%%%%%%%%%%%%%%%%%%%%%%%%%%%%%%%%%%%%%%%%%%%%%%%%%%%%%%%%%%%%
%
\end{appendix}

%%%%%%%%%%%%%%%%%%%%%%%%%%%%%%%%%%%%%%%%%%%%%%%%%%%%%%%%%%%%%%%%%%
%% HEADER BIBLIOGRAPHY
%%%%%%%%%%%%%%%%%%%%%%%%%%%%%%%%%%%%%%%%%%%%%%%%%%%%%%%%%%%%%%%%%%
%
\lhead[\fancyplain{\scshape Appendix}
{\scshape Appendix}]
{\fancyplain{\scshape \leftmark}
  {\scshape \leftmark}}
\rhead[\fancyplain{\scshape \leftmark}
{\scshape \leftmark}]
{\fancyplain{\scshape Appendix}
  {\scshape Appendix}}
%
%%%%%%%%%%%%%%%%%%%%%%%%%%%%%%%%%%%%%%%%%%%%%%%%%%%%%%%%%%%%%%%%%%
%% BIBLIOGRAPHY
%%%%%%%%%%%%%%%%%%%%%%%%%%%%%%%%%%%%%%%%%%%%%%%%%%%%%%%%%%%%%%%%%%
%%
%%%%%%%%%%%%%%%%%%%%%%%%%%%%%%%%%%%%%%%%%%%%%%%%%%%%%%%%%%%%%%%%%%%%%%%%%% 
% \bibliographystyle{IEEEtran}
% \bibliography{defshort1,biblio1}

% Generated by IEEEtran.bst, version: 1.13 (2008/09/30)

%

%
%%%%%%%%%%%%%%%%%%%%%%%%%%%%%%%%%%%%%%%%%%%%%%%%%%%%%%%%%%%%%%%%%%
%% INDEX
%%%%%%%%%%%%%%%%%%%%%%%%%%%%%%%%%%%%%%%%%%%%%%%%%%%%%%%%%%%%%%%%%%
%
%\printindex
%
%%%%%%%%%%%%%%%%%%%%%%%%%%%%%%%%%%%%%%%%%%%%%%%%%%%%%%%%%%%%%%%%%%
%% END FRAME
%%%%%%%%%%%%%%%%%%%%%%%%%%%%%%%%%%%%%%%%%%%%%%%%%%%%%%%%%%%%%%%%%%
%
\end{document}